\documentclass[10pt]{article}
\usepackage{amsmath,hyperref}
\usepackage{amssymb,amsthm}
\usepackage{cite}
\usepackage{amsfonts}
\usepackage{authblk}
\usepackage{graphicx}
\graphicspath{{figs/}}
\usepackage[left=2.5cm,right=2.5cm,top=0.5cm,bottom=1.3cm,includeheadfoot]{geometry}
\usepackage{indentfirst}
\usepackage{color}
\usepackage[shortlabels]{enumitem}
\usepackage[dvipsnames]{xcolor}
\usepackage[font=footnotesize]{caption} %

\hypersetup{
    colorlinks=true,
    linkcolor=blue,
    citecolor=blue,  
    urlcolor=blue,
}

\newcommand{\ham}{\mathcal{H}}
\newcommand{\diff}{\mathcal{D}}

\newcommand{\shift}{{N^x}}
\newcommand{\lapse}{N}
\newcommand{\chif}{G}

\newcommand{\erad}{E^x{}}
\newcommand{\ephi}{E^\varphi{}}
\newcommand{\krad}{K_x}
\newcommand{\kang}{K_\varphi}

\newcommand{\lbar}{\mbox{$\lambda$\hspace*{-7.4pt}\rotatebox[origin=c]{21}{${}^-$}\hspace*{-1.5pt}}{}}
\newcommand{\slbar}{\mbox{\scriptsize{$\lambda$\hspace*{-6pt}\tiny{\rotatebox[origin=c]{21}{$^-$}}\hspace*{-1.5pt}}{}}}

\newcounter{mnotecount}%
\newcommand{\mnote}[1]%
{\protect{\stepcounter{mnotecount}}$^{\mbox{\footnotesize $\bullet$\themnotecount}}$ 
\marginpar{%
\raggedright\tiny\em
$\!\!\!\!\!\!\,\bullet$\themnotecount: #1} }

\def\lproof{{s}}
\def\lmin{\mu}

\def\tt{T}
\def\rr{{r}}

\def\rmax{r_{\infty}}

\def\horizon{a}
\def\bhhor{r_H}

\def\signgeod{\varepsilon}

\newtheorem{theorem}{Theorem}[section]
\newtheorem{lemma}[theorem]{Lemma}
\theoremstyle{remark}
\newtheorem{remark}[theorem]{Remark}

\def\lproof{{s}}

\def\signgeod{\varepsilon}

\def\mofr{{m\big(r(z)\big)}}

\def\manifold{{\mathcal{M}}}

\def\chif{G}
\def\signe{\mathrm{sign}}

\begin{document}

\title{Singularity resolution by holonomy corrections: Spherical charged black holes in cosmological backgrounds}

\author[]{Asier Alonso-Bardaji\footnote{ E-mail address: asier.alonso@ehu.eus}}
\author[]{David Brizuela\footnote{ E-mail address: david.brizuela@ehu.eus}}
\author[]{Raül Vera\footnote{ E-mail address: raul.vera@ehu.eus}}
\affil[]{Department of Physics and EHU Quantum Center, University of the Basque Country UPV/EHU, Barrio Sarriena s/n, 48940 Leioa, Spain}
\date{}
\maketitle
\begin{abstract}
We study spherical charged black holes in the presence of a cosmological
constant with corrections motivated by the theory of loop quantum gravity.
The effective theory is constructed at the Hamiltonian level by introducing
certain correction terms under the condition that the modified
constraints form a closed algebra.
The corresponding metric tensor is then carefully constructed ensuring that the
covariance of the theory is respected, that is, in such a way that
different gauge choices on phase space simply correspond to different
charts of the same spacetime solution.
The resulting geometry is characterized by four parameters:
the three usual ones that appear in the general
relativistic limit (describing the mass, the charge,
and the cosmological constant), as well as a polymerization parameter,
which encodes the quantum-gravity corrections.
Contrary to general relativity, where this family of solutions
is generically singular, in this effective model the presence of
the singularity depends on the values of the parameters.
The specific ranges of values that define the family of singularity-free
spacetimes are explicitly found,
and their global structure is analyzed. In particular,
the mass and the cosmological constant
need to be nonnegative to provide a nonsingular geometry, while
there can only be a bounded, relatively small, amount of charge.
These conditions are suited for
any known spherical astrophysical black hole in the de Sitter
cosmological background, and thus this model provides a globally regular description for them.
\end{abstract}


\section{Introduction}

Effective models are expected to be very useful in extracting
physics from any theory of quantum gravity in a semiclassical regime.
There are different approaches to construct such models,
and their goal is to encode the main effects predicted
by the theory under consideration.
For loop quantum gravity and, more specifically,
for its symmetry-reduced version usually named loop quantum cosmology,
effective models have shown an excellent performance to describe the dynamics
of the corresponding system as compared to the exact quantum dynamics
\cite{Ashtekar:2011ni,Agullo:2016tjh,Rovelli:2013zaa}.
In this context, such effective models are usually constructed
by including the so-called inverse-triad and holonomy corrections
in the Hamiltonian constraint of general relativity (GR). Even if these two types
of corrections are motivated by the quantization performed in the
full theory, in the homogeneous models it has been checked
that holonomy corrections by themselves are able to provide
the resolution of the singularity,
and thus effective theories are usually reduced to describe
this type of corrections.

Hence,
a natural step is to extend holonomy corrections to spherically 
symmetric cases, in order to determine whether singularity resolution is a 
strong prediction of the theory, or simply a by-product produced by an
excessive symmetry assumption. In particular, there is a rich literature about
the study of effective models with holonomy corrections for spherical
vacuum (Schwarzschild black hole)
\cite{BenAchour:2018khr,Gambini:2020nsf,Kelly:2020uwj,Ashtekar:2018lag,Ashtekar:2018cay,Ashtekar:2020ckv,Bodendorfer:2019cyv,Bodendorfer:2019nvy,Modesto:2008im} and dynamical scenarios for collapsing fields \cite{Hossenfelder:2009fc,Kelly:2020lec,Gambini:2021uzf,Husain:2021ojz,Husain:2022gwp}, which have led to a variety of predictions.
However, intermediate cases, with matter that lacks local degrees of freedom \cite{Tibrewala:2012xb,Gambini:2014qta}, are often overlooked. In this work, we extend the vacuum study presented in Refs. \cite{Alonso-Bardaji:2021yls,Alonso-Bardaji:2022ear} to incorporate a Maxwell field and a cosmological constant.

It is worth noting that the nonhomogeneity of this new scenario introduces significant conceptual issues related to the covariance of effective models \cite{Alonso-Bardaji:2020rxb,Alonso-Bardaji:2021tvy,Bojowald:2015zha,Bojowald:2018xxu}. By following the approach in Ref. \cite{Bojowald:2015zha}, we were able to find a modified Hamiltonian \cite{Alonso-Bardaji:2021tvy} that has an unambiguous geometric description \cite{Alonso-Bardaji:2021yls,Alonso-Bardaji:2022ear}. The key feature of our model is that, unlike previous attempts in the literature, it provides a covariant framework for studying effective modifications. Holonomy corrections are introduced at the Hamiltonian level to generate a first-class algebra, with the structure function of that algebra having the correct transformation properties for the model to be embeddable in a four-dimensional spacetime manifold \cite{Teitelboim73,Pons:1996av}. As a result, different gauge choices in phase space simply correspond
to different coordinate systems in the same spacetime.

In contrast to the vacuum case, we find that these corrections are not sufficient
to generically resolve the singularity for all values of the parameters.
We study in detail the ranges of parameters that lead to nonsingular
spacetimes and analyze their global structure.
In particular, the model requires a nonnegative value of the mass and cosmological constant,
while the charge needs to be below a certain maximum threshold.
Interestingly, the set of parameters that would describe a realistic
astrophysical black hole lay in this category.

The article is organized as follows. First, we derive the effective model from the GR Hamiltonian in Sec.~\ref{sec.gr}. Then, in Sec.~\ref{sec.spacetime}, we construct the corresponding metric and provide the solution for different gauge choices in phase space, which are related by coordinate transformations. Next, we study the global behavior of the solution, which involves analyzing a fourth-order polynomial with four free parameters. In Sec.~\ref{sec_study_singularity}, we focus on identifying the
ranges of the parameters that lead to nonsingular solutions. In Sec.~\ref{sec:global}, we
provide the main elements to analyze the global structure of such solutions,
and to construct their conformal diagrams. The main results of the paper
are then discussed and summarized in Sec. \ref{sec_conclusions}.
In Appendix.~\ref{sec_app_conformal} the Penrose diagrams for the complete
family of nonsingular spacetimes are displayed. Finally,
Appendixes.~\ref{app.covariance}--\ref{app.tortoise} provide some technical
details to clarify and extend certain points presented in the main text.

\section{Canonical formulation of the model}\label{sec.gr}

This section is divided into two subsections. In Sec. \ref{sec_classical}
we present the Hamiltonian for a spherically symmetric spacetime
coupled to an electromagnetic field with a cosmological constant
in GR. In Sec. \ref{sec_effective}, we construct an effective
theory by performing a canonical transformation followed
by a linear combination of the constraints. 
\subsection{The classical framework}\label{sec_classical}

In terms of the Ashtekar-Barbero variables, the diffeomorphism and Hamiltonian
constraints of GR are given in smeared form  by
$D[f]:=\int f\diff dx$ and $\widetilde{H}[f]:=\int f\widetilde{\ham} dx$,
with
\begin{subequations}
\begin{align}
\!\diff &=
     -({\widetilde{E}^x})'{\widetilde{K}_x} +{\widetilde{E}^\varphi} ({\widetilde{K}_\varphi})' +\diff_m,\\
     \!\widetilde{\mathcal{H}}
       &= -\frac{{\widetilde{E}^\varphi}}{2\sqrt{{\widetilde{E}^x}}}\left(1+\widetilde{K}_\varphi^2\right)  -2\sqrt{{\widetilde{E}^x}}{\widetilde{K}_x}{\widetilde{K}_\varphi}
       +\frac{1}{2}\Bigg(\frac{\widetilde{E}^x{}'}{2\widetilde{E}^\varphi}\left(\sqrt{\widetilde{E}^x}\right)'
              +\sqrt{\widetilde{E}^x}\left(\frac{\widetilde{E}^x{}'}{\widetilde{E}^\varphi}\right)'\Bigg) +\ham_m,
\end{align}
\end{subequations}
respectively.
In these expressions the prime stands for a derivative with respect to
$x$, $\diff_m$ and $\ham_m$ denote the matter contributions, and the
variable $\widetilde E^x$ is assumed to be nonnegative
$\widetilde E^x\geq 0$. The symplectic structure is canonical,
$$\{\widetilde{K}_x(x_1),\widetilde{E}^x(x_2)\}=\{\widetilde{K}_\varphi(x_1),\widetilde{E}^\varphi(x_2)\}=\delta(x_1-x_2),$$
where we use the standard notation $\{\cdot,\cdot\}$ for the Poisson brackets.

This paper considers gravity weakly coupled to two simple matter
types: a cosmological constant and a Maxwell field.  The
cosmological constant $\Lambda$ can be understood as a
nondynamical scalar field, and it only contributes with a term of the
form ${\cal H}_\Lambda=\sqrt{\widetilde{E}^x}\widetilde{E}^\varphi\Lambda/2$ to
the matter Hamiltonian constraint.
The electromagnetic field is described in terms of the vector potential
 $A_\mu$. Because of the spherical symmetry, it has only two nontrivial
components: $A_0$ and $A_x$. The momentum $p^0$ conjugate
to $A_0$ vanishes. Therefore $p^0=0$ is a primary constraint and thus
$A_0$ is nondynamical. The component $A_x$ and its conjugate momentum,
denoted as $p^x$, obey %
\begin{equation}
 \{A_x(x_1),p^x(x_2)\}=\delta(x_1-x_2).
\end{equation}
The contribution of the electromagnetic field to the diffeomorphism and
Hamiltonian constraints are respectively given by
(see, e.g., Refs. \cite{Tibrewala:2012xb,Gambini:2014qta})
\begin{align}
 {\cal D}_{\rm em} &= A_x (p^{x})^\prime,\\
 {\cal H}_{\rm em} &= \frac{E^{\varphi}(p^x)^2}{2(\widetilde{E}^x)^{3/2}}-(p^x)'A_0.
 \end{align}
Therefore, the matter Hamiltonian for our system reads
${\cal H}_{m}={\cal H}_{\Lambda}+{\cal H}_{\rm em}$, while the matter part of the diffeomorphism constraint is just given by the Maxwell field ${\cal D}_m={\cal D}_{\rm em}$.
It is easy to see that the conservation of the primary constraint,
$0=\dot{p}^0=\{p^0,\widetilde H[N]+D[N^x]\}$, leads to the condition
\begin{equation}
{\cal G}:= (p^x)^\prime=0,
\end{equation}
which is the electromagnetic Gauss law and defines the first-class constraint
$G_{\rm em}[f]:=\int f {\cal G} dx$.
There are no further constraints in the system,
and the total Hamiltonian is thus defined as the linear combination
$\widetilde{H}_{T}:=\widetilde{H}[\lapse]+D[\shift]+G_{\rm em}[\beta+A_0 \lapse - A_x \shift]$,
where the Lagrange multipliers $\lapse$ and $\shift$
  correspond to the lapse and shift of the usual 3+1 decomposition in GR,
  and the smearing function in the Gauss constraint has been conveniently chosen.

Since there are three couples of conjugate
variables and three first-class constraints, there are no
propagating degrees of freedom in this model.
In fact, the Gauss constraint is rather trivial, and it is possible to fix
the gauge for the matter variables without any loss of generality.
More precisely,
the equations of motion for the couple $(A_x,p^x)$ read
\begin{align}\label{dotAx}
  \dot{A}_x=& \frac{N \widetilde{E}^\varphi p^x}{(\widetilde{E}^x)^{3/2}} %
              -\beta^\prime,\\\label{dotpx}
 \dot{p}^x=&N^x p^x{}^\prime.
\end{align}
Since the time derivative of $p^x$ is proportional to the Gauss constraint,
the second equation implies that $p^x$ is conserved.
This observable defines the constant charge of the spacetime,
which we denote by $Q:=p^x$. Hence,
at this point
one can partially fix the gauge by strongly enforcing the Gauss constraint
and by choosing any gauge-fixing condition of the form
$A_x=\Phi(p^x,\widetilde{E}^x,\widetilde{E}^\varphi,\widetilde{K}_x,\widetilde{K}_\varphi)$. The conservation of this condition, $\dot{A}_x=\dot{\Phi}$,
will then provide the form of the Lagrange multiplier $\beta$
through \eqref{dotAx}. However,
since neither $A_x$ nor $\beta$ appear in other equations of motion
besides \eqref{dotAx}, their
specific form will not modify the evolution of the geometric
variables $(\widetilde{E}^x,\widetilde{E}^\varphi,\widetilde{K}_x,\widetilde{K}_\varphi)$, which are thus insensitive to the chosen gauge $\Phi$.

In this way, for the spherical Einstein-Maxwell
model with a cosmological constant $\Lambda$,
one gets an exactly vanishing matter diffeomorphism constraint
${\cal D}_m=0$, while the matter Hamiltonian takes the form
\begin{equation}
{\cal H}_m=
\frac{1}{2}\sqrt{\widetilde{E}^x}\widetilde{E}^\varphi
\left(\Lambda+\left(\frac{Q}{\widetilde{E}^x}\right)^2\right),
\end{equation}
in terms of
the two constant parameters $Q$ and $\Lambda$. The constraints obey the usual
hypersurface deformation algebra,
\begin{subequations}\label{hdaclass}
  \begin{align}
    \big\{D[f_1],D[f_2]\big\}&=D\big[f_1f_2'-f_1'f_2\big],\\
    \big\{D[f_1],\widetilde{H}[f_2]\big\}&=\widetilde{H}\big[f_1f_2'\big],\\
    \label{hhbrackets}
    \big\{\widetilde{H}[f_1],\widetilde{H}[f_2]\big\}&=D\left[
    {\widetilde{E}^x(\widetilde{E}^\varphi)^{-2}}(f_1f_2'-f_1'f_2)\right],
  \end{align}
\end{subequations}
and the total Hamiltonian simplifies to $\widetilde{H}_T=\widetilde{H}[N]+D[N^x]$.
The equations of motion for the four remaining variables
$(\widetilde{E}^x,\widetilde{E}^\varphi,\widetilde{K}_x,\widetilde{K}_\varphi)$
can be readily obtained by computing their Poisson brackets with this Hamiltonian.

\subsection{The effective model}\label{sec_effective}

In Ref. \cite{Alonso-Bardaji:2021tvy}, an effective Hamiltonian for spherically symmetric gravity
coupled to a scalar matter field was presented, which included holonomy corrections
respecting the first-class nature of the algebra. This effective constraint
was shown to be related to the GR Hamiltonian
through a canonical transformation followed by a linear combination of the constraints. 
Here we will apply the
same method to construct an effective Hamiltonian for the
Einstein-Maxwell-de Sitter model. More precisely, we first
perform the following transformation for the geometric degrees of freedom,
\begin{align}\label{cantransf}
     \widetilde{E}^x= \erad\,,\quad
     \widetilde{K}_x= \krad\,,\quad
     \widetilde{E}^\varphi= \frac{\ephi}{\cos(\lambda \kang)}\,,\quad \widetilde{K}_\varphi= \frac{\sin(\lambda \kang)}{\lambda}\,.
\end{align}
This transformation is canonical and thus leaves the symplectic structure invariant
$\{\krad(x_1),\erad(x_2)\}=\{\kang(x_1),\ephi(x_2)\}=\delta(x_1-x_2)$. We then implement a regularization
multiplying the Hamiltonian constraint by $\cos(\lambda K_\varphi)$,
which removes the inverse of this function from the Hamiltonian.
Despite this regularization is a choice, and introduces changes in the dynamics in general,
it is not entirely arbitrary;
the structure functions of the deformed constraint algebra
need to have the correct gauge transformation properties
in order to be able to provide a covariant spacetime representation of the model.
Below we will explain this in detail.
In addition, this regularization still recovers
GR in the limit $\lambda\rightarrow 0$.

The so-constructed Hamiltonian constraint does not obey, along with $\mathcal{D}$, the canonical
form of the algebra. Therefore,
we define the linear combination
\begin{align}
 {\cal H}:=&\,\left(\widetilde{\ham}%
 +\lambda\, {\sin(\lambda \kang)} \frac{\sqrt{\erad}\erad'}{2\ephi^2}{\cal D}\right)
  \frac{\cos(\lambda \kang)}{\sqrt{1+\lambda^2}},
   \label{H_normal}
\end{align}
so that 
the two constraints of the modified model are explicitly given by
\begin{subequations}\label{polham}
\begin{align}
\diff =& -\erad'\krad+\ephi\kang'%
,\label{D}\\
\ham =&  \frac{1}{\sqrt{1+\lambda^2}}\Bigg[-\frac{{\ephi}}{2\sqrt{{\erad}}}\left(1+\frac{\sin^2{{(\lambda {\kang})}}}{{{\lambda^2}}}\right)  -\sqrt{{\erad}}{\krad}\frac{\sin{(2\lambda {\kang})}}{\lambda}\left(1+\left(\frac{\lambda {\erad}'}{2{\ephi}}\right)^{\!2}\right) \nonumber\\
       &\quad+\frac{\cos^2{(\lambda {\kang})}}{2}\Bigg(\frac{\erad'}{2\ephi}\left(\sqrt{\erad}\right)'
              +\sqrt{\erad}\left(\frac{\erad'}{\ephi}\right)'\Bigg)
+ \frac{1}{2}\sqrt{E^x}E^\varphi
\left(\Lambda+\left(\frac{Q}{E^x}\right)^2\right)\Bigg],\label{H}
\end{align}
\end{subequations}
which satisfy the canonical Poisson algebra,
\begin{subequations}\label{algebra}
\begin{align}
\{D[f_1], D[f_2] \} &= D[f_1f_2'-f_1'f_2],\\
  \{D[f_1], {H}[f_2] \} &= {H}[f_1f_2'],\\
   \{{H}[f_1], {H}[f_2] \} &= D[F(f_1 f_2'-f_2 f_1')]
\end{align}
\end{subequations}
with the structure function
\begin{align}\label{Fdef}
F:=  
    \frac{\cos^2(\lambda {\kang})}{1+\lambda^2}\left(1+\Big(\frac{\lambda {\erad}'}{2{{\ephi}}}\Big)^{2}\right)\frac{\erad}{(E^{\varphi})^2}.
\end{align}
It can be easily checked %
that the constraints and their Poisson
algebra reproduce the corresponding structures of GR in the limit $\lambda\rightarrow 0$.
Therefore, in this context, $\lambda$ is interpreted to be a small positive
parameter that encodes the quantum effects and will modify the classical dynamics. 
However, because of the way in which it appears in the expressions, it turns out to be convenient
to define instead the parameter
\begin{align}
\lbar:=\frac{\lambda^2}{1+\lambda^2},
\end{align}
which %
takes values %
in the range $(0,1)$,
and the limit $\lbar\to 0$ corresponds to GR.
Let us also define %
\begin{align}\label{eq.masspol}
  m:=\frac{\sqrt{\erad}}{2}\left(1+\frac{\sin^2(\lambda {{\kang}})}{\lambda^2}-\left(\frac{{\erad}'}{2{{\ephi}}}\right)^{2}\cos^2(\lambda {{\kang}})\right) ,
\end{align}
in terms of which the structure function \eqref{Fdef} can be reexpressed as
\begin{equation}\label{Fdefm}
F =\left(1-\frac{2\lbar m}{\sqrt{\erad}}\right)\frac{\erad}{(E^{\varphi})^2}.
\end{equation}
Since $\erad\geq 0$, it is straightforward to see from the definition \eqref{Fdef}
that the structure function $F$ 
is nonnegative and vanishes only when $\cos^2(\lambda\kang)$
or $\erad$ do.

Let us note that formally \eqref{Fdef} and \eqref{Fdefm} are
the same expressions for the structure function $F$ as in the vacuum case (see Eqs.~(8) and (10) in Ref. \cite{Alonso-Bardaji:2022ear}).
In vacuum  $m$ is a Dirac observable (thus constant),
and therefore $\sqrt{\erad}$ is globally bounded from below by $2\lbar m$,
which is positive if $m>0$.
In that case the points with $E^x=0$, which correspond
to the singularity in GR, are excluded from the dynamics,
and the singularity is thus resolved.
In the present model, however, $m$ is point dependent, and the analysis
is not so straightforward.
In particular, to check whether points
with $\erad=0$ are excluded or not from the dynamics,
we need to analyze the existence of positive roots
of the structure function, that is, whether $\sqrt{\erad}=2\lbar m$
holds for some positive value of $\sqrt{\erad}$.
As we will find below, this will strongly depend on the specific values of
the different parameters of the model.

\section{The covariant spacetime representation}\label{sec.spacetime}

In this section we provide the covariant spacetime representation of our model.
More precisely, in Sec. \ref{sec_metric_tensor}, we construct the metric tensor in terms of the phase-space variables and show how different gauge choices provide different coordinate systems of the same spacetime.
In Sec. \ref{sec_charts}, we introduce three different coordinate
systems, which will be useful for later analysis, and a coordinate system
that describes the near-horizon geometries. Section \ref{sec_curvature} presents the curvature invariants of
the spacetime and the location of potential singularities.

\subsection{The metric tensor}\label{sec_metric_tensor}

To provide a meaningful spacetime representation of the model, we need
to construct the corresponding \emph{covariant} metric.
Let us stress that by \emph{covariant} we mean that %
the gauge transformations on the phase space
correspond to infinitesimal coordinate transformations in the spacetime manifold.
More precisely, the gauge transformation of any phase-space function $f$ is
generated by the first-class constraints as $\delta_\epsilon f= \{f,H[\epsilon^0]+D^x[\epsilon^x]\}$,
with gauge parameters $\epsilon^0$ and $\epsilon^x$.
This should coincide with the transformation followed by $f$,
as a function on the manifold, under
an infinitesimal coordinate transformation
$(t,x)\rightarrow (t+\xi^t,x+\xi^x)$. The two couples $(\xi^t,\xi^x)$
and $(\epsilon^0,\epsilon^x)$ are components of the same vector in different
basis, and are related by $\epsilon^0=N\xi^t$ and $\epsilon^x=\xi^x+\xi^t N^x$.

In particular, in Ref. \cite{Alonso-Bardaji:2022ear}, we explicitly showed that the metric
 \begin{align*}%
     ds^2=-N^2dt^2 + q_{xx}\big(dx+N^xdt\big)^2+\erad d\Omega^2
 \end{align*}
 with $q_{xx}=1/F$, and where $d\Omega^2$ is the metric of the unit sphere,
 obeys the correct transformation properties as long as
 the gauge transformation of $1/F$ coincides with the change of $q_{xx}$
 under a coordinate transformation.
 On the one hand, it is easy to check that, under a general coordinate transformation, the component $q_{xx}$ changes as
 \begin{equation}\label{coordtransf}
  \delta q_{xx} = \dot{q}_{xx} \xi^t+q_{xx}' \xi^x+2q_{xx} (N^x  \xi^t{}'+\xi^{x\prime}).
 \end{equation}
On the other hand, making use of the equations of motion for the
present model, the gauge transformation of $1/F$ can be written as
(see Appendix~\ref{app.covariance} for the details)
\begin{equation}\label{eq.gaugetransf}
\delta_\epsilon (1/F)=\{1/F,H[\epsilon^0]+D[\epsilon^x]\}=
\frac{\epsilon^0}{N}(1/F){\dot{}}\, + \left(\epsilon^x-\frac{N^x}{N}\epsilon^0\right) (1/F)'  +2\left(\epsilon^{x\prime}-\frac{\epsilon^0}{N} N^x{}'\right)(1/F).
\end{equation}
Using the relations
$\epsilon^0=N\xi^t$ and $\epsilon^x=\xi^x+\xi^t N^x$ commented above, it is
straightforward to check that the above two transformations coincide.
Therefore, we conclude that the metric tensor
  \begin{align}\label{eq.metric}
     ds^2=-N^2dt^2 +\left(1-\frac{2\lbar m}{\sqrt{\erad}}\right)^{-1}\frac{(E^{\varphi})^2}{\erad} \big(dx+N^xdt\big)^2+\erad d\Omega^2,
  \end{align}
  with $m$ as defined in \eqref{eq.masspol}, provides a covariant representation of our model in the spacetime.
  The solution
of the equations of motion for a given choice of gauge will simply provide the expression
of the metric tensor in certain coordinates.

To follow the standard notation,
since the scalar $\erad$ as a function on the manifold
is the square of the area-radius function,
we will use $\sqrt{\erad}=:r$ when convenient.
By construction, the area-radius function $r$ can take values on
the positive real line. However,
contrary to GR, the problem under consideration will, in general, restrict the possible
values that $r$ can attain. As already
commented above,
the fundamental reason for the existence
of such ranges is that the structure function $F$
is nonnegative, and that forbids ranges of the function $r$
for which $F$ would formally be negative.
In terms of each spacetime solution, this fact will appear as
minimum or maximum attainable values for the scalar $r$,
to which we will refer as ``critical values'' in short.
The function $m$ will turn out to be a function of $r$ only [see Eq.~\eqref{eq.mass} below], and
the zeros of the structure function, $r=2\lbar m$, will determine at most
three such critical values, which will be named $R$, $r_0$, and $\rmax$, respectively.

  As it will be explicitly shown below, the spacetime under consideration
  will be composed by two kinds of nonoverlapping regions,
  defined for given intervals of $r$ in terms of the sign of the function
  $\chif:=1-2m/r$. On the one hand,
  there will be static regions, similar to the exterior region of Schwarzschild,
  where $\chif>0$ and $r$ can be chosen as a spatial coordinate.
  On the other hand, there will be homogeneous\footnote{We follow the usual convention
      and refer to ``homogeneous'' regions as those
      that admit a foliation by homogeneous spacelike leaves. More precisely, in this
    paper all homogeneous regions will be of Kantowski-Sachs type.} regions, like the
  interior region of  Schwarzschild, where $\chif<0$ and
  hypersurfaces of constant $r$ are timelike.
  As in GR, the boundary between these two kinds of regions will define a horizon at points where $\chif=0$. Depending on the specific values of the mass, charge, and cosmological constant,
  in a given spacetime there might appear several
  horizons, corresponding
  to the usual inner (Cauchy) and outer (black-hole) horizons of the Reissner-Nordstr\"om
  black hole (located at $r=r_I$ and $r=r_H$, respectively) and the cosmological de Sitter horizon (located at $r=r_C$), with $0<r_I\leq r_H\leq r_C<\infty$.
  In GR, the usual structure of the regions is described in terms of ranges of $r$ bounded by the horizons:
there are two static regions in the ranges $(0,r_I)$ and $(r_H,r_C)$, while there are
two homogeneous regions in the ranges $(r_I,r_H)$ and $(r_C,\infty)$.
The critical values defined in the present model
must be located inside a homogeneous region because $\lbar<1$, and thus $2\lbar m< 2 m$.  
Therefore, if the three critical values $R$, $r_0$, and $\rmax$ exist, they will generically split the above structure into two spacetimes with ranges
  $0<r_I<R$ and $r_0<r_H\leq r_C<\rmax$ (see Fig.~\ref{fig_schemer}).
  Depending on the number of horizons and critical values, several different possibilities will arise. We provide that study in detail in Sec.~\ref{sec_study_singularity}.
  
\begin{figure}
 \centering
    \includegraphics[width=0.7\textwidth]{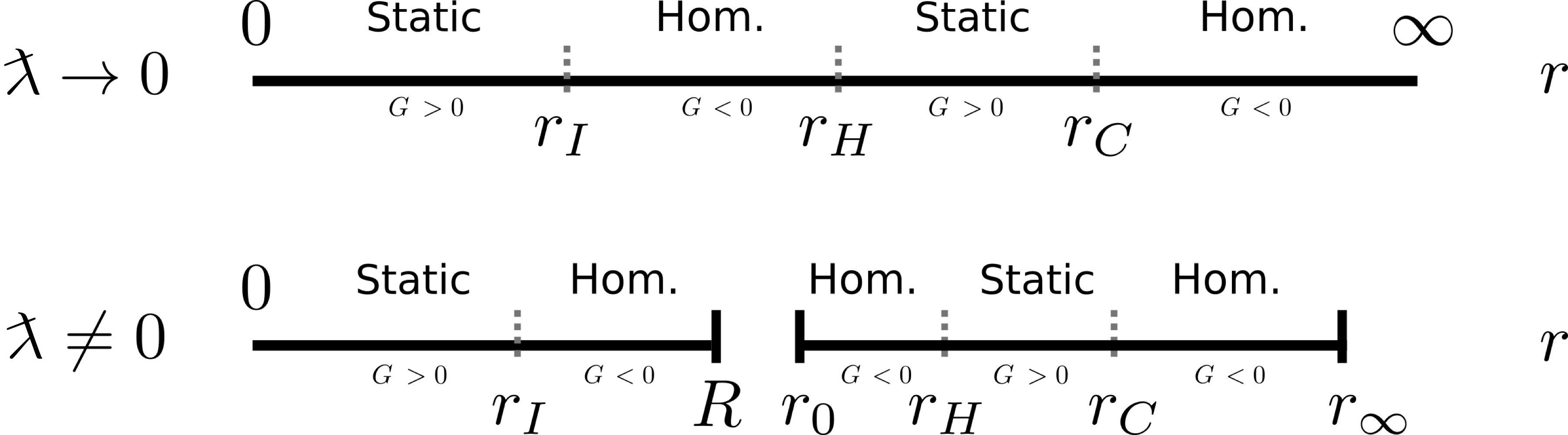}
    \caption{Scheme of the location of the different horizons (if they exist)
      in $r$ in the GR case ($\lbar\to 0$), separating the static and homogeneous
      regions of the spacetime. In our case ($\lbar\neq 0$) the range
      of $r$ in the solution is cut into two intervals, thus providing
      two spacetime solutions for the same values of the parameters:
      one in which the function $r$ is bounded from above at $R$,
      and another in which it is bounded from below at $r_0$ (and possibly
      from above at $\rmax$). All the critical values that bound $r$ appear in the homogeneous regions.}
    \label{fig_schemer}
\end{figure}

\subsection{Different gauges and corresponding charts}\label{sec_charts}
  
In this subsection we introduce four gauge choices.
The first one (in Sec. \ref{sec_static}) will
provide the coordinate system valid in the static regions,
while the second one (in Sec. \ref{sec_homogeneous})
will correspond to homogeneous regions. To show that
the static and homogeneous regions are part of the same spacetime
that are separated by a horizon, in Sec. \ref{sec_covering},
we present a third choice of gauge leading to a domain that
includes all the horizons and thus overlaps with
both the homogeneous and the static regions.
The fourth gauge (in Sec. \ref{sec_nearhorizon}) will describe the limiting
case of the near-horizon geometries.\footnote{As it will be explained below,
these geometries have a fixed value of the area-radius function,
which is given by the position of degenerate horizons, and 
thus they only appear for the values of parameters
that lead to a double degeneration of the form $r_I=r_H$ or
$r_H=r_C$, or to the triple degeneration $r_I=r_H=r_C$.
In the GR limit these correspond to the Bertotti-Robinson
geometry (for which $\Lambda=0, Q\neq 0$, and the degeneration is given
by $r_I=r_H$), and to the Nariai geometry
(for which $\Lambda> 0, Q= 0$, and the degeneration is given
by $r_H=r_C$). The triple degeneration $r_I=r_H=r_C$ is sometimes named
the ultra-extreme case (see, e.g., Ref. \cite{Brill1994}).
}

Different charts
of the spacetime solution will be composed by two coordinates
on the spheres (that we will not specify) plus two
coordinates on the Lorentzian space, orthogonal to the spheres,
corresponding to the pair $(t,x)$.
For each choice of gauge in phase space, which will provide
a different chart, we will conveniently
rename $(t,x)$, now as functions on the manifold.

\subsubsection{Static regions}\label{sec_static}

For the first gauge we choose $\dot \erad=0$,
with $\erad'$ not vanishing everywhere,
and impose $\sin(\lambda K_\varphi)=0$.
As shown in Appendix~\ref{app.static}, the solution of the equations of motion then depends on
  an arbitrary function $r(x):=\sqrt{\erad}(x)$. As a result,
  the gauge is completely fixed by prescribing $r(x)$.
  In Appendix~\ref{app.static} we take a convenient
choice of $r(x)$, for which
the corresponding solution
leads to the diagonal form of the metric
\begin{align}\label{eq.metricstatic}
{ds}^2 &=-\left(1-\frac{2\mofr}{{r(z)}}\right)d\tt^2 +
  \left(1-\frac{2\mofr}{{r(z)}}\right)^{-1}{d{z}^2}
  +{r(z)}^2{d\Omega}^2,
\end{align}
where the pair $(t,x)$ has been renamed as $(\tt,z)$.
In this gauge
the function $m$ \eqref{eq.masspol} takes the form
\begin{align}\label{eq.mass}
  m(r)= M-\frac{Q^2}{2 r}+\frac{\Lambda}{6} r^3,
\end{align}
where $M$ is an integration constant,
and the scalar $r(z)$
is implicitly defined through
\begin{align}\label{eq.rp}
    \left(\frac{dr(z)}{dz}\right)^2 
    = 1-\frac{2\lbar \mofr}{r(z)}.
\end{align}

The domain of these coordinates is given
by $\tt\in\mathbb{R}$, %
while $z$ is limited by $\rr(z)>0$ (by construction)
and the condition $2 \mofr<\rr(z)$. 
Given the values of the parameters $M$, $\Lambda$, $Q$, and $\lbar$,
the interval (or intervals) of $z$  will be thus determined by the zeros
of $r(z)=2 m(r(z))$, %
which will signal the presence of a horizon,
and the zeros of $r(z)=0$.
Observe that $z$ is fixed
up to an arbitrary additive constant, and that
constant can be chosen differently on each interval (if there is more than one).
We will obtain the ranges of $z$ %
and fix that constant conveniently
when we study the global structure
of the solution later, in Sec.~\ref{sec_study_singularity}.

Since $\lbar <1$, once $2 \mofr<\rr(z)$ holds,
we have $2\lbar \mofr<r(z)$.
Therefore the right-hand side of Eq.~\eqref{eq.rp} does not reach zero
at any point of this domain.
As a result, $r(z)$ is a monotonic function of $z$
on the static regions, and therefore $r$
can be taken to be the spatial coordinate
instead of $z$. The metric in the coordinates $(\tt,r)$
is given by \eqref{eq.metapp1} in  Appendix~\ref{app.static},
where we also show
the corresponding choice of gauge in phase space.
The change of coordinates is made explicit in Appendix~\ref{sec:coord_transf}.

\subsubsection{Homogeneous regions}\label{sec_homogeneous}

Next, we consider %
$\erad'=0$ and $\ephi'=0$ with a nonconstant $\erad$.
As shown in Appendix~\ref{app.homogeneous} this sets all spatial derivatives to zero and, moreover, we can choose $\shift=0$ without loss of generality.
As in the previous case
the gauge is conveniently fixed completely by an explicit choice of
$r(t):=\sqrt{\erad}(t)$, which is again given implicitly by \eqref{eq.rp}.
The equations of motion for the remaining variables,
after renaming the pair $(t,x)$ as $(z,\tt)$ (mind the order)
as functions on the manifold,
leads to the diagonal form of the metric
\begin{align}\label{eq.metrichomogeneous}
{ds}^2 &=- 
  \left(\frac{2\mofr}{r(z)}-1\right)^{-1}{dz^2}+\left(\frac{2\mofr}{r(z)}-1\right){d\tt}^2
  +r(z)^2{d\Omega}^2,
\end{align}
where $m(r)$ is given by \eqref{eq.mass}, and the function $r(z)$ satisfies \eqref{eq.rp}.
Note that, for notational convenience, here we are using the same names
for the two coordinates as in the static region above. This is
because below we will provide a third chart that covers
these two regions and, in particular, the coordinate $z$
can be chosen to be the same (restricted to the corresponding
domain) for the three charts. In the GR limit $(\lbar=0)$, it is
straightforward to see that $r=z$ and thus the line elements
\eqref{eq.metricstatic} and \eqref{eq.metrichomogeneous}
correspond to the metric in Schwarzschild coordinates
in their corresponding domains.

The range of the coordinates $(z,T)$ is given by $\tt\in\mathbb{R}$
  while $z$ is constrained by the conditions $r(z)>0$ and $r(z)<2m(r(z))$.
As in the static region, given the values of the parameters $M$, $\Lambda$, $Q$, and $\lbar$,
  the interval (or intervals) of $z$  will be thus determined by the zeros of $r(z)=2 m(r(z))$ and $r(z)=0$.
  The coordinate $z$ can be freely shifted
  by a constant (on each interval, if there is more than one).

  Contrary to the static region, the right-hand side of
  Eq.~\eqref{eq.rp} can vanish in this case, and thus $r(z)$,
  as a function of $z$, may have turning points inside the homogeneous
  region. As long as $r$ is nonmonotonic, choosing $(T,r)$ as coordinates would only partially cover the homogeneous region, since
  the turning points, where $r(z)=2\lbar m(r(z))$, would be excluded.

\subsubsection{The covering domain}\label{sec_covering}

In order to look for the global structure
of the solution, %
we perform another more convenient choice of gauge.
This gauge will provide a coordinate system that overlaps
static and homogeneous regions, as those presented above,
including also the horizons. 
    
We now begin fixing the gauge with the choice
$\dot{\erad}=0$, $\dot\ephi=0$, and assume that $\erad$ is not constant.
It is then shown in Appendix~\ref{app.covering}
that the solution of the equations
of motion depends on a pair of free functions
$r(x):=\sqrt{\erad}(x)$ and $s(x):=\ephi(x)$.
Using then the convenient choice
$s=(r^2-2\lbar r m(r))^{1/2}$, with $m(r)$ given as in \eqref{eq.mass},
and requiring $r(x)$ to satisfy \eqref{eq.rp} (replacing $z$ by $x$), the metric reads
\begin{align}
    ds^2
  &=-\bigg(1-\frac{{2\mofr}}{r(z)}\bigg) {d\tau}^2 %
  +2\sqrt{\frac{{2\mofr}}{r(z)}}{d\tau}{dz} + {{dz}^2}+r(z)^2 d\Omega^2,
  \label{metric:tau_x}
\end{align}
where $(t,x)$ has been renamed as the pair of real functions $(\tau,z)$ on the manifold.
The range of these coordinates is given by $\tau\in \mathbb{R}$
while $z$ is (or may be) restricted by the conditions $r(z)>0$ and
\begin{align}
  \mofr\geq 0.\label{cond_m}
\end{align}

Observe first that the condition $r(z)=2 m(r(z))$, which bounds the
domains of the static and homogeneous regions and defines the horizons,
is absent here.
On the other hand, the condition \eqref{cond_m}, which is inherent to this chart, may indeed produce a bound in the range of $z$
if $m(r(z))=0$ is satisfied for some value of $z$.
Nevertheless, due to the condition $r(z)<2 m(r(z))$,
any hypersurface defined by $m(r(z))=0$ cannot be located
in a homogeneous region. Therefore, if such surfaces exist, they
will be embedded in a static region.

It is straightforward to check that the line element \eqref{metric:tau_x}
changes to the forms \eqref{eq.metricstatic}
and \eqref{eq.metrichomogeneous} by the coordinate transformation
\begin{align}\label{eq.coordtransf}
  d\tau&=d\tt+\bigg(1-\frac{{2\mofr}}{r(z)}\bigg)^{-1}
         \sqrt{\frac{2\mofr}{r(z)}}dz,
\end{align}
restricted to $0\leq2\mofr<r(z)$ and $0<r(z)<2\mofr$, respectively. This,
together with the fact that $m(r(z))$ cannot vanish in a homogeneous region,
explicitly shows that the coordinates $(\tau,z)$
cover completely one (or more) homogeneous regions,
and partially or completely [depending on whether some surface 
$m(r(z))=0$ exists]
one (or more) static regions.

\subsubsection{Near-horizon geometries}\label{sec_nearhorizon}
In the starting point of the static gauge, with $\dot \erad=\dot\ephi=0$
and $\sin(\lambda\kang)=0$
we had left out the case $\erad'=0$. We deal with that case now,
and thus assume $\sqrt{\erad}=a$ is a constant.
We show in Appendix~\ref{app.nhg} that, depending on whether the cosmological
constant is vanishing or not, $a$ must be given by [cf. \eqref{eq:a}]

\begin{subequations}\label{eq:a_main}
\begin{equation}
  a^2 =\frac{1}{2\Lambda} \left(1\pm \sqrt{1-4Q^2\Lambda}\right),
\end{equation}
for $\Lambda\neq 0$, and by
\begin{equation}
  |a| = |Q|,
\end{equation}
\end{subequations}
for $\Lambda=0$.
In addition, the value of $m$ is related to this constant as
\begin{equation}\label{eq.m.nhg}
  m=\frac{a}{2}.
\end{equation}
The corresponding geometry is isometric to $\manifold={\cal M}^2\times S^2$,
where $S^2$ is the sphere of radius $a$ and ${\cal M}^2$ a Lorentzian space
of constant Gaussian curvature,
\begin{equation}\label{eq.k}
  \kappa=\left(\Lambda-\frac{Q^2}{a^4}\right)(1-\lbar).
\end{equation}
These solutions contain the usual near-horizon geometries
  in the GR limit, which are trivially recovered for $\lbar\to0$.

For completeness, in Appendix~\ref{app.nhg} we show that a convenient choice of gauge,
given by $\shift=0$, $\dot\lapse=0$, and $\ephi=a\sqrt{1-\lbar}$,
leads, after relabeling $(t,x)$ as $(T,z)$ to the form of the metric
\begin{equation}\label{eq.metric.nhg1}
  ds^2=-\sin^2(\sqrt{\kappa} z) dT^2+dz^2+a^2d\Omega^2,
\end{equation}
for $\kappa>0$,
\begin{equation}\label{eq.metric.nhg2}
  ds^2=-\sinh^2(\sqrt{-\kappa} z) dT^2+dz^2+a^2d\Omega^2,
  \end{equation}
for $\kappa< 0$, and \[ds^2= - z^2 dT^2+dz^2+a^2 d\Omega^2,\] for $\kappa=0$. The range of $z$ is the real line if $\kappa\leq 0$ and $z\in(0,\pi/\sqrt{\kappa})$
if $\kappa>0$.
Observe that %
the value of $a$ does not depend on $\lbar$, so that $\lbar$
only affects the curvature of the Lorentzian part $\kappa$ with a multiplying factor $(1-\lbar)$. In particular we trivially recover in the limit $\lbar\to 0$
the Bertotti-Robinson ($\Lambda=0$) and Nariai ($Q=0$) solutions in GR (see, e.g., Ref. \cite{Brill1994}).

As in GR, we expect the near-horizon geometries to be obtained
performing a convenient limit of the whole family of spacetimes parametrized by
$(M,Q,\Lambda,\lbar)$. This limit corresponds to one of the two
possible limits to
the extremal cases (see, e.g., Ref. \cite{Bengtsson2022}),
which will be discussed below.
Here we simply note that
from \eqref{eq.m.nhg}, and reading $m(a)$ from \eqref{eq.mass} for $r=a$,
we can isolate $M$ to obtain
\begin{equation}\label{eq.M.nhg}
  M=\frac{a}{2}+\frac{Q^2}{2a}-\frac{\Lambda}{6}a^3=a\left(1-\frac{2}{3}a^2\Lambda\right),
\end{equation}
where we have used \eqref{eq:a_main} in the second equality.
This, together with \eqref{eq:a_main}, provides
a relation between the parameters $M$, $Q$, and $\Lambda$.
Since neither \eqref{eq:a_main} nor \eqref{eq.m.nhg}
depend on $\lbar$, these relations are the same as in GR.
In particular, from \eqref{eq.M.nhg} and \eqref{eq:a_main}
we get that, if $\Lambda=0$,
\begin{equation}
  M=|Q|,\quad a=|Q|,\qquad \implies \kappa=-\frac{1}{a^2}(1-\lbar), \label{BR_cond}
\end{equation}  
and, if $Q=0$, which implies a positive value of $\Lambda>0$,
\begin{equation}
  M=\frac{1}{3\sqrt{\Lambda}},\quad a=\frac{1}{\sqrt{\Lambda}}, \qquad \implies \kappa=\frac{1}{a^2}(1-\lbar).\label{Nariai_cond}
\end{equation}
In Appendix~\ref{app.nhg_limits} we show that Eqs.~\eqref{eq:a_main}
and \eqref{eq.M.nhg}
correspond indeed to the extremal cases, i.e., in the limit
where the location of two horizons coincides.
The conditions \eqref{BR_cond} and \eqref{Nariai_cond} correspond
to the conditions for the Bertotti-Robinson ($r_I=r_H$) and the Nariai ($r_H=r_C$)
spacetimes as limits of Reissner-Nordstr\"om and Schwarzschild-de Sitter,
respectively, in GR.

For completeness, we include here the ultra-extreme case (see, e.g., Ref. \cite{Brill1994}),
which can be seen as the limit where all horizons coincide
$r_I=r_H=r_C$, and it is given by
\[
  M=\frac{1}{3}2\sqrt{2} |Q|,\quad \Lambda=\frac{1}{4Q^2},\quad a=\sqrt{2}|Q|,
  \qquad \implies \kappa=0.
  \]

\subsection{Curvature invariants}\label{sec_curvature}

As expected, in the GR limit ($\lbar\rightarrow 0$) the geometry
reproduces the family of spherically symmetric solutions
of the classical Einstein-Maxwell equations coupled to
a cosmological constant $\Lambda$.
Except for the maximally symmetric cases with $M=0=Q$, i.e.,
Minkowski ($\Lambda=0$), de Sitter ($\Lambda>0$), or anti-de Sitter ($\Lambda<0$) spaces,
and the near-horizon geometries, all the geometries of this
family present a singularity at $r=0$ (see Ref. \cite{Brill1994}).
Let us now compute some
curvature scalars for this modified model to check
their possible divergences.

Using the metric in any of the forms presented above \eqref{eq.metricstatic}, \eqref{eq.metrichomogeneous}, or \eqref{metric:tau_x}, the Ricci 
scalar %
takes the form
\begin{align}\label{eq.R}
    {\cal R}=&   4\Lambda\left(1+\frac{\lbar}{2}\right)     +2\lbar\Bigg(\frac{3M^2}{r^4}+\frac{Q^2}{r^4}\left(1-\frac{4M}{r}+\frac{Q^2}{r^2}\right)-\Lambda\left(\frac{4M}{r}+\Lambda r^2\right)+ \frac{4\Lambda Q^2}{3r^2}\Bigg),
\end{align}
while the Kretschmann scalar reads
\begin{align}
{\cal R}_{abcd}{\cal R}^{abcd}=& \frac{8 \Lambda^2}{3}{(1+\lbar)}+\frac{48 M^2}{{r}^6}-\frac{96 M Q^2}{{r}^7}+\frac{56 Q^4}{{r}^8}-{\lbar}\left(\frac{8}{3}\Lambda^3 r^{2} { 
{-\frac{152Q^6}{r^{10}}}}-\frac{240M^3}{r^7}+\frac{P_{8}(r)}{r^{9}}\right)\nonumber\\
    & +{\lbar^2}\left(\frac{20}{27}\Lambda^4 r^{4}-\frac{40}{27}\Lambda^3r^{2}+\frac{16}{3}M\Lambda^3r +{\frac{108Q^8}{r^{12}}+\frac{324M^4}{r^8}+\frac{P_{10}(r)}{r^{11}}}\right),
\end{align}
where {$P_{8}(r)$} and {$P_{10}(r)$} are polynomials of degree {$8$} and {$10$} in $r$, respectively,
whose explicit form is not relevant for our discussion.

From these expressions we have that
the curvature also diverges in the modified model if points where $r=0$
are reached.
In fact, unlike in GR, where $\mathcal{R}$ is constant,
for $\lbar\neq 0$ the Ricci scalar diverges at $r=0$ for all the cases
except for the trivial $M=0=Q$ geometry.
Moreover, the Kretschmann scalar diverges faster
than its GR counterpart near $r=0$: it goes as $r^{-12}$ if $Q\neq 0$
and as $r^{-8}$ if $Q=0$ and $M\neq 0$, while in GR the dominant terms are $r^{-8}$
and $r^{-6}$, respectively.
In addition, and contrary to GR, the modified model introduces diverging values of the curvature
as $r\to\infty$ for all the cases with $\Lambda\neq 0$.

Finally, the Ricci scalar of the near-horizon geometries can be directly computed
from the line elements \eqref{eq.metric.nhg1} and \eqref{eq.metric.nhg2}, or simply
by adding the constant curvatures of the spheres $S^2$ and ${\cal M}^2$ (see, e.g., Ref. \cite{ONeill1983}),
  \[
    {\cal R}=2\left(\kappa+\frac{1}{a^2}\right).
  \]
  It is easy to check that
this result can also be obtained as a particular limit of
  \eqref{eq.R}, by making use of the relations \eqref{eq:a_main}, \eqref{eq.k},
  and \eqref{eq.M.nhg}.

\section{Study of the singularity resolution}\label{sec_study_singularity}

As seen above, nonvanishing values of the mass $M$ and charge $Q$
parameters introduce divergences in the curvature invariants of the
spacetime under consideration at $r=0$, while a nonvanishing
cosmological constant $\Lambda$ produces an infinite curvature
as $r\to\infty$. Therefore, this model can only resolve the GR singularity at $r=0$ by making the points where $r=0$ unreachable.
That is, the singularity at $r=0$ will be resolved only if there exists a positive lower bound for $r$,
which we will denote by $r_0>0$. Correspondingly, to get a singularity-free spacetime,
avoiding also the curvature divergence at
$r\to\infty$, the existence of a finite upper bound $\rmax$ will also be
necessary.

The goal of this section is to study the existence of such $r_0$ and $\rmax$
depending on the specific values of the parameters $M$, $Q$, $\Lambda$,
and $\lbar$. For such a purpose, in Sec. \ref{sec_ranges}, we discuss the
conditions that constrain the possible values of the scalar $r$,
and reduce them to the study of the roots of a fourth-order polynomial.
Section \ref{sec_singularity_free} presents then all the possible values of
the parameters $M$, $Q$, $\Lambda$, and $\lbar$,
that lead to a singularity-free spacetime, apart, of course, from
  the maximally symmetric and near-horizon geometries.
Let us stress, however,
that certain sets of values of the parameters will define two different spacetimes
(see Fig.~\ref{fig_schemer}):
one singularity free, with $r$ taking values in an interval of
the form $[r_0,\rmax]$ or $[r_0,\infty)$,
and another one defined in the range $r\in(0,R]$, for some certain
value $R<r_0$, which is thus affected by the singularity at $r=0$.
We analyze the existence of such singular spacetimes in  Sec. \ref{sec_non_uniqueness}.

\subsection{Possible ranges of the scalar $r$}\label{sec_ranges}

Apart from the requirement $r(z)>0$, the domain of $r(z)$
is restricted by its defining equation \eqref{eq.rp}.
In particular, this equation is a nonlinear autonomous
differential equation and, denoting the derivative with respect to $z$ with a prime, it can be written in the form
\begin{equation}\label{eq:withf}
 \frac{1}{2}(r')^2+V(r)=0,
\end{equation}
which formally corresponds to the Hamiltonian of a one-dimensional
Newtonian particle moving in the potential $V(r)$ with zero total energy. 
This analog model, where $r=r(z)$ is understood as the position of the particle at a time
$z$, will be very helpful in understanding the properties of the solution
$r=r(z)$ and, in particular, its domain of definition.
Since the energy of the particle
is zero, and it is conserved, the particle can only move along intervals
of $r$ where $V(r)$ is nonpositive. The roots of the potential
$V(r)=0$ thus define the turning points.

More precisely, depending on the positive roots of $V(r)$,
there might appear several, say $N$, nonoverlapping intervals $[r_{\rm min}^{(i)},r_{\rm max}^{(i)}]$,
with $i=1,\dots,N$, where $V(r)$ is nonpositive.
The lower boundary of a given interval is defined either by
a root of $V(r)$, that is, $V(r^{(i)}_{\rm min})=0$ or by $0$.
Correspondingly,
the upper boundary is either a root $V(r^{(i)}_{\rm max})=0$ or infinity.
Therefore, there are four possible kinds of intervals:
finite intervals of the form $[r_{\rm min}, r_{\rm max}]$ or
$[0, r_{\rm max}]$, and unbounded intervals of the form
$[r_{\rm min},\infty)$ or $[0,\infty)$.

To analyze the existence
of the solution $r(z)$ %
at the zeros of $V(r)$, and its behavior,
we first take the derivative with respect to $z$ of \eqref{eq:withf}
to obtain the usual equation for the acceleration,
\begin{equation}\label{eq.rpp}
 r''=-V_{,r}(r),
\end{equation}
where $r'(z)\neq 0$.
For continuity, this must also
be obeyed at the turning points $r'(z)= 0$. Contrary to the sign of the velocity $r'(z)$,
the acceleration $r''(z)$ is completely defined in terms of the
gradient of the potential $V(r)$.
As a consequence,  if $\alpha$ is a simple root of $V(r)$, with corresponding
$z_\alpha$ defined by  $\alpha=r(z_\alpha)$, continuity of $r''$
at $\alpha$ demands that if $V_{,r}(\alpha)<0$, then
$r'(z)=\mbox{sign}(z-z_\alpha)\sqrt{|V(r(z))|}$ around $z_\alpha$,
and $\alpha$ corresponds to a lower bound of the domain of definition of $r$,
whereas if $V_{,r}(\alpha)>0$ then
$r'(z)=-\mbox{sign}(z-z_\alpha)\sqrt{|V(r(z))|}$ around $z_\alpha$,
and  $\alpha$ is an upper bound.
From these relations one concludes that the function $r=r(z)$ is symmetric
around any of these simple roots, that is,
the solution will obey $r(z-z_\alpha)=r(z_\alpha-z)$.
Therefore, intervals with just one simple root of the
potential $\alpha$ are either of the form $[\alpha,\infty)$ or $[0,\alpha]$.
In the former (latter) case, there is a global minimum (maximum)
at $\alpha$,
and the coordinate $z$ will cover the whole
range of $z$ for which $r$ stays on its domain.
If $\alpha$ and $\beta$
are two simple roots of the potential, the solution $r(z)$
in the interval $[\alpha,\beta]$
oscillates between $\alpha$ and $\beta$ periodically and the
range of $z$ is the real line.

A second consequence of the acceleration equation \eqref{eq.rpp}
is that, if there is a root of the potential $\alpha$ with multiplicity higher than one, that is,
$V(\alpha)=V_{,r}(\alpha)=0$, then both $r'$ and $r''$
also vanish there. In fact, recursively deriving \eqref{eq.rpp},
one finds that all the derivatives of $r(z)$ must also vanish there,
and $\alpha$ is thus an equilibrium point.
Consider $n>1$ to be the lowest-order nonzero derivative, $V^{(n)}(\alpha)\neq 0$.
If $n$ is odd, then $\alpha$ is an inflection point;
the potential $V$ changes sign there, and therefore $\alpha$ clearly
determines the lower [if $V^{(n)}(\alpha)<0$] or upper [if $V^{(n)}(\alpha)>0$] bound of some interval.
In such a case, the particle reaches $\alpha$ only for infinite values of $z$
[i.e., $r(z)\to\alpha$ only as $|z|\to\infty$].
However, if $n$ is even, then $V$ has a local minimum [if $V^{(n)}(\alpha)>0$]
or a maximum [if $V^{(n)}(\alpha)<0$] at $\alpha$.
In the former case, $\alpha$ is a stable equilibrium point and the particle
stays put in that position $r(z)=\alpha$, and thus it does not
define a finite interval of $r$.
In the latter case, the point $\alpha$ is an unstable equilibrium
point, the particle takes infinite time to reach it (from either side),
and it thus defines both a lower and an upper bound to two disjoint intervals
of the form $[\alpha,\dots]$ and $[\dots, \alpha]$, respectively.

To sum up, simple roots of the potential $V(r)$ define turning points
of $r(z)$, and are also fixed points of a reflection symmetry.
In turn, double (or higher) roots are only reached
at asymptotic values of $z$, hence do not define
turning points and have no associated symmetry properties.
As a result, each of the intervals in $r$ will describe an independent
spacetime and the associated domain of definition on $z$ will determine the range of the coordinate $z$ of that domain.

In our particular case, the potential is %
given by
$V(r)=(-r+2\lbar m(r))/(2r)$. 
In order to analyze the intervals where it is negative,
it turns out convenient to define
\begin{equation}\label{def:P}
  P(r,\lproof):=\left\{
    \begin{array}{lr}
      \cfrac{\lproof\Lambda}{3}r^4-r^2+2\lproof M r-\lproof Q^2, & {\rm if}\,\,\, Q\neq 0,\\
      \cfrac{\lproof\Lambda}{3}r^3-r+2\lproof M, & {\rm if}\,\,\, Q= 0,
    \end{array}\right.
\end{equation}
in such a way that $V(r)=P(r,\lbar)/(2 r^\ell)$,
with $\ell=1$ for $Q=0$ and $\ell=2$ for $Q\neq0$.
This choice has been made so that %
$P$ is a polynomial with a nonvanishing free term in all cases.
Clearly, for $r>0$ the roots $\alpha$ of $P$ (on $r$) and $V$ coincide,
and the sign of the gradient $V_{,r}(\alpha)$ is the same as
the sign of $P'(\alpha,\lbar)$,\footnote{Primes on $P(r,s)$ will denote derivatives with respect to its first argument.} 
so that, in particular, $V_{,r}(\alpha)=0$ if and only if $P'(\alpha,\lbar)=0$. By iteration,
one can check that the first $k$ derivatives of $V$ are vanishing at $\alpha$
if and only if the first $k$ derivatives of $P$ vanish there.
In addition, in such a case, the sign of the next
derivative of both functions will coincide. Note that, for convenience,
we have included a second argument in the polynomial $P(r,s)$.
This is because $P$
also serves to study the horizon structure of the spacetime
(see  Sec. \ref{sec_horizons}), since
the function $G=1-2m/r$ satisfies %
$G(r)=-P(r,1)/r^\ell$.
In what follows, whenever
we speak about the roots and derivatives of $P(r,s)$, we will refer
to the roots on and the derivatives with respect to its first argument,
respectively.

\subsection{Singularity-free spacetimes}\label{sec_singularity_free}

The analysis of the solutions $r(z)$ and the possible singularity resolution
is thus reduced to the classification of the zeros
of $P(r,\lbar)$ in $r$, and the behavior of $P$ there, depending on the
specific values of $M$, $Q$, $\Lambda$, and $\lbar$.
Our concern in this section is to provide the cases in which the
solution avoids the curvature divergences. In particular, in Sec. 
\ref{sec_existence_r0} we present the set of cases
for which a lower bound $r_0$ for $r$ exists, and thus avoid
the singularity at $r=0$.
In Sec. \ref{sec_existence_rinfinity} we
analyze in what cases
the divergent behavior at $r\to\infty$ is avoided by an upper
bound $\rmax$ for $r$.
  
\subsubsection{Existence of $r_0$}\label{sec_existence_r0}
  
Let us start with the avoidance of the singularity at $r=0$. As already commented,
for that we need a positive minimum $r_0>0$ of $r(z)$.
This will be accomplished if there exists $r_0>0$ such that $P(r_0,\lbar)=0$
and $P(r,\lbar)<0$ on some interval for which $r_0$ is infimum. 
Equivalently, this condition can be stated in a local way by requesting
that there exists a positive root $r_0>0$ of the polynomial, $P(r_0,\lbar)=0$,
such that the first nonvanishing derivative of $P(r,\lbar)$
at $r_0$ is negative.
Since $P(r,\lbar)$ is (at most) a fourth-order
polynomial, in principle, there are four possibilities for that: 
(a) $r_0$ is a simple root and $P'(r_0,\lbar)<0$;
(b) $r_0$ is a double root and $P''(r_0,\lbar)<0$;
(c) $r_0$ is a cubic root and  $P'''(r_0,\lbar)<0$;
(d) $r_0$ is a quartic root and $P''''(r_0,\lbar)<0$.

The complete list of sets of values that lead to the existence of $r_0$
is given by Lemma~\ref{res:lemma1} in Appendix~\ref{app.proofs}.
In particular, $r_0$ exists only when (a)
or (b) hold, and thus the possibilities (c) and (d) are excluded.
Explicitly, $r_0$ is a simple root with $P'(r_0,\lbar)<0$
in the following cases:
\begin{itemize}
\item $C_1:=\Big\{$$\Lambda>0$, $Q\neq 0$, $M>0$, with
  $8Q^2<9\lbar M^2$ and $\Lambda\in(\Lambda_-,\Lambda_+)\cap(0,\Lambda_+)\Big\}$,
\item $C_2:=\Big\{$$\Lambda>0$, $Q= 0$, and $\Lambda\in \left(0,\frac{1}{9\slbar^3 M^2}\right)\Big\}$,
\item $C_3:=\Big\{$$\Lambda=0$ and $|Q|<\sqrt{\lbar} M\Big\}$,
\item $C_4:=\Big\{$$\Lambda<0$, $Q\neq 0$,  $|Q|<\sqrt{\lbar} M$, and $\Lambda\in(\Lambda_-,0)\Big\}$,
\item $C_5:=\Big\{$$\Lambda<0$, $Q= 0$, and $M>0\Big\}$,
\end{itemize}
where we have defined
\begin{equation}\label{def:Lambda_l}
  \Lambda_{\pm}:=\frac{3}{32\lbar^4 Q^6}\left[36 \lbar ^3 M^2 Q^2-27 \lbar ^4 M^4-8 \lbar ^2 Q^4
    \pm\sqrt{\lbar ^5 M^2 \left(9 \lbar  M^2-8 Q^2\right)^3}
  \right],
\end{equation}
while $r_0$ is a double root with $P''(r_0,\lbar)<0$ in the cases
\begin{itemize}
\item $D_1:=\Big\{$$\Lambda>0$, $Q\neq 0$, $M>0$, with
  $8Q^2<9\lbar M^2<9Q^2$ and $\Lambda=\Lambda_-\Big\}$,
\item  $D_3:=\Big\{$$\Lambda=0$ and $|Q|=\sqrt{\lbar} M>0\Big\}$,
\item $D_4:=\Big\{$$\Lambda<0$, $Q\neq 0$,  $|Q|<\sqrt{\lbar} M$ and $\Lambda=\Lambda_-\Big\}$. %
\end{itemize}
Observe that, by Remark~\ref{remarkbeta+}, the sign of
$\Lambda_-$ coincides with the sign of the combination $Q^2-\lbar M^2$.
In consequence, the fourth requirement 
in the case $C_1$ can be refined as $\{\{Q^2\leq \lbar M^2$ and $\Lambda\in(0,\Lambda_+)\}$ or $\{8Q^2<9\lbar M^2<9 Q^2$ and $\Lambda\in(\Lambda_-,\Lambda_+)\}\}$.

In the cases with a nonvanishing charge ($C_1$, $C_3$ with $Q\neq 0$,
and $C_4$),
the polynomial $P$ is fourth order. Therefore, it might have at most
four real roots, but one of them is always negative,
which leads to the existence of at most three positive roots $R$, $r_0$, and $\rmax$.
The degeneracy of $R$ and $r_0$ into a double root defines the corresponding
degenerate cases $D_1$, $D_3$, and $D_4$. Concerning the cases with a vanishing
charge ($C_2$, $C_3$ with $Q=0$ and $C_5$), the polynomial is third order, and $r_0$ cannot
degenerate into a double root.

Note that, in particular, the existence of $r_0$ generically requires a nonnegative value
of the mass parameter $M$. If $M$ is negative, the GR singularity at $r=0$
is not resolved by the present model. Concerning the charge $Q$, the condition
of singularity resolution introduces an upper bound proportional to the mass
$M$ and the polymerization parameter $\lbar$. Something similar happens
with the cosmological constant $\Lambda$: the singularity is resolved
whenever the absolute value of $\Lambda$ is below a certain maximum threshold.
These requirements are qualitatively similar to the conditions one finds
in GR for the existence of horizons. As it will be explained below, this is
not a coincidence but comes from the fact that the problem of the existence of horizons
is mathematically equivalent to the existence of $r_0$ for $\lbar\to1$.

The different cases $C_1-C_5$ have been arranged to correspond
to well-known black-hole metrics in GR. More precisely, $C_1$ (and its degenerate $D_1$) corresponds
to Reissner-Nordstr\"om-de Sitter, $C_2$ to Schwarzschild-de Sitter,
and $C_3$ (and its degenerate $D_3$) to Reissner-Nordstr\"om. This latter case also includes
Schwarzschild $(\Lambda=Q=0)$.
The last two cases, $C_4$ (along with its degenerate $D_4$) and $C_5$, correspond to 
Reissner-Nordstr\"om-anti-de Sitter and Schwarzschild-anti-de Sitter, respectively.
Note that the maximally symmetric de Sitter, anti-de Sitter, and Minkowski geometries are not
included in the above list, since in such cases no $r_0$ exists,
though these metrics are perfectly well defined at $r=0$.
Concerning the near-horizon geometries,
the Bertotti-Robinson-like solution (with $r_I=r_H$) is not included in the above
cases, but the Nairai-type solution (with $r_H=r_C$) is a particular case of $C_2$.
As shown in Fig. \ref{fig_schemer}, the reason is that the existence of $r_0$ implies that $r_I$ and $r_H$
  cannot be present in the same spacetime.

\subsubsection{Existence of $\rmax$}\label{sec_existence_rinfinity}

For the cases listed in the previous subsection,
let us now study the existence of an upper bound $\rmax$ of $r$
that avoids the divergent behavior of the curvature invariants
for large radii produced by a nonvanishing value of the cosmological
constant. By Remark~\ref{remark_rinf}, the cases $C_1-C_5$,
and their degenerate $D_1$, $D_3$, and $D_4$, 
are bounded from above by some $\rmax$ if and only if $\Lambda>0$.
Therefore, cases with a negative cosmological constant have
a singular behavior at large radii, while cases with a positive
cosmological constant will avoid the curvature divergence
by the presence of $\rmax$.
In the cases $C_1$ and $C_2$, $r_0$ and $\rmax$ will be a lower and upper
turning point, respectively, while in the cases $C_3$, $C_4$,
and $C_5$, $r_0$ will be the only turning point.
The ranges for $r$ and $z$ are given as follows:
\begin{itemize}
  \item In the cases $C_1$ and $C_2$ the solution $r(z)$, with
  image on $[r_0,\rmax]$, has support on the whole real line, $z\in\mathbb{R}$. This solution
is symmetric around $r_0$ and $\rmax$.
These turning points are reached at finite values of $z$,
and thus it is periodic.
\item  In the cases $C_3$, $C_4$, and $C_5$, the solution $r(z)$,
with image on $[r_0,\infty)$, has support on the
whole real line, $z\in\mathbb{R}$. This solution is symmetric around
the turning point $r_0$, and $r(z)\to \infty$ as $z\to \pm \infty$.
\end{itemize}
In the degenerate cases $r_0$ is not a turning point, but an unstable
equilibrium point which can only be reached at infinite values of $z$.
Therefore, one finds the following ranges of definition:
\begin{itemize}
\item  In the degenerate case  $D_1$, the solution $r(z)$,
with image on $(r_0,\rmax]$, has
support on the whole real line, $z\in\mathbb{R}$. This solution is
symmetric around the turning point $\rmax$,
and $r(z)\to r_0$ as $z\to \pm\infty$.
\item In the degenerate cases $D_3$ and $D_4$, the solution
$r(z)$, with image on $(r_0,\infty)$,
has support on the whole real line, $z\in\mathbb{R}$. This solution goes from
$r(z)\to r_0$ as $z\to -\infty$ to $r(z)\to +\infty$ as $z\to +\infty$.
\end{itemize}

In summary, the present model provides singularity-free spacetimes
in the cases $C_1$, $C_2$, and $C_3$, along with their
degenerate cases $D_1$ and $D_3$. In the GR limit this corresponds to the
Reissner-Nordstr\"om-de Sitter black hole (including
Schwarzschild-de Sitter, Reissner-Nordstr\"om, and Schwarzschild as particular cases).
In cases $C_4$, $C_5$, and $D_5$, although the GR singularity at $r=0$
is avoided, the spacetime presents a singularity at large radii.

\subsection{Nonuniqueness of the spacetime solutions: singularity at $r=0$}
\label{sec_non_uniqueness}

As already commented above, and schematically shown in Fig. \ref{fig_schemer},
when $Q\neq 0$ in the cases listed above in which $r_0$ exists,
the same values of the parameters $(M,Q,\Lambda,\lbar)$
define two different spacetimes:
one singularity free, with $r$ taking values in an interval with
an infimum $r_0$, and another one defined in a range containing
the origin $r=0$ and with an upper bound $R\leq r_0$.
Concerning the cases presented above, the requirement $Q\neq 0$
excludes the cases $C_2$ and $C_5$, while for the rest,
by Remark~\ref{remark_origin}, we have the following:
  \begin{itemize}
  \item In $C_1$, $C_3$ with $Q\neq0$, and $C_4$, the solution $r(z)$, with image on $(0,R]$, has support on an interval $z\in (-z_0,z_0)\subset \mathbb{R}$,
and $r(\pm z_0)=0$. This solution 
is symmetric around $R=r(0)$, and $R<r_{0}$ is a critical point with
$P(R,\lbar)=0$ and $P'(R,\lbar)>0$.
\item In the degenerate cases $D_1$, $D_3$, and $D_4$, the critical points
$r_0$ and $R$ coincide. For such cases, the solution $r(z)$, with image on $(0,r_0)$,
has support on an interval $z\in (z_0,\infty)\subset \mathbb{R}$ and it is monotonic.
The constant $z_0$ can be chosen so that $r(z_0)=0$, and $r\to r_0$
in the limit $z\to\infty$.
\end{itemize}
We can thus state that, whenever $Q$ is nonvanishing and $r_0$ exists,
the generalization of the GR Hamiltonian presented here
originates the existence of more than one spacetime solution for the same
  set of values of the parameters $(M,Q,\Lambda,\lbar)$.
However, the singularity-free requirement, in particular, singles out
only one of the two possible solutions. That requirement on the spacetime solution
  is equivalent to demand a ``black-hole solution'', which we define
  as that containing a static region exterior to an isolated horizon.
  We will use that wording in the next section, where we study
  the global structure of the solutions.

\section{Global structure of the singularity-free family of spacetimes}\label{sec:global}

In this section we focus on the analysis of the singularity-free
spacetimes from the previous section. That is, we analyze the global structure
of the family of solutions given by the sets of values
of the parameters presented in the cases $C_1$, $C_2$, and $C_3$, and their degenerate cases
$D_1$ and $D_3$.
Note that, in particular, this restricts the discussion
to a positive mass parameter $M>0$, and to a nonnegative cosmological
constant $\Lambda\geq 0$.

The first consequence of this is that $m(r)$ is positive for all
$r\geq r_0$,
and therefore any point of the spacetime solution in those cases is
included in a domain of the type covered by the coordinates $(\tau,z)$.
To show that $m(r)>0$ for all $r\geq r_0$,
we define the auxiliary function, cf. \eqref{eq.mass},
$$
j(r):=2 r m(r)  = \frac{\Lambda}{3} r^4 +M r - Q^2.
$$
Because $M>0$ and $\Lambda\geq 0$,
$j(r)$ is monotonically increasing for all $r>0$. Further, since its value
at $r=r_0$ is positive, $j(r_0)=r_0m(r_0)=r_0^2/\lbar$ [because $r_0$ is a root of $V(r)$],
it has no zeros in $r>r_0$. Therefore $j(r)$, and thus also $m(r)$, are positive for all $r>r_0$.

Since any point in the singularity-free spacetimes
  (cases $C_1$, $C_2$, $C_3$, $D_1$, and $D_3$) can be covered by
  the coordinates $(\tau,z)$, their domain of definition, which we will call $\mathcal{U}$
  in the remainder, provides us with the fundamental building block for the
  global study of the solution. Once we produce
  the conformal (Penrose) diagram for $\mathcal{U}$,
  we can follow the usual  periodic construction to build up
  a maximal analytic extension $\mathcal{M}$.

The remainder of this section is divided into two subsections. In Sec. \ref{sec_horizons}
we discuss the existence of the different horizons in the solutions,
while in Sec. \ref{sec_conformal}
we briefly present the main ingredients for the construction of the conformal
diagrams by means of the global properties around the horizons
  and critical values of $r$.
The diagrams for all the considered cases are displayed in Appendix~\ref{sec_app_conformal}.
We exclude from this analysis the near-horizon geometries,
which have a trivial global structure.

\subsection{Horizon structure}\label{sec_horizons}

Let us start by checking the geometry of the hypersurfaces of constant $z$,
and thus of constant $r$. Imposing $z=a$ in the line element \eqref{metric:tau_x}, 
we obtain the induced metric
\[
  h=%
  -\chif(r_a)d \tau^2+r_a^2 d\Omega^2,
\]
with $r_a=r(z_a)$, and the function $G(r)=1-2m(r)/r$ already defined above.
In consequence, the hypersurface $z=a$ will be timelike, null, or spacelike
if
$G(r_a)>0$, $G(r_a)=0$, or $G(r_a)<0$,
respectively.
Since the critical values of $r$, i.e., $r_0$, $\rmax$,
and $R$, are zeros of $V(r)=-\frac{1}{2}(1-2\lbar m(r)/r)$, the function $G(r)$
is negative there, $G(r_0)=G(\rmax)=G(R)=(\lbar-1)/\lbar<0$.
In consequence, the sets of points that satisfy $r=r_0$ or $r=\rmax$ define spacelike hypersurfaces
in the manifold. We will refer to those as critical hypersurfaces,
because $r$ attains a local minimum or maximum there.

Next, we compute the mean curvature vector $H=(2/r)\nabla r$
of the surfaces $S_a$ of constant $z=a$ and $\tau$, which correspond to spheres of
area $4\pi r_a^2$, to obtain
\[
  H= \frac{2}{r_a^2} r'(a)
  \left(\sqrt{2m(r_a) r_a}\partial_\tau - (2m(r_a)-r_a)\partial_z\right).
\]
Therefore, $H$ %
is future pointing if $r'(a)>0$ and past pointing if
$r'(a)<0$. Its module, after using \eqref{eq.rp}, takes the form
\[
  g(H,H)= \frac{4}{r_a^2} (r'(a))^2
  \left(1-\frac{2m(r_a)}{r_a}\right)=
  \frac{4}{r_a^2}\left(1-\frac{2\lbar m(r_a)}{r_a}\right)
  \left(1-\frac{2m(r_a)}{r_a}\right).
\]
As a result, using that $\lbar\in(0,1)$, we have that
\begin{itemize}
\item $r_a>2m(r_a)\implies H$ is spacelike,
\item $r_a=2m(r_a)\implies H\neq 0 $ and $H$ is null,
\item $r_a<2m(r_a)$ and $r_a\neq 2\lbar m(r_a) \implies H$ is timelike,
\item $r_a= 2\lbar m(r_a)%
\implies H=0$.
\end{itemize}
Observe that, in the fourth case, we necessarily have $r_a<2m(r_a)$.
Therefore, in particular, the critical hypersurfaces $r=\alpha$
are minimal spacelike hypersurfaces covered by spheres
of area $4\pi \alpha^2$.
In addition, these spheres are nontrapped in the regions where $r>2m(r)$,
trapped (to the future) where $r<2m(r)$ and $r'(z)>0$,
and antitrapped (to the future) where $r<2m(r)$ and $r'(z)<0$.

In fact, the function $G$ can be defined intrinsically as
$G=-g(\xi,\xi)$, where $\xi$ is the Killing vector field that provides
the staticity in the static regions, and the homogeneity in the homogeneous
regions.
In particular, in the coordinates $(\tau,z)$, this vector reads
$\xi=\partial_\tau$. On the hypersurfaces $r=\horizon$ with
$G(\horizon)=0$, the Killing vector $\xi$
is null, and these are therefore Killing horizons
(called extremal or degenerate if $\horizon$
is a higher-multiplicity root of $G$, so that
$\chif'(\horizon)= 0$).
As usual, then, the nontrapped regions correspond to the static regions
with a timelike $\xi$,
while the homogeneous regions (with a spacelike $\xi$) are trapped where
$r'>0$ and antitrapped where $r'<0$.
Note that the critical hypersurfaces are always located
in homogeneous regions.

The function $G$ is independent of the polymerization
parameter $\lbar$, and therefore the determination of the horizons
and their relation with their limiting regions coincides
with that in GR.
As a result, there are at most three horizons: $r_I$, $r_H$, and
$r_C$, each one related to one of the parameters of
the model. Namely, $r_I$ can only exist for a nonvanishing value
of the charge parameter $Q\neq 0$, while the existence of $r_H$ and
$r_C$ requires a positive value of $M$ and $\Lambda$, respectively.
Also, in the nonextremal case
the horizons $r_I$ and $r_C$ bound from above (in $r$) static regions
[where $G(r)>0$], while $r_H$ bounds from above homogeneous regions
[where $G(r)<0$]. %
If they exist, they generically obey $r_I<r_H<r_C$,
although they can degenerate into the extremal cases
$r_I=r_H$, $r_H=r_C$, or $r_I=r_H=r_C$ (see, e.g., Ref. \cite{Brill1994}).

We are restricting ourselves to the cases
where $\Lambda\geq 0$ and $r_0$ exists
and, if $\Lambda>0$, also $\rmax$ exists.
As mentioned above,
$G(r_0)=G(\rmax)=(\lbar-1)/\lbar$ are negative.
On the other hand, direct inspection shows that in a neighborhood of $r=0$, $G(r)>0$
if $Q\neq 0$, while $G(r)<0$ if $Q=0$ and $M>0$.
Therefore, since $G(r)$ is a smooth function with no poles in $r>0$,
in the case $Q\neq 0$, there must be a root $r_I$ of $G(r_I)=0$
between the origin and $r_0$, that is, $r_I<r_0$. The existence of $r_0$
requires thus the existence of $r_I$, although this inner horizon
is not contained in the singularity-free domain.

Concerning the other two horizons ($r_H$ and $r_C$),
on the one hand, for $\Lambda=0$ we have cases $C_3$ and $D_3$, and the range of the solution
is $[r_0,\infty)$ and $(r_0,\infty)$, respectively. Since $G(r)$ tends to $1$ as
$r\rightarrow\infty$ %
and $G(r_0)$ is negative,
the function $G(r)$ must have one root larger than $r_0$.
Since for $\Lambda=0$ there is no $r_C$, that root must correspond to $r_H$.
Therefore we have that $r_0<r_H<\infty$.
As a result, for $\Lambda=0$ the existence of $r_0$ always requires the existence
of the horizon $r_H$, and no cosmological horizon $r_C$ exists.

On the other hand,
for $\Lambda>0$, we have cases $C_1$, $C_2$, and $D_1$, and the range is given by
$[r_0,\rmax]$ (the first two) or $(r_0,\rmax]$ (the degenerate case).
Since $G$ is negative at both extremes, $G(r_0)<0$ and $G(\rmax)<0$,
there can be either no roots, two simple roots, or one double root of $G$ in that interval.
These two roots correspond to $r_H$ and $r_C$ if different, and the double root
defines the extremal case with $r_H=r_C$. In other words, the existence of both $r_0$ and $\rmax$
requires either the existence of both $r_H$ and $r_C$ (with equal or different values),
or none of them. More explicitly each of the cases $C_1$, $C_2$, and $D_1$
is subdivided into three different cases as follows,
\begin{itemize}
\item $C_1^{\rm BH}:=\left\{C^1\,\,|\,\,\Lambda_+|_{\slbar\to 1}>\Lambda_- \mbox{ and } \Lambda\in \left(\Lambda_-,\Lambda_+|_{\slbar\to 1}\right)\cap \left(0,\Lambda_+|_{\slbar\to 1}\right)\right\}$,
  where two horizons $r=\bhhor$ and $r=r_C$ exist with $\bhhor<r_C$;
\item $C_1^{\rm extremal}:=\left\{C^1\,\,|\,\,\Lambda_+|_{\slbar\to 1}>\Lambda_- \mbox{ and }
    \Lambda=\Lambda_+|_{\slbar\to 1}\right\}$, where
    there is a degenerate horizon at $r=\bhhor=r_C$;
\item $C_1^{\rm cosmos}:=\left\{C^1\,\,|\,\,\Lambda_+|_{\slbar\to 1}\leq \Lambda_- \mbox{ or }
      \Lambda\in (\Lambda_+|_{\slbar\to 1},\Lambda_+)\right\}$
where no horizons exist;
\item $D_1^{\rm BH}:=\left\{D_1\,\,|\,\, \Lambda_+|_{\slbar\to 1}> \Lambda_-\right\}$, where two horizons
  exist at $r=\bhhor$ and $r=r_C$ with $\bhhor<r_C$;
\item $D_1^{\rm extremal}:=\left\{D_1\,\,|\,\, \Lambda_+|_{\slbar\to 1}= \Lambda_-\right\}$,
where there is a degenerate horizon at $r=\bhhor=r_C$;
\item $D_1^{\rm cosmos}:=\left\{D_1\,\,|\,\, \Lambda_+|_{\slbar\to 1}< \Lambda_-\right\}$,
  where no horizons exist;
\item $C_2^{\rm BH}:=\left\{C_2\,\,|\,\, \Lambda\in\left(0,\frac{1}{9M^2}\right)\right\}$,
where two horizons $r=\bhhor$ and $r=r_C$ exist with $\bhhor<r_C$;
\item $C_2^{\rm extremal}:=\left\{C_2\,\,|\,\, \Lambda=\frac{1}{9M^2}\right\}$, where
there is a degenerate horizon at $r=\bhhor=r_C$;
\item $C_2^{\rm cosmos}:=\left\{C_2\,\,|\,\,\Lambda\in\Big(\frac{1}{9M^2},\frac{1}{9\slbar^3 M^2}\Big)\right\}$, where no horizons %
  exist.
\end{itemize}
Note that this subdivision is based on the sign of the difference
$\Lambda_+|_{\slbar\to 1}-\Lambda_-$, which depends on $(M,Q,\lbar)$, and, as shown in
Appendix~\ref{app.proofs.hor}, in each case it can be either positive, negative, or zero.

The use of the names for the subcases is descriptive:
we use the superscript ``BH'' for black hole, so that $C_1^{\rm BH}$ and $C_2^{\rm BH}$
correspond to black-hole solutions as defined in the previous section.
In turn, the cases labeled with the superscript ``cosmos''
correspond to solutions of Kantowski-Sachs type,
with homogeneous spacelike slices and no horizons,
while the superscript ``extremal''
stands for cases with a degenerate horizon.
Observe that these extremal cases are composed by homogeneous regions,
so they describe cosmological solutions with horizons.
Since $C_3$ and $D_3$
are in all cases black-hole solutions
with the critical hypersurface $r=r_0$ hidden behind a
horizon $r_H$, we will not make use
of any specific label to identify them.

In order to end this section, let us point out
the deep connection
between the horizons $r_I$, $r_H$, and $r_C$, and the critical points
$R$, $r_0$, and $\rmax$. Let us recall that the former are the positive
roots of the function $G(r)=(1-\frac{2 m}{r})$, while the latter
correspond to the positive roots of $V(r)\propto (1-\frac{2\slbar m}{r})$.
Since in terms of $P(r,s)$ we have $G(r)\propto P(r,1)$
  and $V(r)\propto P(r,\lbar)$, it becomes explicit that
a rescaling of the parameters 
$(M,Q,\Lambda)\rightarrow(\lbar M, \sqrt{\lbar}Q, \lbar\Lambda)$
maps the roots of the first into the second.
In particular,
one has the following specific relations between the critical points of
a given solution, and the horizons of its rescaled version:  
$R(M,Q,\Lambda,\lbar)=r_I(\lbar M,\sqrt{\lbar}Q,\lbar\Lambda)$,
$r_0(M,Q,\Lambda,\lbar)=\bhhor(\lbar M,\sqrt{\lbar}Q,\lbar\Lambda)$,
and $\rmax(M,Q,\Lambda,\lbar)=r_C(\lbar M,\sqrt{\lbar}Q,\lbar\Lambda)$.
Therefore, in particular, for a given set of parameters $(M,Q,\Lambda,\lbar)$,
the minimum $r_0$ takes the value of the GR horizon $r_H$
with parameters $(\lbar M,\sqrt{\lbar}Q,\lbar\Lambda)$.
We can thus synthesize the singularity resolution principle of this model
in the following form:
In this model a solution with parameters $(M,Q,\Lambda,\lbar)$ exists that avoids the singularity at $r=0$
if and only if the singularity of the GR solution with
parameters $(\lbar M,\sqrt{\lbar}Q,\lbar\Lambda)$ is not naked.

Although it is implicit from the above analysis, let us note that
  the horizons $r_I$, $r_H$, and $r_C$ can never coincide with the critical
  values $R$, $r_0$, and $\rmax$. This comes explicit through the identity
  \begin{equation}\label{id.P}
    \lbar P(r,1)=P(r,\lbar)-(\lbar-1)r^\ell,
  \end{equation}
  or, equivalently,
  $
    \lbar G(r)=(\lbar-1)-2V(r),
  $
  which holds by construction. Therefore, $G(r)$ and $V(r)$ cannot vanish
  at the same point.

\subsection{Elements for the construction of the conformal diagrams}\label{sec_conformal}

From the previous analysis we have determined that if the conditions
$r=r_0$ and $r=\rmax$ define a set of points in the manifold,
this set forms a spacelike hypersurface. As explained in Sec. \ref{sec_singularity_free},
for the generic cases
under consideration ($C_1$, $C_2$, and $C_3$), the critical
values $r_0$ and $\rmax$ are simple roots  of $V(r)$, and are thus
reached at finite values of the coordinate $z$. Therefore, $r=r_0$ and $r=\rmax$
  define hypersurfaces, which are spacelike.
  Nonetheless, in the degenerate cases $D_1$ and $D_3$,
the value $r_0$ happens to be a double root of $V(r)$
and therefore it is reached at infinite values of $z$.

In order to characterize the sets of points defined
  by constant $r$ as $r\to r_0$
in the degenerate cases, %
we start by analyzing the radial geodesics in the regions
around those values.
The radial geodesics of the metric \eqref{metric:tau_x}, parametrized as
$(\tau(s),z(s))$ with affine parameter $s$,
are determined by
\begin{align*}
  \gamma&=-\left(1-\frac{2m(r(z))}{r(z)}\right)\left(\frac{d\tau}{ds}\right)^2
  +\left(\frac{dz}{ds}\right)^2+2\sqrt{\frac{2m(r(z))}{r(z)}}\frac{d\tau}{ds}\frac{dz}{ds},\\   \mathcal{E}&=    -\bigg(1-\frac{2m(r(z))}{r(z)}\bigg)\frac{d\tau}{ds} +\sqrt{\frac{2m(r(z))}{r(z)}}\,\frac{dz}{ds},%
\end{align*}
with $\gamma=0,1,-1$ for null, spacelike, and timelike geodesics, respectively,
and the constant $\mathcal{E}$ is
the energy (the conserved quantity associated with the timelike Killing vector field $\partial_\tau$).
The combination of both equations yields
\begin{align}\label{zdot_E}
    \left(\frac{dz}{ds}\right)^2 =\mathcal{E}^2+\gamma\bigg(1-\frac{2m(r(z))}{r(z)}\bigg).
\end{align}
We consider the null geodesics
with nonzero energy, $dz/ds\neq 0$,
so that we can choose the parameter $s$ so that $dz/ds=\signgeod=\pm 1$.
This means that $z$ is an affine parameter of the radial null geodesics,
and therefore
the (affine) distance from any point to a point moving
towards %
a double root of $V(r)$
goes to infinity. It is straightforward to check that
the same holds for timelike and spacelike geodesics
[the culprit in all cases is the appearance of a $r'$, and thus a
  $\sqrt{-2 V(r)}$, in the denominator of the integral for the proper time
  or affine parameter on each case, see Appendix~\ref{app.tortoise}]. 
In consequence, any double root of $V(r)$ %
truly represents an ``infinity'' in our manifold,
in the sense that geodesics reach those points at
  infinite values of the affine parameter. %

In order to establish the character of those infinities,
  one can follow the usual procedure for the construction
  of Penrose diagrams in terms of the study of the zeroes
  of the relevant metric functions %
  in some suitable coordinates (see, e.g., Refs. \cite{Brill1994,notes_jose}).
Here we provide a summary of the main results,
and leave the more detailed analysis to Appendix~\ref{app.tortoise}:
\begin{itemize}
\item The zeros $r=2m(r)$ of $G(r)$ show the usual isolated horizon structure
  in both the nonextremal (simple root) and the extremal (double root) cases.
  Let us recall that a triple root is prevented by the existence of $r_0$.
\item If $r=2\lbar m(r)$ is a simple root of $V(r)$, then it is a minimal spacelike hypersurface. This applies to all critical points ($r_0$ and $\rmax$) present in the generic cases $C_1$, $C_2$, and $C_3$, as well as to $\rmax$ in the degenerate case $D_1$.
\item If $r=2\lbar m(r)$ is a double root of $V(r)$, then it represents a null future or past boundary at infinity.
This applies to $r_0$ in the cases $D_1$ and $D_3$.
\end{itemize}

With all this information at hand, and that from the previous subsection,
one can produce the Penrose diagrams for each singularity-free
case, which are shown in Appendix~\ref{sec_app_conformal}.

\section{Conclusions}\label{sec_conclusions}

In this article we extend the vacuum spherical black-hole model with
holonomy corrections presented in
Refs. \cite{Alonso-Bardaji:2021yls,Alonso-Bardaji:2022ear} to
incorporate charge and a cosmological constant. We do so by
considering a canonical transformation of phase-space variables plus a
regularization of the GR Hamiltonian, under the condition that the
hypersurface deformation algebra closes.  The structure function in
the commutation relation between two Hamiltonian constraints is shown
to have the correct transformation properties to embed the $3+1$
theory in a four-dimensional spacetime. This fact allows us to
construct the corresponding metric tensor in terms of phase-space functions in
a completely unambiguous and covariant way.
 
After solving the equations of motion, we obtain that the resulting
metric is described in terms of four free parameters: the mass $M$,
the charge $Q$, and the cosmological constant $\Lambda$, which already
appear in the GR solution, plus the additional parameter
$\lbar\in(0,1)$, which measures the departure of the
effective model from GR (which is recovered in the limit
$\lbar\to 0$). This bounded constant is directly related to the
polymerization parameter of the holonomy corrections. Remarkably,
the Minkowski geometry is an exact solution for $M=Q=\Lambda=0$, and any value of $\lbar$.

For certain values of the parameters,
the singularity that appears in GR at $r=0$ is resolved in this
model by the appearance of a finite minimum $r_0$ for the area-radius
function $r$. More precisely, our results
concerning the singularity resolution
can be stated in a very compact way as follows:
\\

{\it
Given the parameters $(M,Q,\Lambda,\lbar)$,
this model provides a solution that avoids the singularity at $r=0$,
if and only if the singularity of the GR solution with
parameters $(\lbar M,\sqrt{\lbar}Q,\lbar\Lambda)$ is not naked.
}\\

 However, the specific form of the holonomy corrections
 produces curvature divergences as $r\to\infty$, whenever the cosmological
constant is not zero. Quantum-gravity corrections are in principle
expected to be relevant near the central singularity, though this
model also shows large effects at cosmological scales.  Nonetheless,
depending on the sign of the cosmological constant, the model still
provides a physical description consistent with observations.  On the
one hand, for a positive cosmological constant, the equations describe
a cyclic cosmological evolution in de Sitter-like regions, and the
curvature divergence at $r\to\infty$ is never reached due to the
presence of a finite maximum $\rmax$ of the area-radius function.
On the other hand, such bound does not appear in the cases with a
negative cosmological constant, and thus the model breaks down in the
asymptotic regions of Reissner-Nordstr\"om-anti-de Sitter black holes.

We explicitly characterize the existence of the extrema of the area-radius
function, $r_0$ and $\rmax$, in terms of the values of the parameters
$(M,Q,\Lambda,\lbar)$.  %
The main conclusion is that the family of singularity-free spacetimes of
this model is given by the sets of values
defined by the cases $C_1$, $D_1$, $C_2$, $C_3$,
and $D_3$ in Sec.~\ref{sec_singularity_free}.
In particular, the requirement of regularity of
the solution imposes $M$ and $\Lambda$ to take nonnegative values. In
addition, the charge $Q$ needs to be below a certain maximum
threshold. These conditions suit well any astrophysical black hole, and
therefore we can conclude that this model provides a completely
regular and singularity-free description for any such realistic black hole.

The analysis of the horizon structure of the singularity-free
family of geometries shows that the different cases are subdivided
into two generic families: those that describe black-hole spacetimes, and those that
describe an evolving cosmology. On the one hand, for the black-hole solutions,
$r=r_0$ lies behind a horizon and, if $\Lambda>0$, a maximum value of the area-radius function $\rmax$
exists, and it is located beyond a cosmological horizon. For $\Lambda=0$
the spacetime is asymptotically flat and there is no $\rmax$. In all cases
there is no inner Cauchy horizon. %
On the other hand, the cosmological solutions, which are everywhere spacelike homogeneous, present two kind of behaviors. 
  One class, that we call cosmos,
  does not contain any horizon, and it is thus formed by
  a whole (spatially connected) region that expands and contracts from
  $r_0$ to $\rmax$ in a cyclic manner.
  The other class, which we call extremal, contains degenerate horizons
  that separate homogeneous regions. This extremal class
  corresponds to a limiting case between the black-hole 
and cosmos solutions. This limit is not unique,
  since, as in GR, we show that
  the limits that go extremal (when two or more horizons degenerate)
  provide also the so-called near-horizon geometries.

If points determined by $r=r_0$ and $r=\rmax$ are located in the manifold,
they define minimal spacelike hypersurfaces embedded in
a homogeneous region of the spacetime. In some limiting cases
such points define null infinities.
All these properties can be seen
in Figs.~\ref{fig.C1BH}--\ref{fig.D3}, where we display the Penrose diagrams of all
the singularity-free family of solutions with parameters
  $(M,Q,\Lambda,\lbar)$, apart
  from the near-horizon geometries.
  As it is explicit in the figures, in the cases $C_1$, $C_2$, and $C_3$,
  the causal structures of the solutions with parameters
  $(M,Q,\Lambda,\lbar)$ and $(M,Q=0,\Lambda,\lbar)$
  coincide.
  
  As an interesting application, let us remark
  that any spherical astrophysical black hole, that is, with a relatively
large mass, small charge and embedded in a universe with a small
positive cosmological constant, is described in this model
by the conformal
diagram shown in Fig. \ref{fig.C1BH}. 
In this case, any particle that crosses the
``black-hole'' horizon $r_H$ ends up emerging through a ``white-hole'' horizon
$r_H$ in the finite future.
The minimal hypersurface $r=r_0$ simply implies a minimum value of the scalar
$r$. The same applies to any particle crossing the cosmological
horizon $r_C$: it reaches a maximum value $\rmax$ of the area-radius
function, and then goes back to smaller radii and crosses another horizon $r_C$.

\section*{ACKNOWLEDGEMENTS}
A.~A.~B. acknowledges financial support by the FPI Fellowship \mbox{PRE2018-086516} funded by\linebreak MCIN/AEI/10.13039/501100011033 and by ``ESF Investing in your future''. 
This work was supported by the Basque Government Grant \mbox{IT1628-22}, and by the Grant PID2021-123226NB-I00 (funded by\linebreak MCIN/AEI/10.13039/501100011033 and by “ERDF A way of making Europe”).

\appendix

\section{Conformal diagrams}\label{sec_app_conformal}
In this Appendix we show the conformal diagrams of the singularity-free
family of spacetimes given by the set of values of the parameters defined
in cases $C_1$, $D_1$, $C_2$, $C_3$, and $D_3$.
Along with each diagram we present a schematic plot
where we display in blue the polynomial $P(r,\lbar)$
and in red the curve $(\lbar-1)r^\ell$, with $\ell=1, 2$
for $Q=0$ and $Q\neq 0$, respectively.
The zeros of the polynomial
define the critical values $R$, $r_0$, and $\rmax$,
while the intersection between the polynomial and the red curve
marks the presence of a horizon, that is, a root of $G(r)$,
because, from \eqref{id.P}, we have
$G(r)=0\iff  P(r,\lbar)=(\lbar-1)r^\ell$ for all $r>0$.

We follow the usual conventions for the Penrose diagrams.
The main diagrams correspond to the domain
$\mathcal{U}$, covered by the coordinates $(\tau, z)\in \mathbb{R}^2$,
in which the metric is given by \eqref{metric:tau_x}.
We also depict a maximal analytic extension $\mathcal{M}$
following the usual periodic construction.
Both the domain $\mathcal{U}$ in the diagrams, and the corresponding range
of values of $r$ in the function plots 
are represented by a gray background color.
Likewise, in diagrams and plots,
horizons are depicted as red lines, and critical (minimal) hypersurfaces
as purple lines. Lines are continuous when covered by $\mathcal{U}$ and
dotted when not. Long-dashed lines correspond to (null) infinities,
all at $z\to \pm \infty$: dark gray for asymptotic ends ($r\to\infty$),
and purple for ``minimal'' ends  ($r\to r_0$).
The small rings denote holes in the diagram, and correspond to either
timelike ($i^{\pm}$) or spacelike ($i^0$) infinities.
Some hypersurfaces of constant $z$,
and thus of constant $r$, are represented by white thin curves,
and that serves to indicate the static or homogeneous character of the corresponding regions.

Some diagrams cannot be depicted over a simple sheet of paper because,
  at the intersections of horizons and critical hypersurfaces,
  the diagram bifurcates. Observers coming from the right or left
  of such a bifurcation will have nonintersecting futures. Take, for instance,
  the right extreme of the $r=r_0$ line in Fig.~\ref{fig.C1BH}: the observer
  crossing $r=r_0$ and the observer crossing $r=r_\infty$ will end up in disconnected
  future regions.
  That fact is indicated in the drawings by a shadowed
overlapping of some parts of the diagram over others.

\begin{figure}
  \centering
    \includegraphics[width=0.7\textwidth]{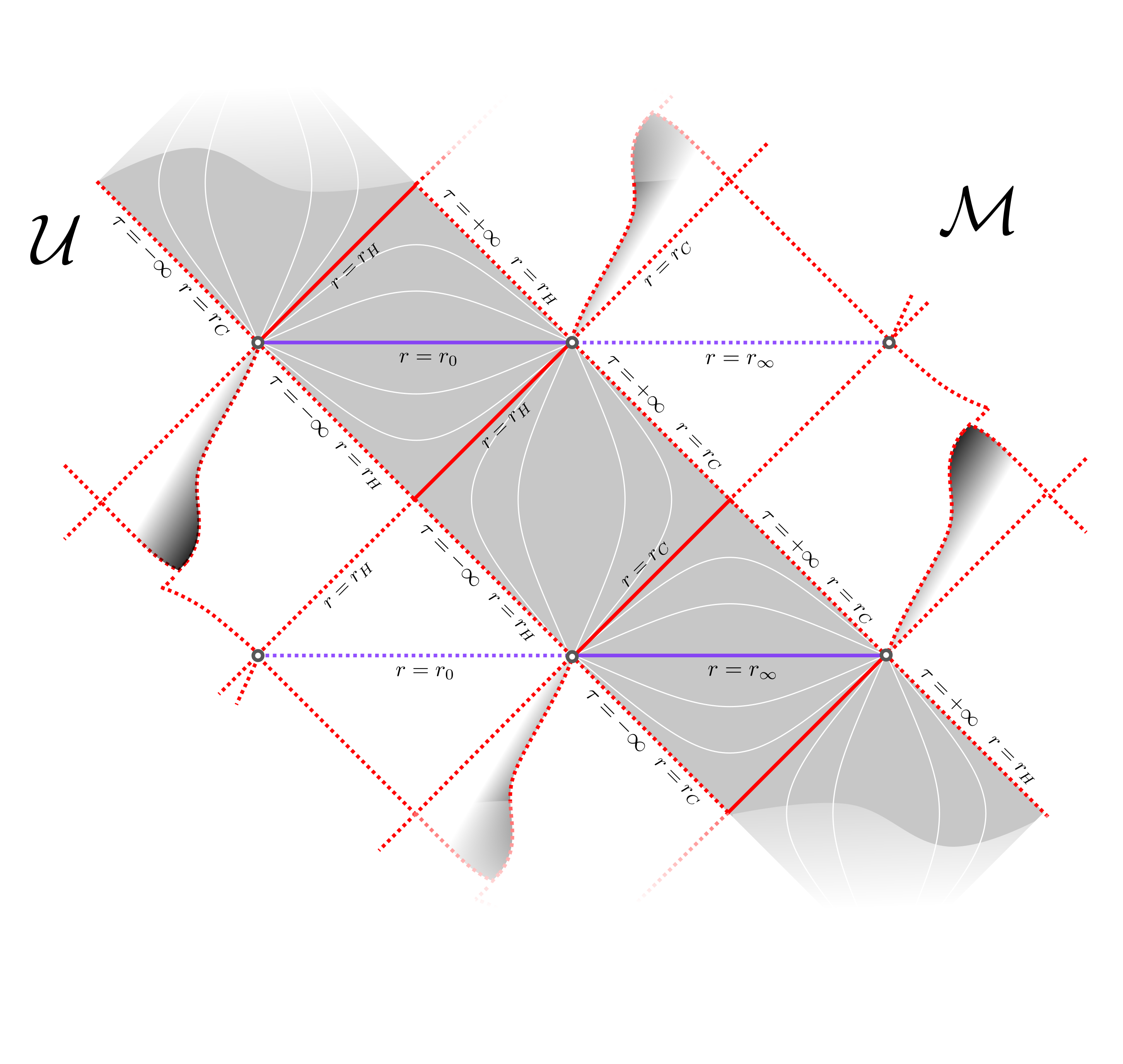}
  \hfill
  \begin{minipage}[b]{0.25\textwidth}
    \includegraphics[width=\textwidth]{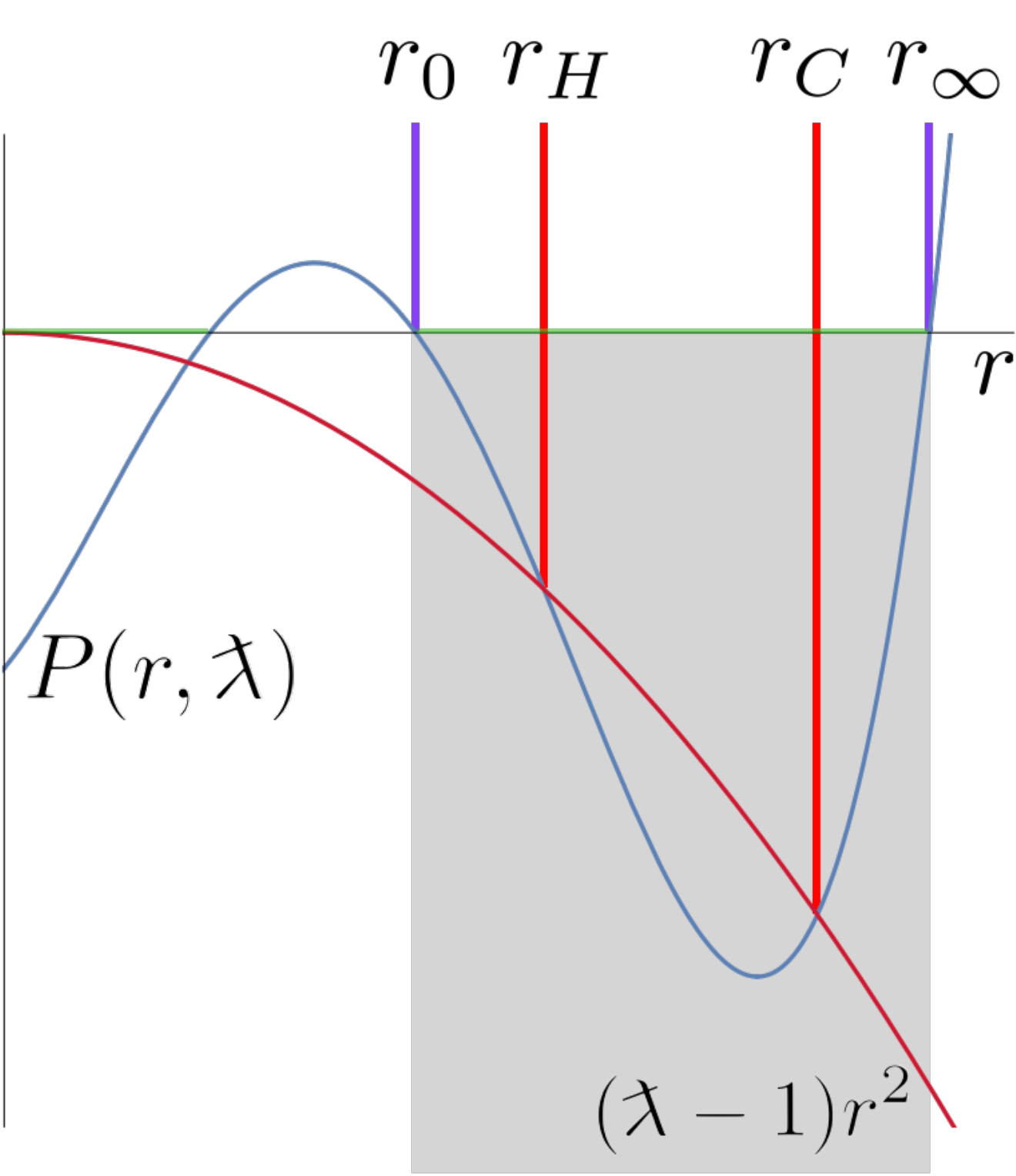}\\
    \centering
    \footnotesize{$C_{1}^{\rm BH}$}\\
    \vspace{10mm}
    \includegraphics[width=\textwidth]{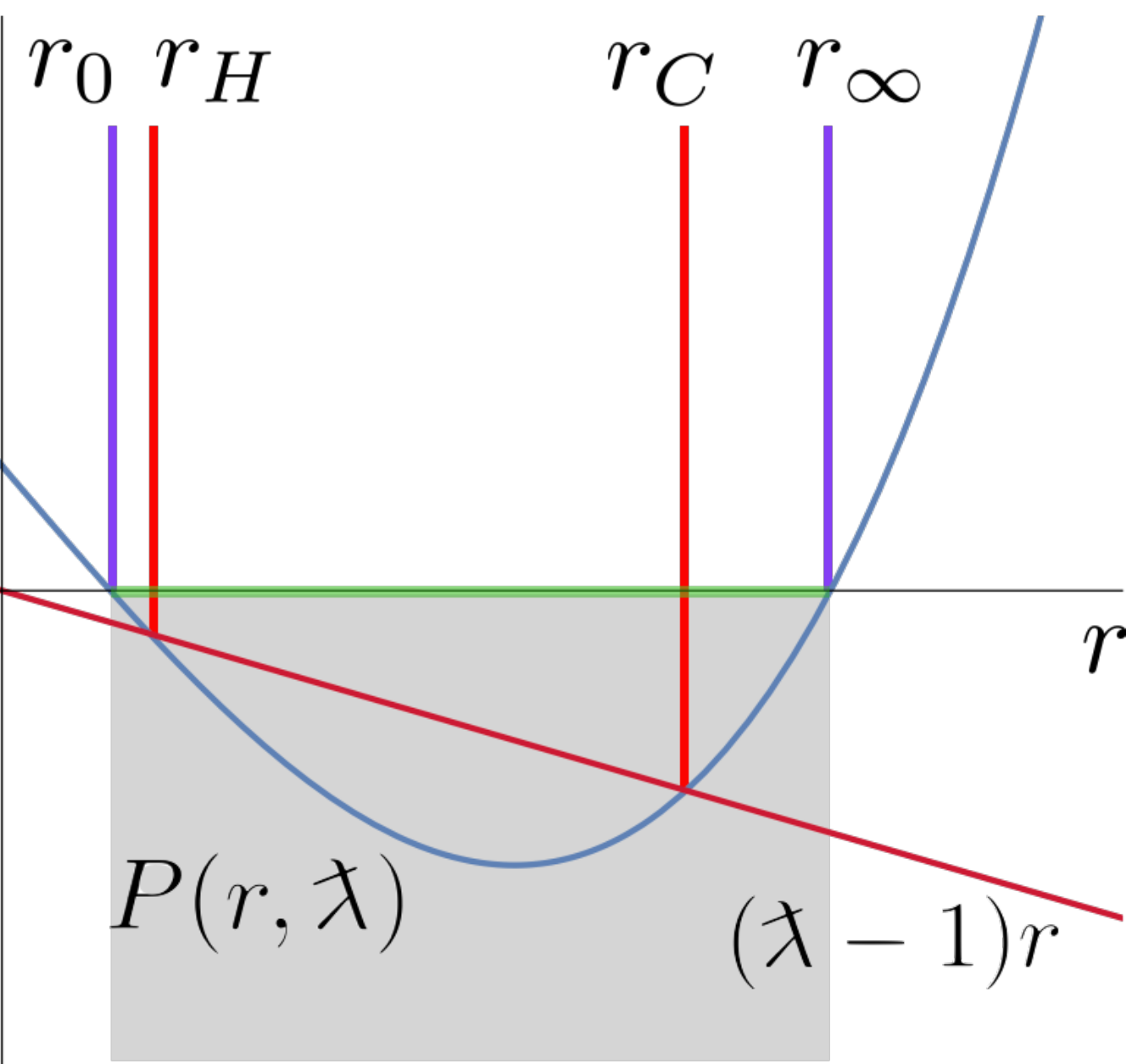}
     \footnotesize{$C_{2}^{\rm BH}$}\\
  \end{minipage}   
\caption{On the left we present the conformal diagram for the cases $C_1^{\rm BH}$ and $C_2^{\rm BH}$,
    which have the same global structure.
  The domain ${\cal U}$ (in gray)
  is an infinite and periodic stripe,
  and ${\cal M}$ is built up with infinite copies of ${\cal U}$ along all directions
  conveniently layering up around each gray ring.
   The plot at the top right corresponds to $C_1^{\rm BH}$ and the plot at the bottom right corresponds to $C_2^{\rm BH}$.
  These geometries modify
  the singular Reissner-Nordstr\"om-de Sitter spacetime in GR. The minimal hypersurface
  $r=r_0$ replaces the singularity structure present in GR, while the finite value $r_\infty$
  substitutes both $\mathcal{J}^+$ and $\mathcal{J}^-$.
  The gray rings correspond both to $i^+$ and $i^-$:
    the former for the region below and the latter for the region above.
  Radial infalling observers cross consecutive
 horizons and critical hypersurfaces in a finite proper time. Beginning in a static
  region at a given radius $r$, they will first traverse $r_H$,
  and then subsequently cross $r_0$, $r_H$, $r_C$, $r_\infty$, $r_C$, until reaching again their initial $r$
  in a different static region. This process will be repeated indefinitely.
  Accelerated observers, on the other hand, may choose to stay at their original
  static region, never crossing $r_H$ or $r_C$, and ending up in the corresponding infinite future $i^+$. 
  }
  \label{fig.C1BH}
\end{figure}
\begin{figure}
  \centering
  \includegraphics[width=0.7\textwidth]{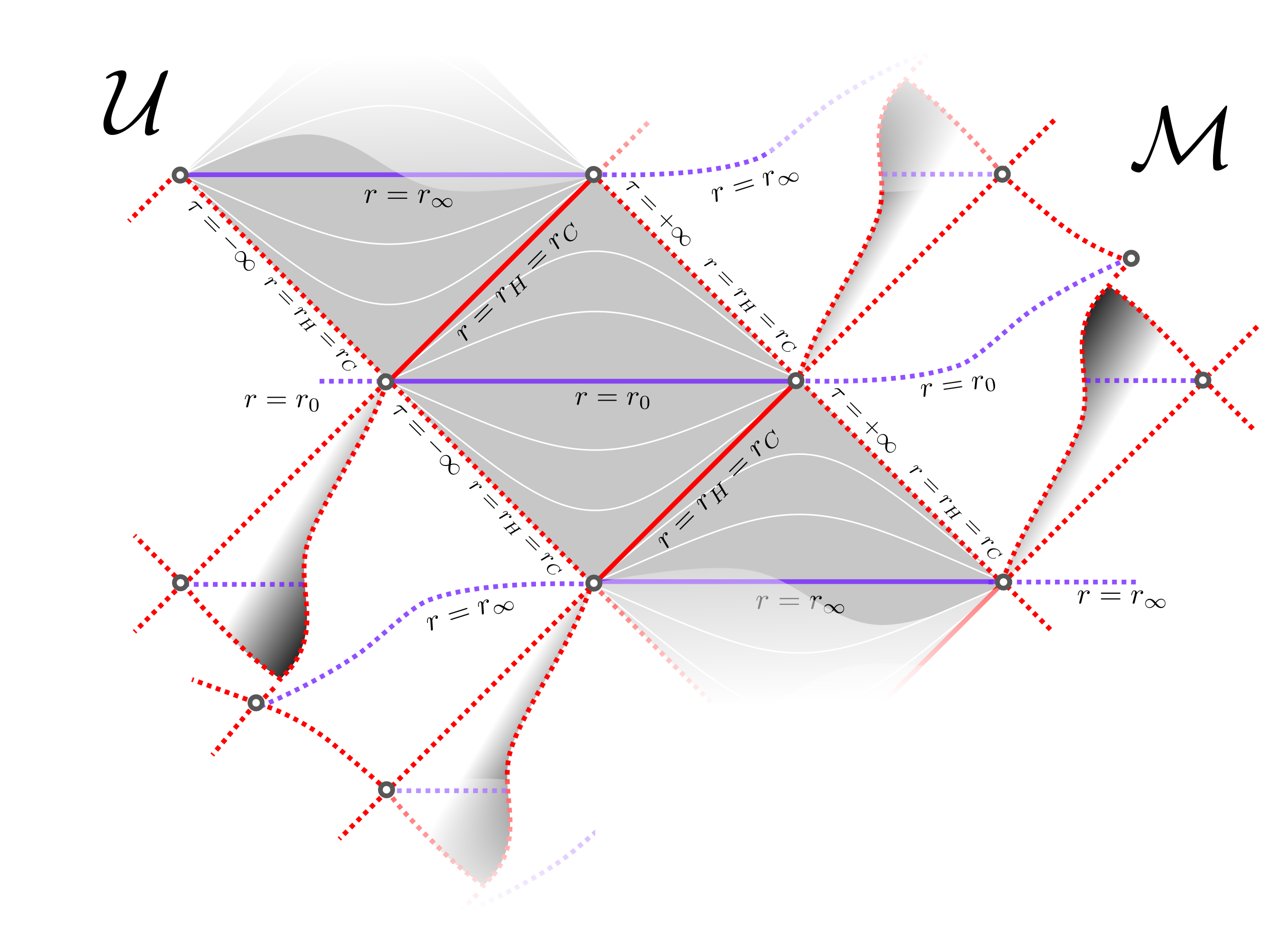}
  \hfill
  \begin{minipage}[b]{0.25\textwidth}
    \includegraphics[width=\textwidth]{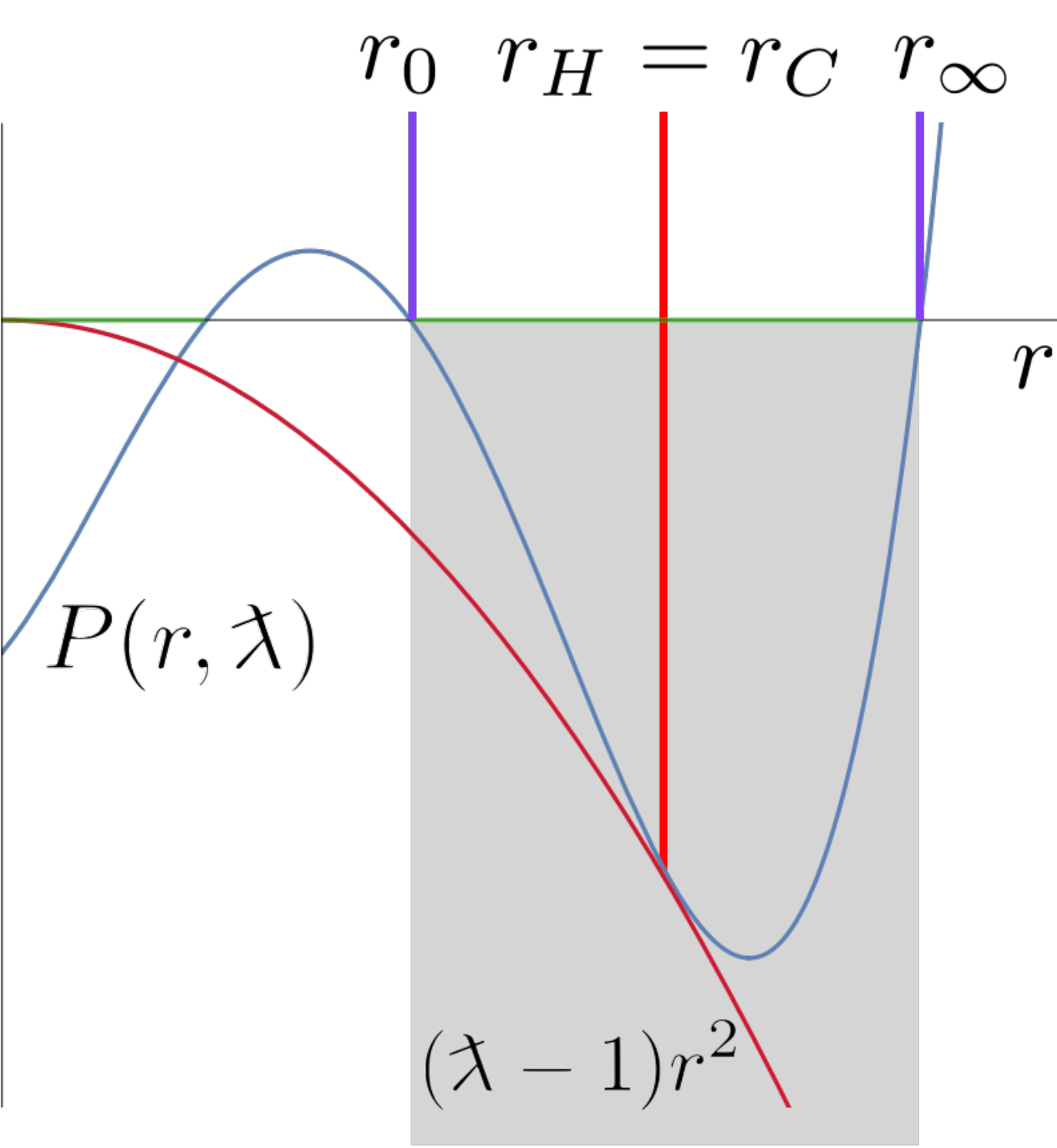}\\
    \centering
    \footnotesize{$C_{1}^{\rm extremal}$}\\
    \vspace{2mm}
    \includegraphics[width=\textwidth]{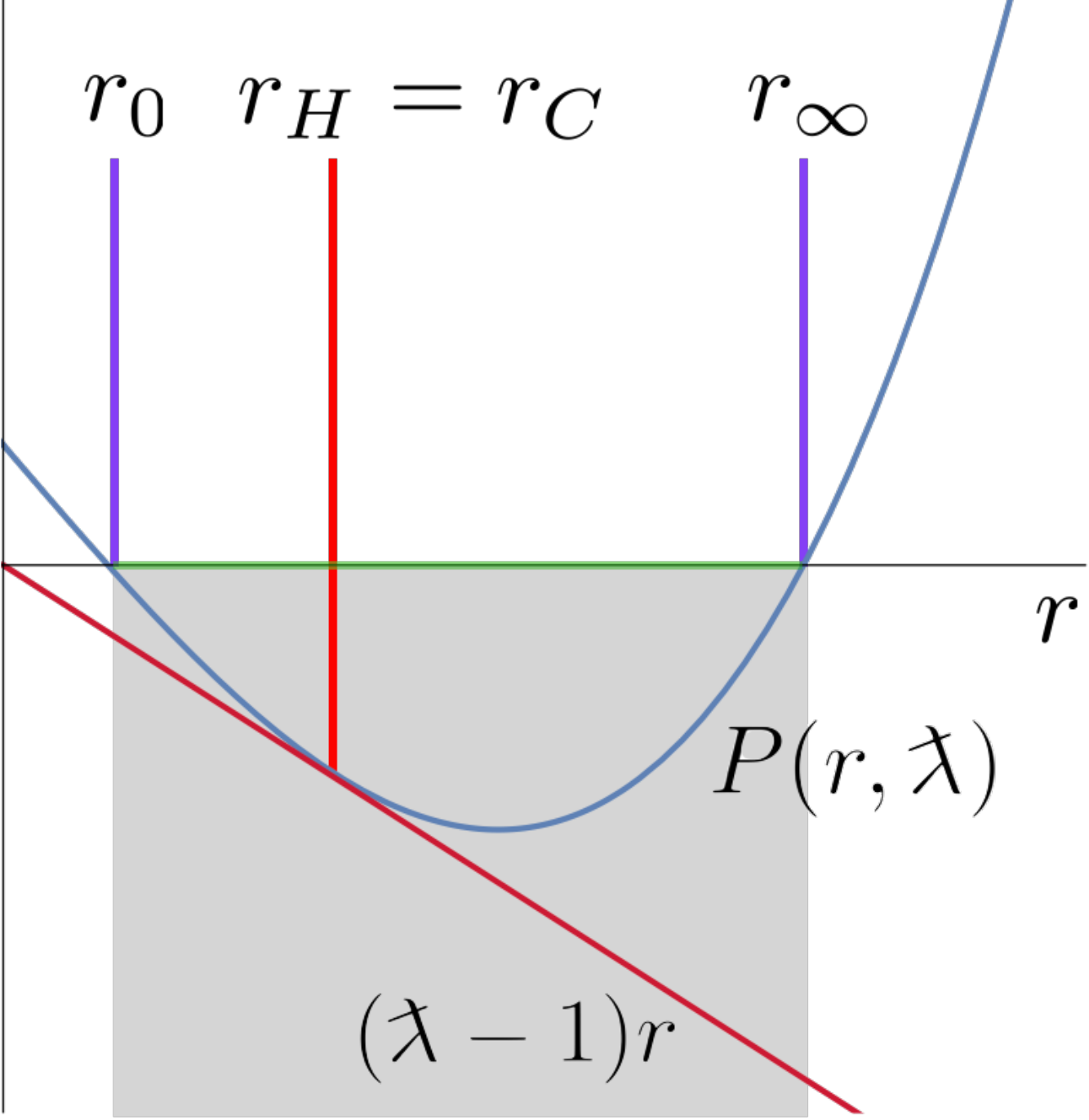}\\
    \footnotesize{$C_{2}^{\rm extremal}$}\\
  \end{minipage}   
  \caption{Cases $C_1^{\rm extremal}$ and $C_2^{\rm extremal}$.
    In these cases the horizons $r_H$ and $r_C$
    coincide, becoming a degenerate horizon bounding homogeneous regions. There are no static regions.}
  \label{fig.C1eC}
\end{figure}

\begin{figure}
  \centering
  \includegraphics[width=0.7\textwidth]{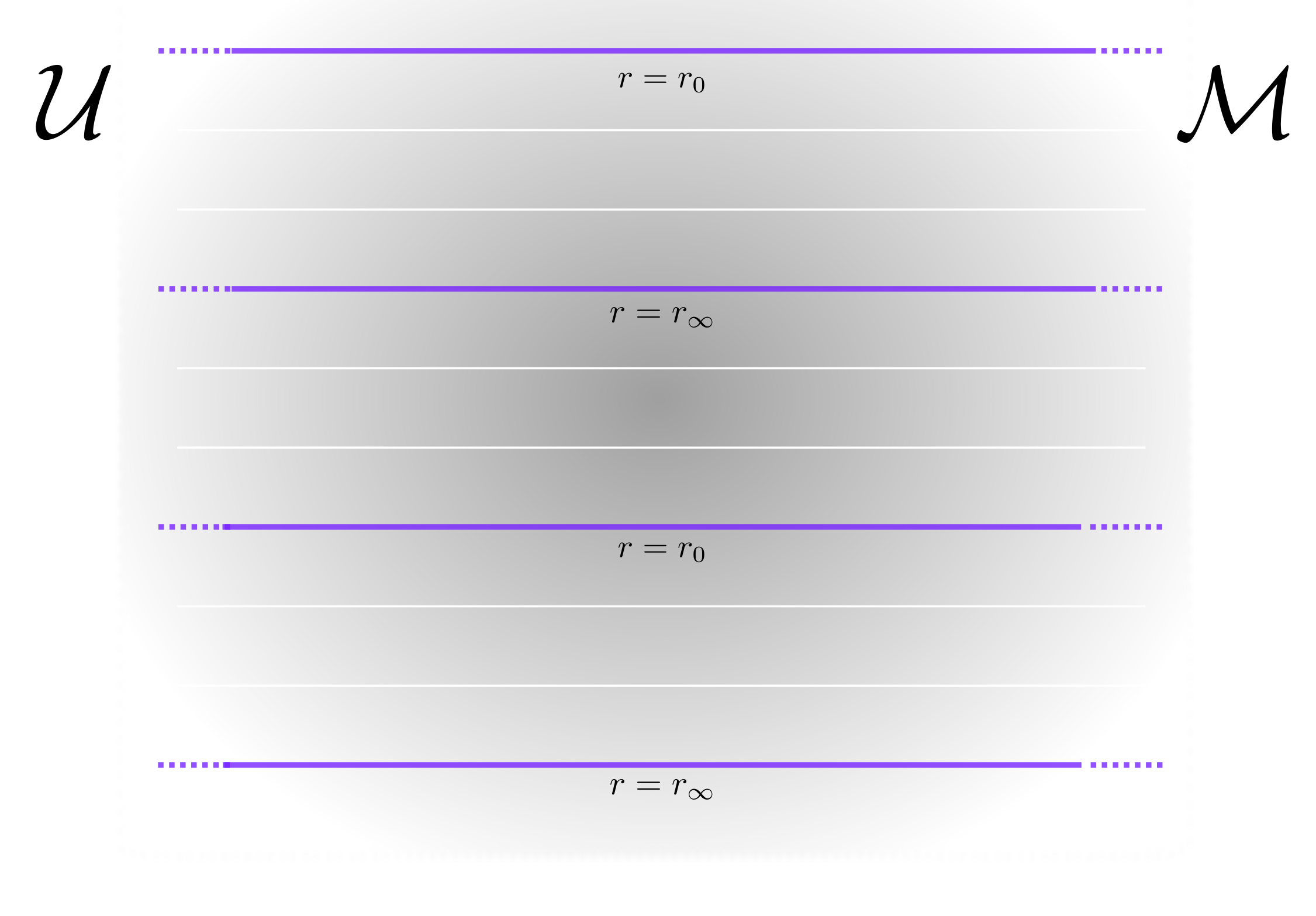}
  \hfill
  \begin{minipage}[b]{0.25\textwidth}
    \includegraphics[width=\textwidth]{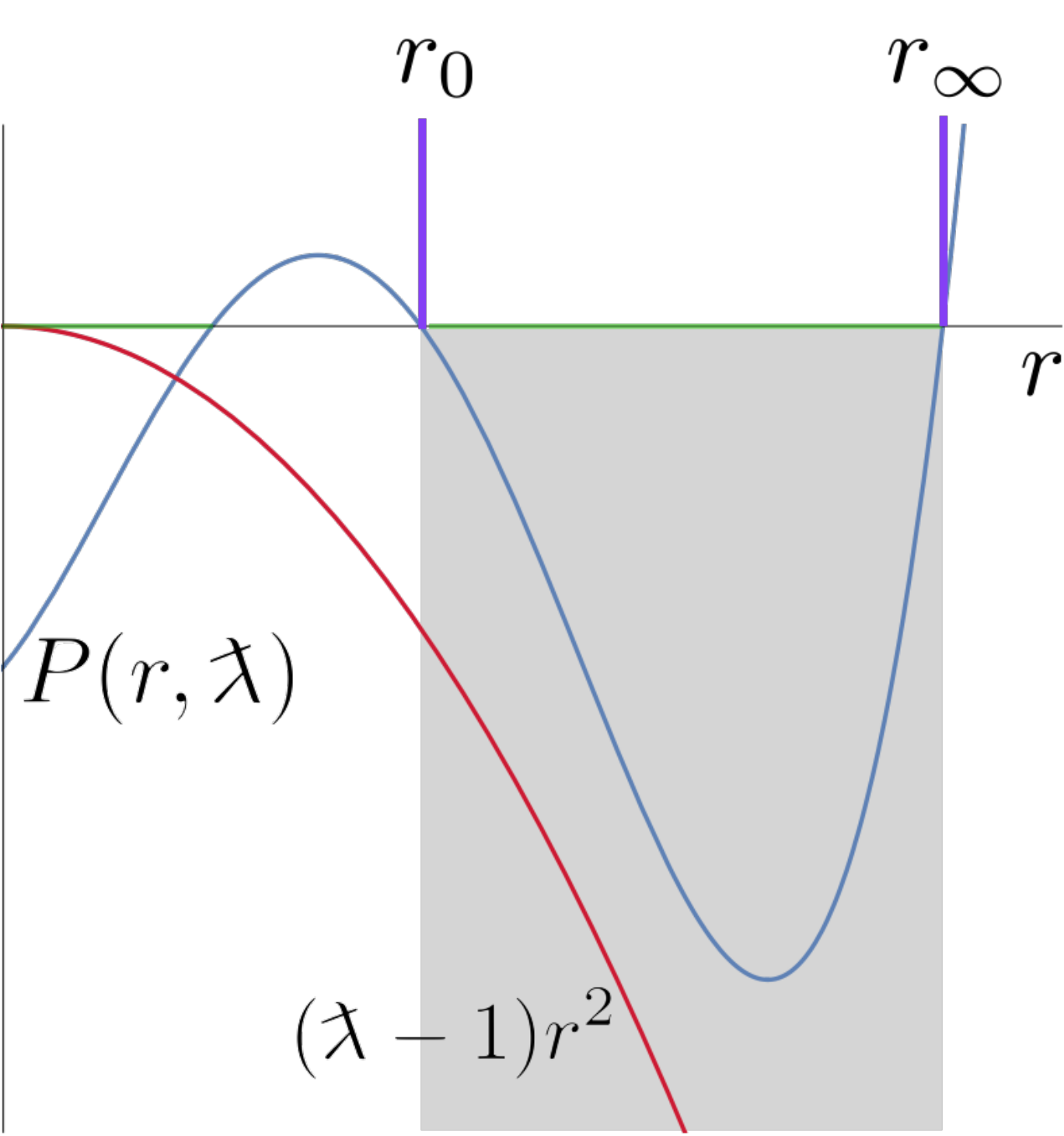}\\
    \centering
    \footnotesize{$C_{1}^{\rm cosmos}$}\\
    \vspace{2mm}
    \includegraphics[width=\textwidth]{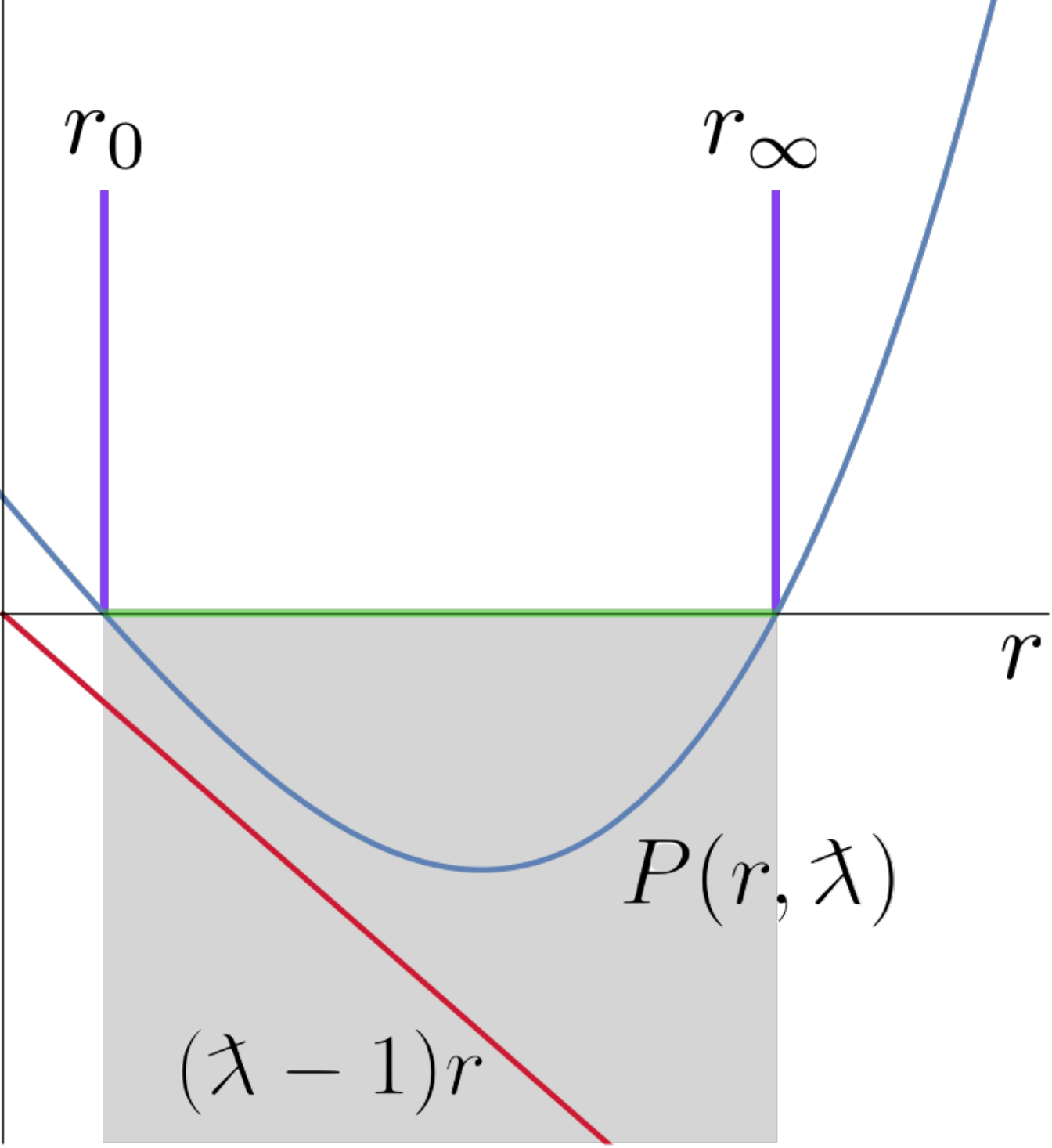}\\
    \footnotesize{$C_{2}^{\rm cosmos}$}\\
  \end{minipage}   
  \caption{Cases $C_1^{\rm cosmos}$ and $C_2^{\rm cosmos}$. There are no horizons.
    These cases represent cyclic (Kantowski-Sachs) cosmologies
    that expand to a hypersurface foliated
    by spheres of area $4\pi r_\infty^2$,
    and then contract to a hypersurface foliated by spheres of area
    $4\pi r_0^2$.}
  \label{fig.C1C}
\end{figure}

\begin{figure}
  \centering
  \includegraphics[width=0.7\textwidth]{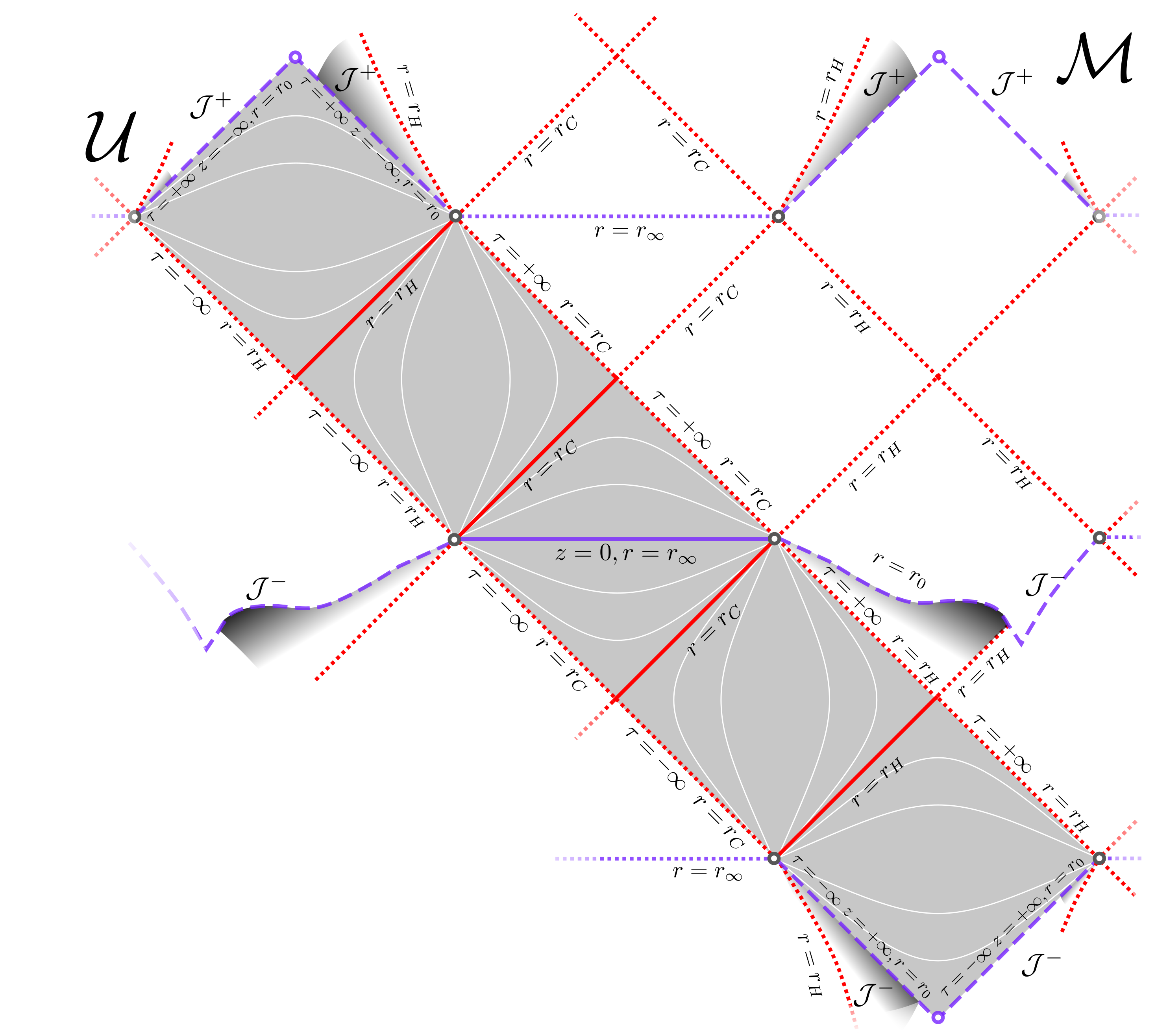}
  \hfill
  \begin{minipage}[b]{0.25\textwidth}
    \includegraphics[width=\textwidth]{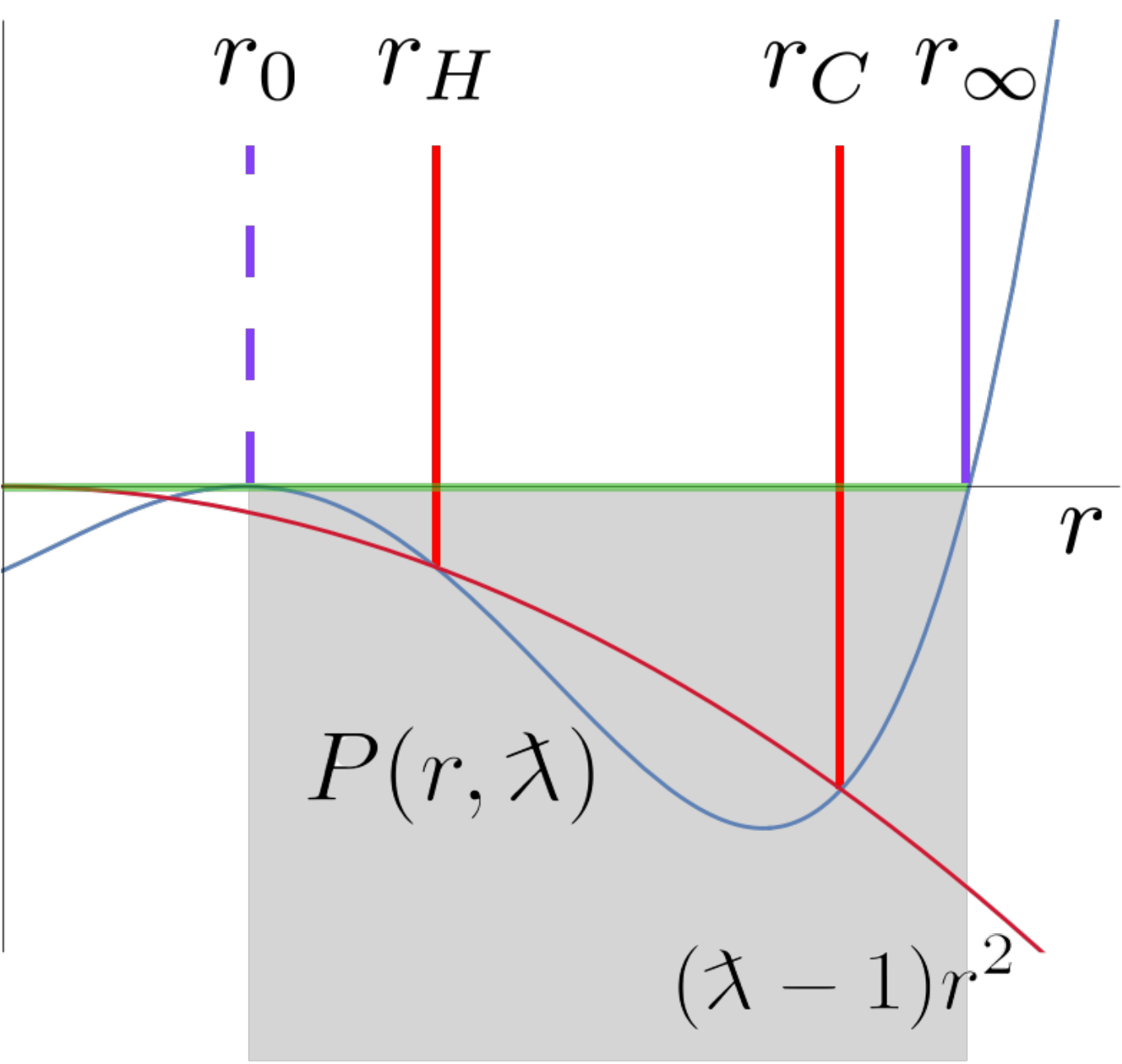}
  \end{minipage}   
  \caption{Case $D_1^{\rm BH}$. This diagram is similar to that of the case $C_1^{\rm BH}$,
    in Fig.\ref{fig.C1BH}.
      The difference is that now the diagram for $\mathcal{U}$ is a finite
    rectangle, because the locations defined by $r=r_0$ correspond to
      past (for $\tau\to-\infty,z\to+\infty$) or future
      (for $\tau\to +\infty, z\to -\infty$)
      null infinity. Observe that these infinities are not approached
      as $r$ goes to infinity, but as the spacelike homogeneous leaves
      tend to (but never reach) a minimal hypersurface foliated, in turn, by spheres of radius
      $r_0$.
      The origin of the coordinate $z$ has been chosen so that $r(0)=r_\infty$
    (see Sec.\ref{sec_existence_rinfinity}).
  }
  \label{fig.D1}
\end{figure}

\begin{figure}
  \centering
  \includegraphics[width=0.6\textwidth]{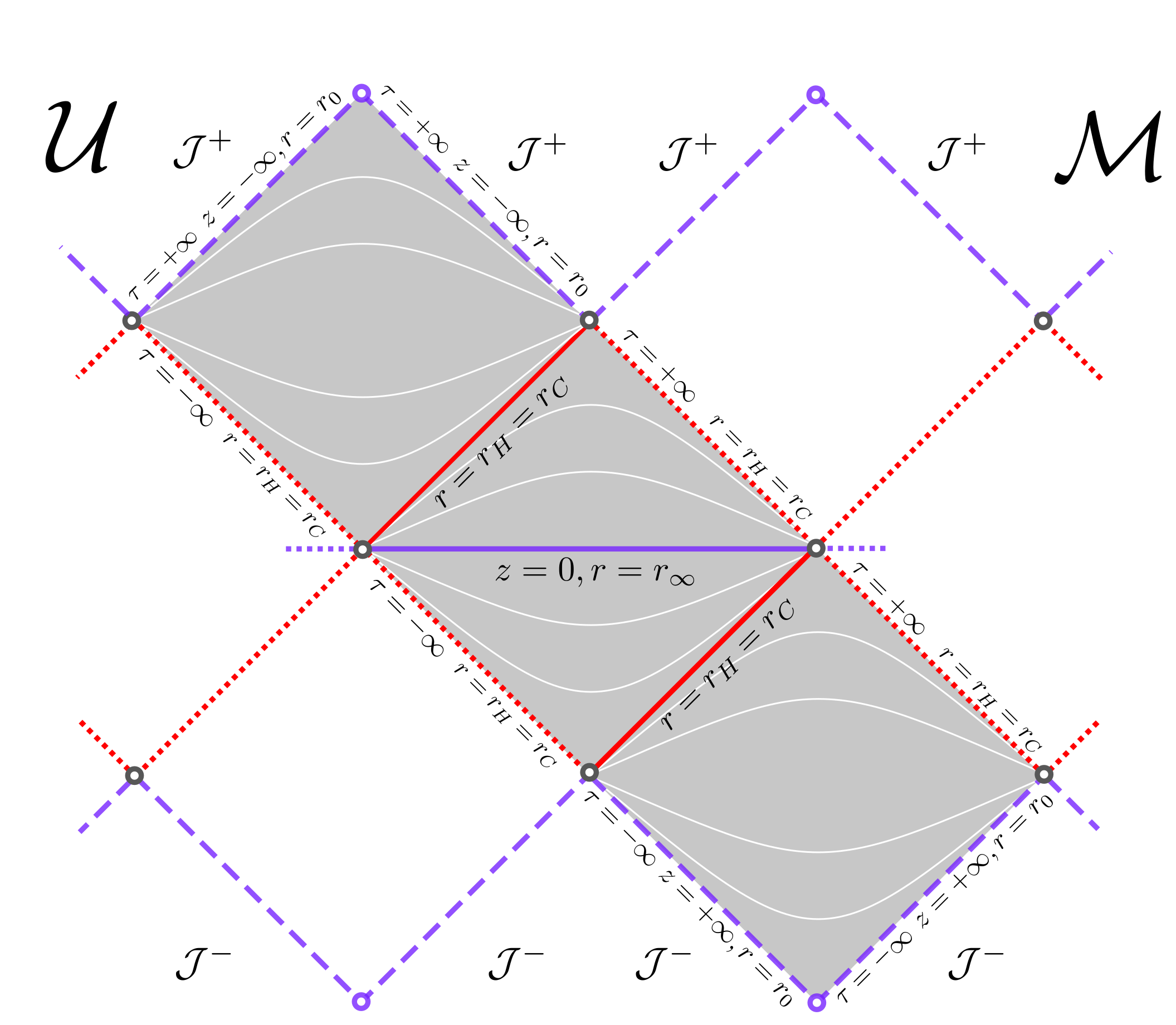}
  \hfill
  \begin{minipage}[b]{0.25\textwidth}
    \includegraphics[width=\textwidth]{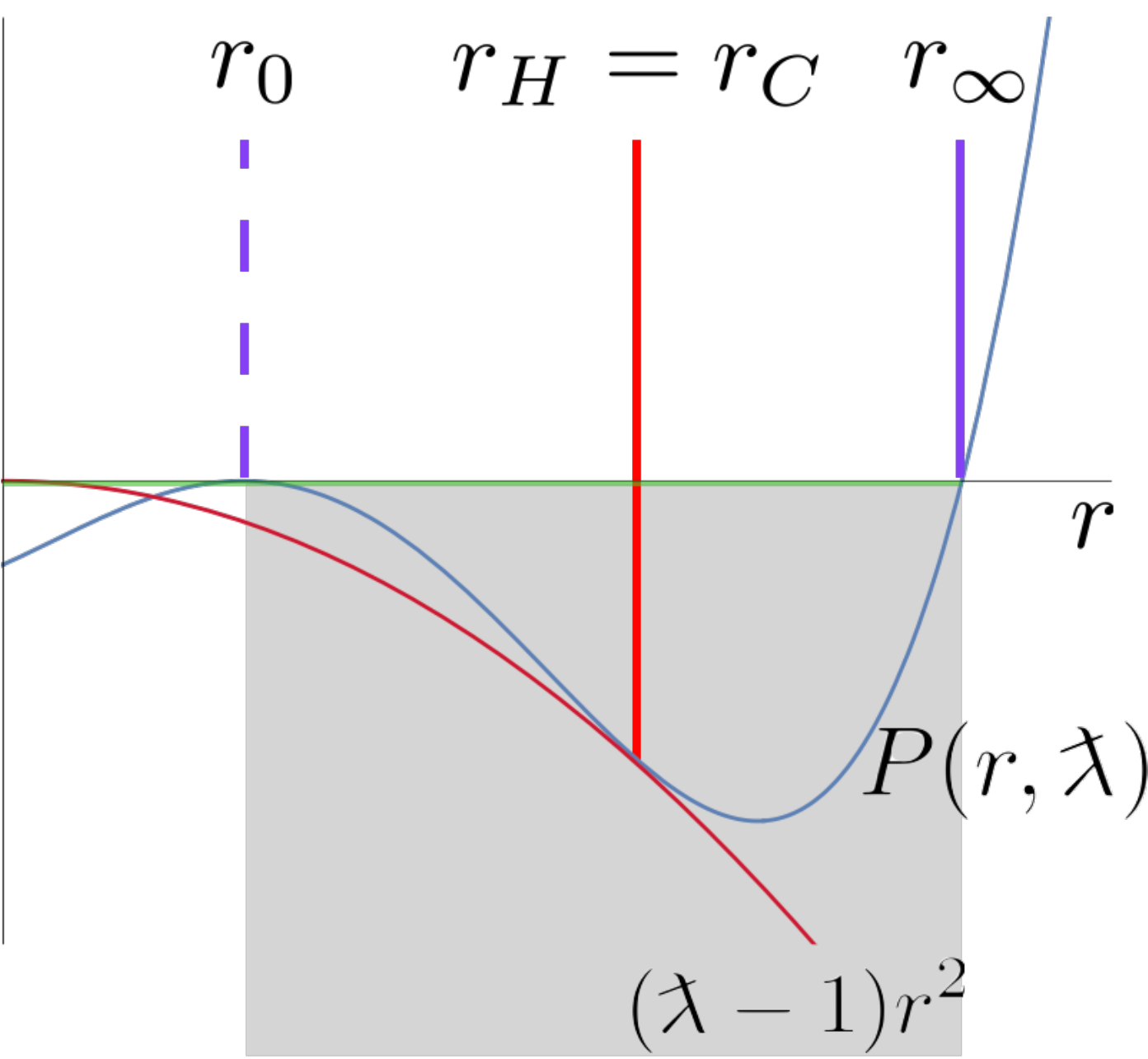}
  \end{minipage}   
  \caption{Case $D_1^{\rm extremal}$. Diagram similar to Fig.\ref{fig.D1},
      but now the horizons $r_H$ and $r_C$ degenerate, and bound homogeneous regions.
    }
  \label{fig.D1eC}
\end{figure}

\begin{figure}
  \centering
  \includegraphics[width=0.45\textwidth]{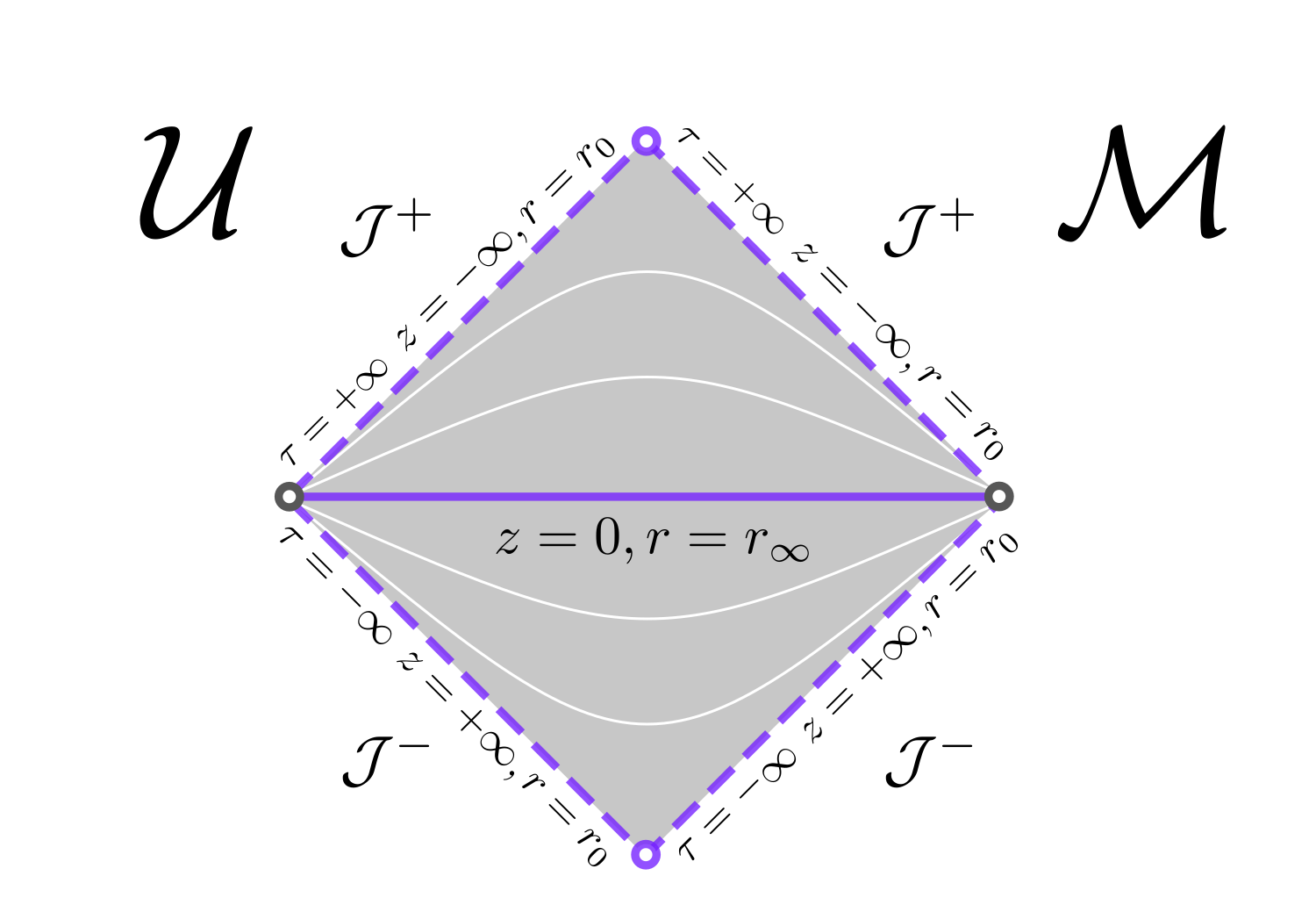}
  \hfill
  \begin{minipage}[b]{0.25\textwidth}
    \includegraphics[width=\textwidth]{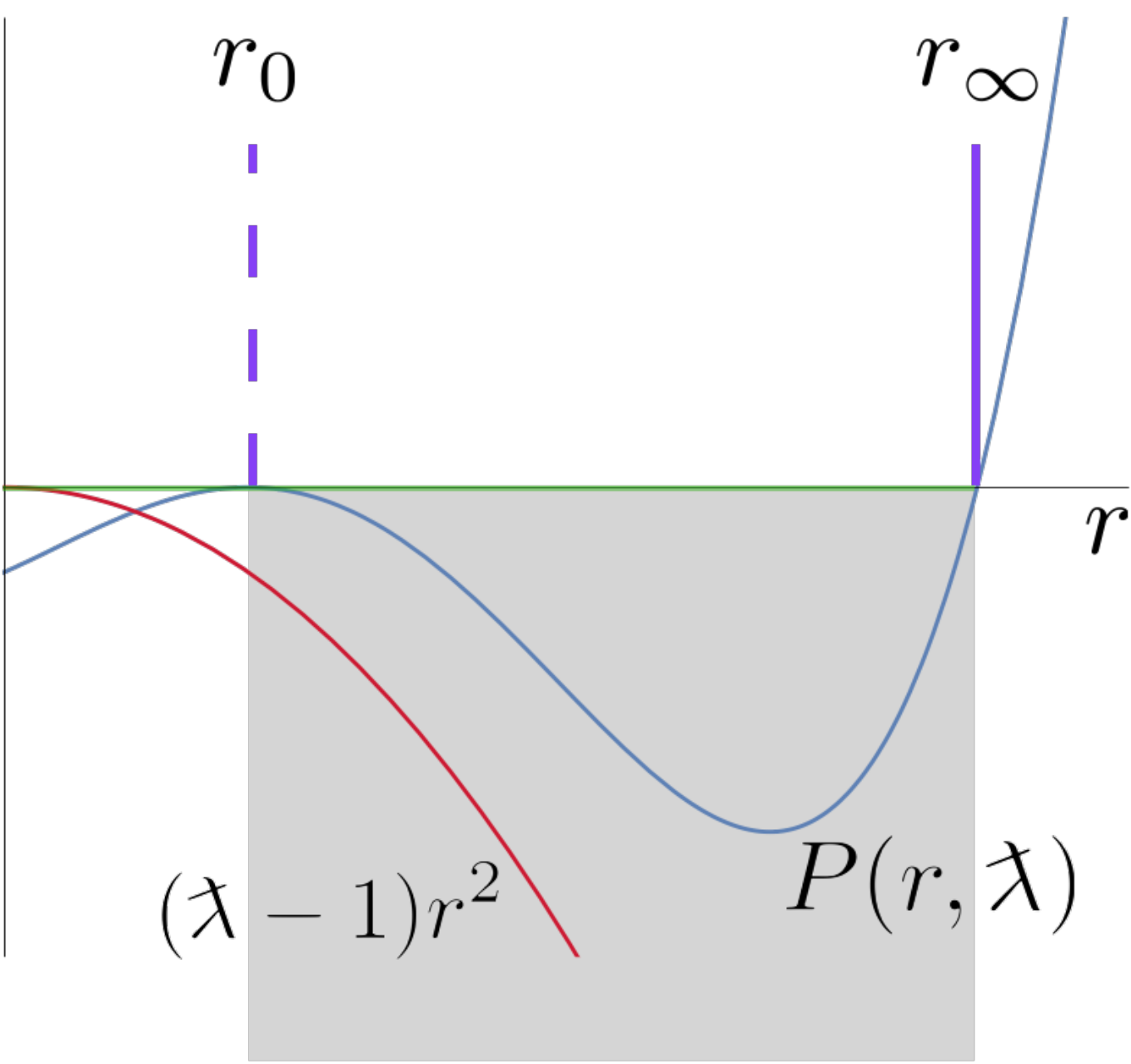}
  \end{minipage}   
  \caption{Case $D_1^{\rm cosmos}$. As in Fig.~\ref{fig.C1C}, there are no horizons.
      The difference with that case is that $r=r_0$ corresponds to past and future null infinities,
      as in Fig.\ref{fig.D1}.}
  \label{fig.D1C}
\end{figure}

\begin{figure}
  \centering
  \includegraphics[width=0.6\textwidth]{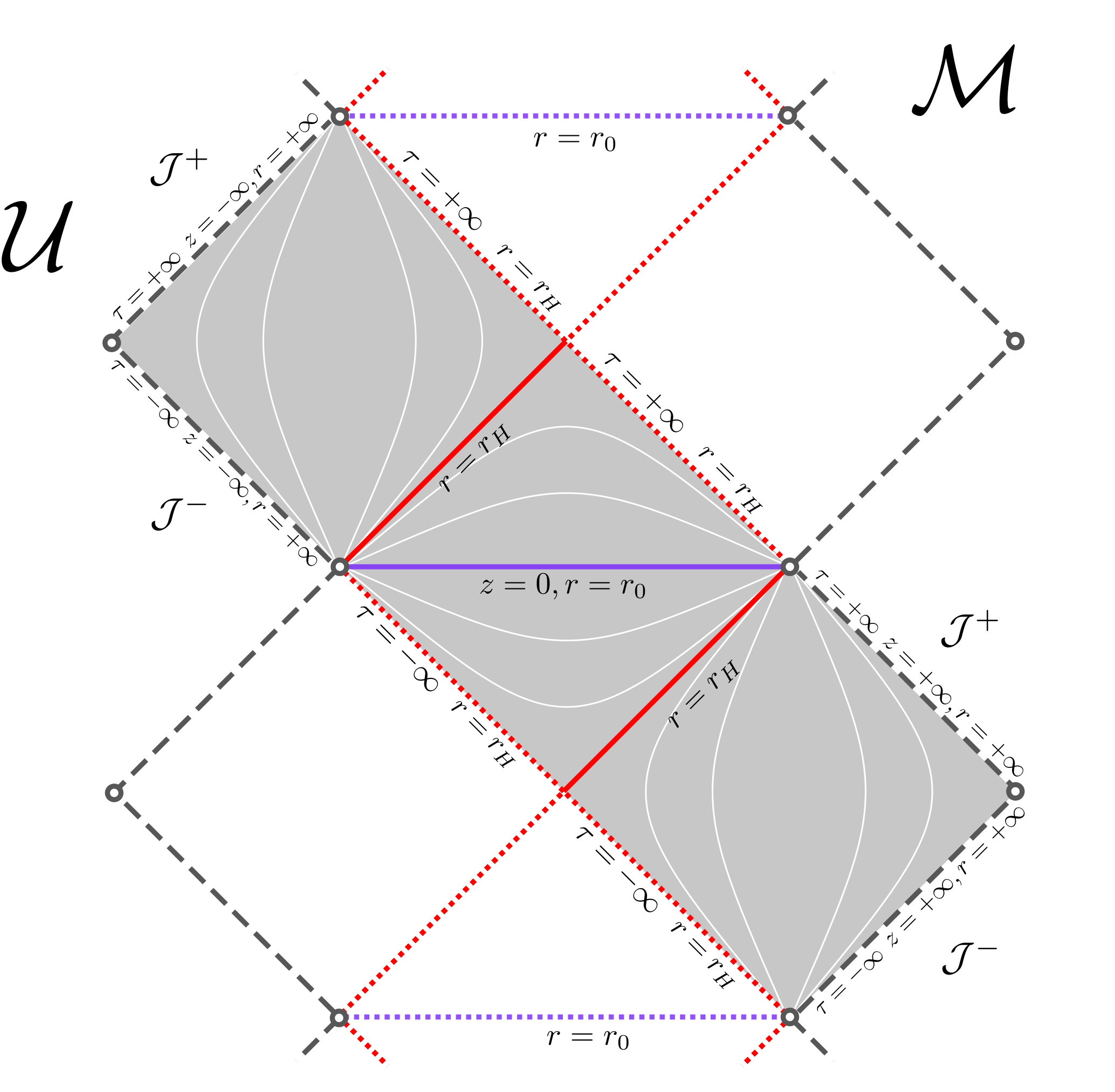}
  \hfill
  \begin{minipage}[b]{0.25\textwidth}
    \includegraphics[width=\textwidth]{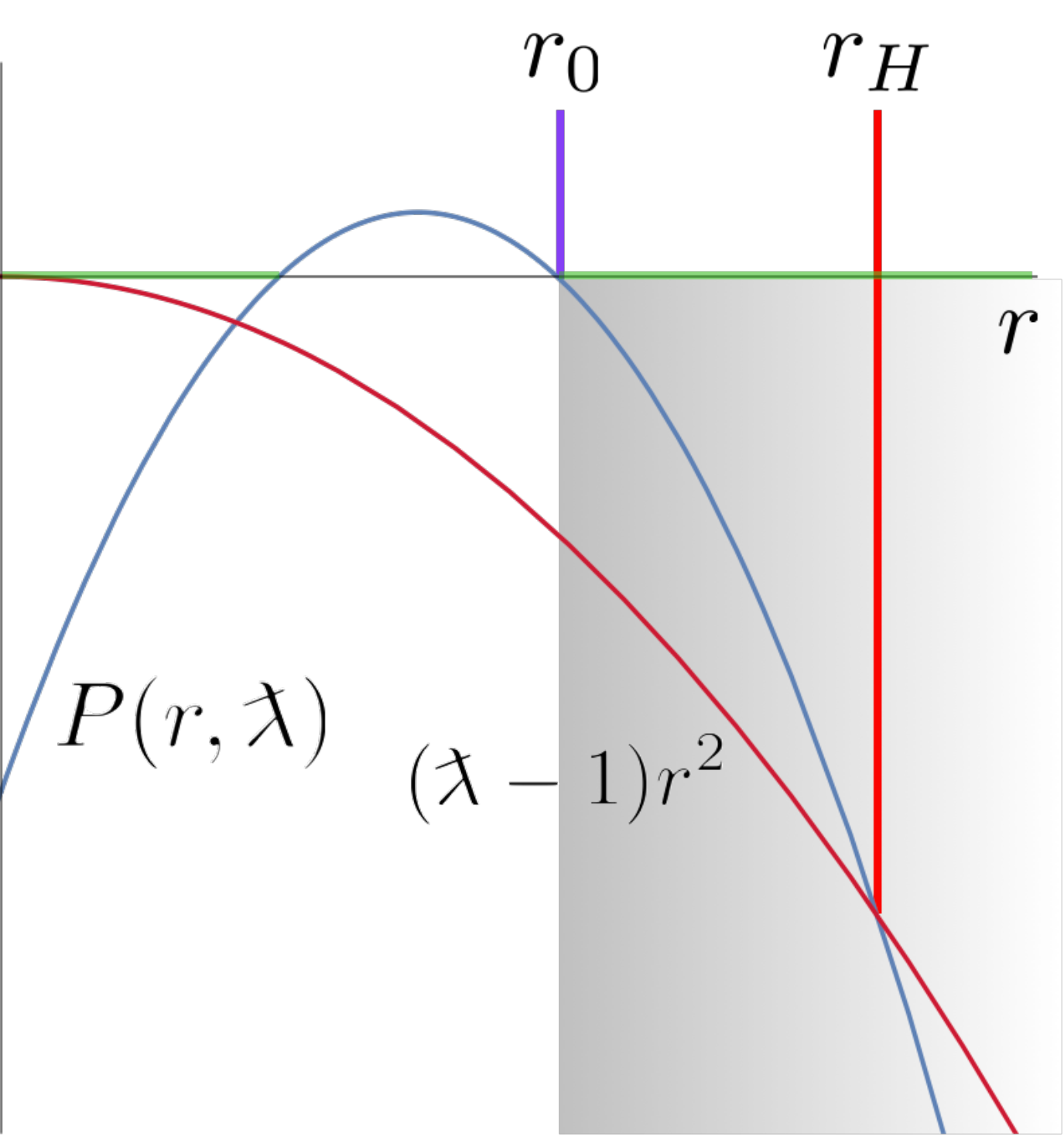}\\
    \centering
    \footnotesize{$C_{3}$, $Q\neq 0$}\\
    \vspace{2mm}
    \includegraphics[width=\textwidth]{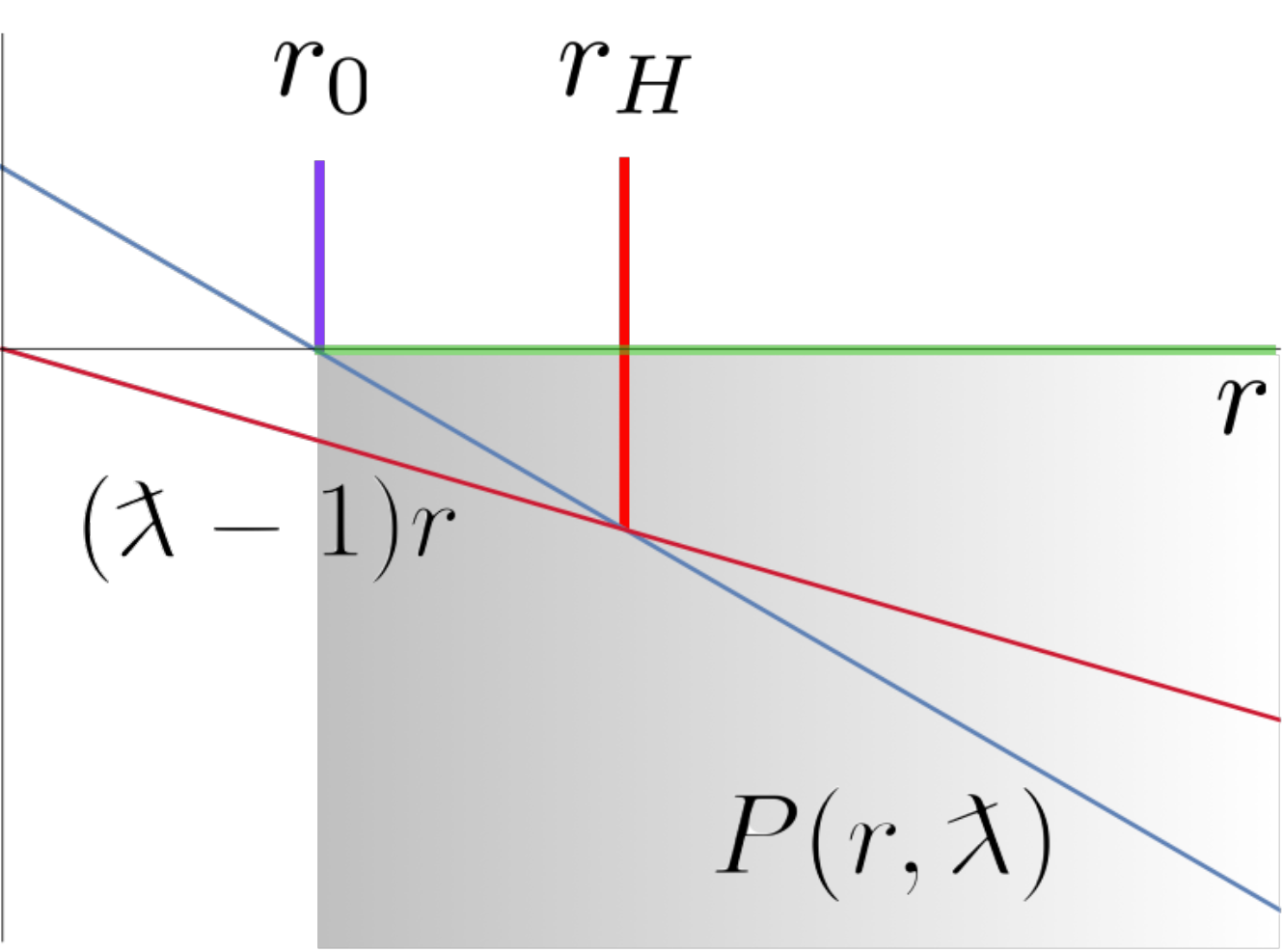}\\
    \footnotesize{$C_{3}$, $Q = 0$}\\
  \end{minipage}   
  \caption{Case $C_3$, which includes both $Q\neq 0$ (top-right plot) and $Q=0$ (bottom-right plot).
    Since $\Lambda=0$ in this case, we recover the usual asymptotic flat ends.
    We recover the conformal diagram for vacuum (with $M>0$ and $\Lambda=Q=0$)
    shown in Refs. \cite{Alonso-Bardaji:2021yls,Alonso-Bardaji:2022ear}.}
  \label{fig.C3}
\end{figure}

\begin{figure}
  \centering
  \includegraphics[width=0.6\textwidth]{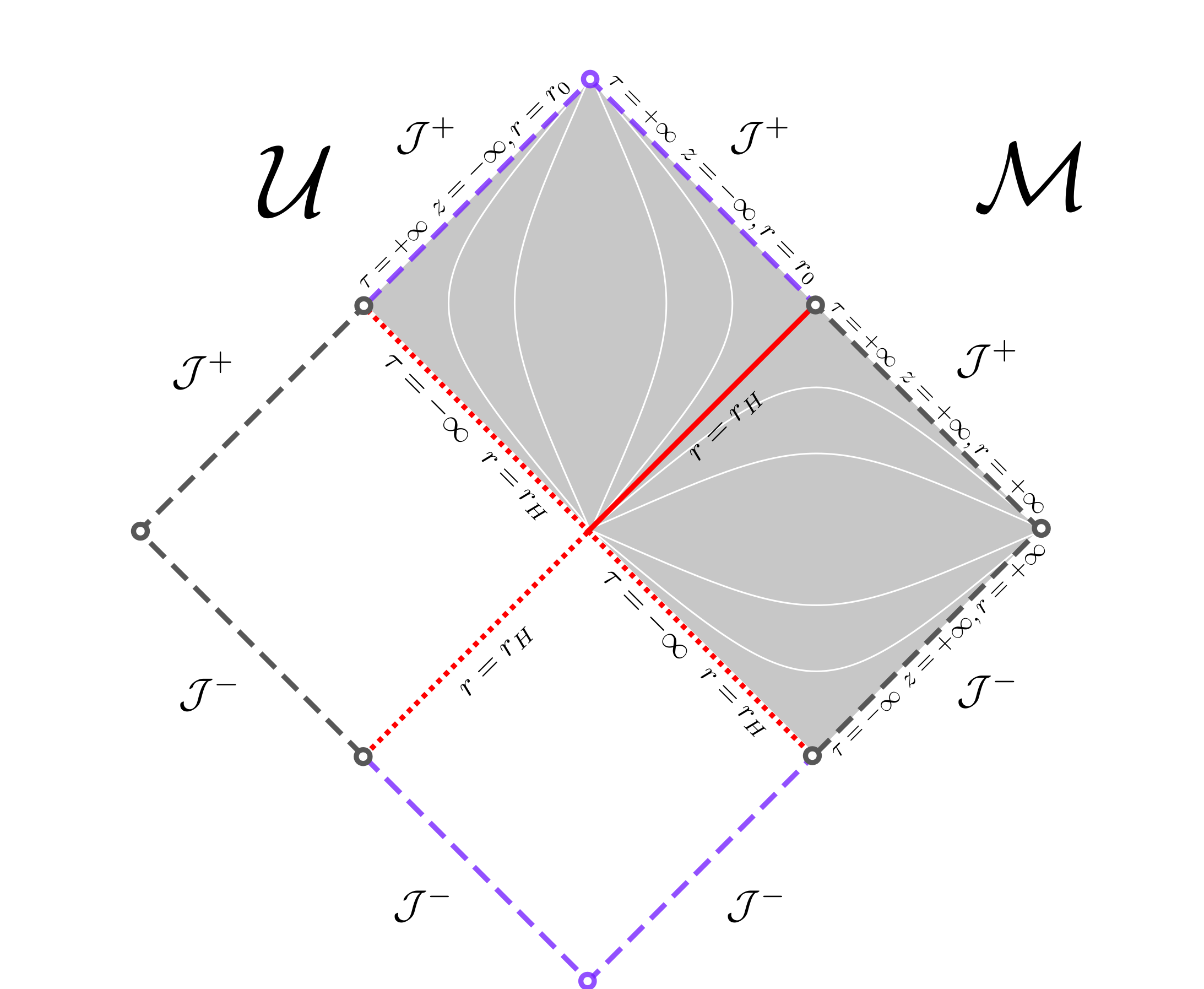}
  \hfill
  \begin{minipage}[b]{0.25\textwidth}
    \includegraphics[width=\textwidth]{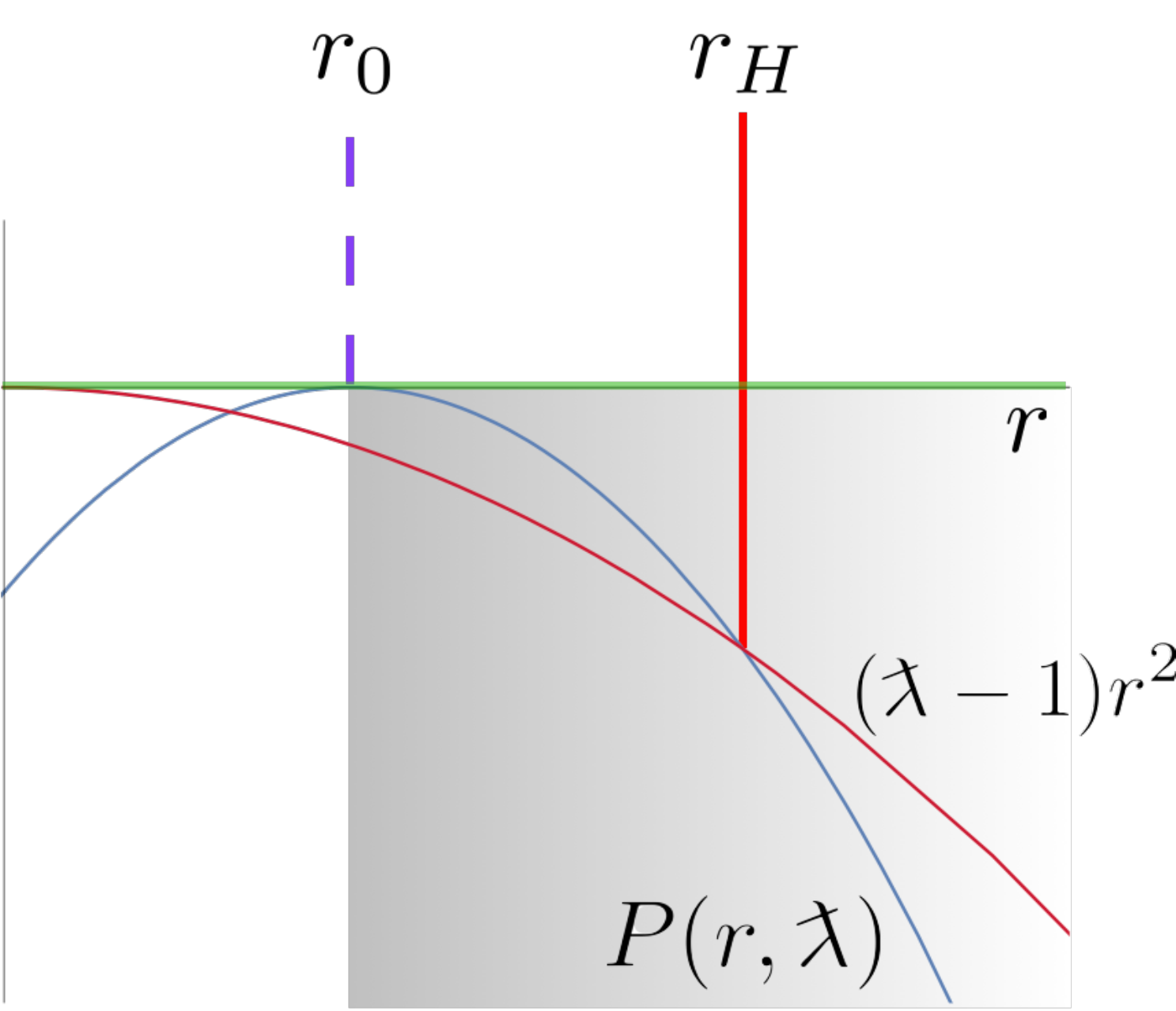}
  \end{minipage}   
  \caption{Case $D_3$. This is the limit of case $C_3$, depicted
      in Fig.\ref{fig.C3},
      with $Q=\sqrt{\lbar} M$, so that $r=r_0$ becomes null infinity (see Fig.\ref{fig.D1}).}
  \label{fig.D3}
\end{figure}

\section{Correspondence between coordinate and gauge transformations}\label{app.covariance}

In this Appendix we show that Eq. \eqref{eq.gaugetransf} is satisfied.
On the one hand, the Poisson brackets between $(1/F)$ and the generators of gauge transformations read
\begin{align*}
    \{ (1/F),H[\epsilon^0]+D[\epsilon^x]\}&=\,(1/F)\Bigg(
2\epsilon^x{}'+\epsilon^x\left(2\lambda\tan(\lambda\kang){\kang'}+4\frac{\ephi'}{\ephi}-\frac{\erad'}{\erad}-\frac{8\ephi\ephi'+2\lambda^2\erad'\erad''}{4\ephi^2+(\lambda\erad')^2}\right)\nonumber\\
+&\epsilon^0\Bigg(\frac{8\sqrt{\erad}\krad}{\sqrt{1+\lambda^2}\ephi}\left(\frac{\cos(2\lambda\kang)}{2}+\left(\frac{\lambda\erad'}{2\ephi}\right)^2\cos^2(\lambda\kang)\right)\nonumber\\&-\frac{\tan(\lambda\kang)}{\sqrt{1+\lambda^2}\lambda\sqrt{\erad}}\left(1+\lambda^2\left(1-\Lambda\erad-\frac{Q^2}{\erad}\right)\right)\nonumber\\+&\frac{\sqrt{\erad}\sin(2\lambda\kang)}{\sqrt{1+\lambda^2}\lambda\ephi^2}\left(\frac{\ephi^2}{2\erad}+\lambda^2\left(\frac{(\erad')^2}{8\erad}+\erad''-\frac{\erad'\ephi'}{\ephi}\right)\right)-\frac{\lambda^2\sqrt{\erad}\erad'\kang'}{\sqrt{1+\lambda^2}\ephi^2}\Bigg)\Bigg).
\end{align*}
On the other hand, the derivatives of $(1/F)$ have the form 
\begin{align*}
    (1/F)\dot{} &= (1/F)\left(\frac{4\dot{\ephi}}{\ephi}-\frac{\dot{\erad}}{\erad}-\frac{8\ephi\dot{\ephi}+2\lambda^2\erad'\dot{\erad'}}{4\ephi^2+\lambda^2(\erad')^2}+2\lambda\tan(\lambda\kang)\dot{\kang}\right),\\
    (1/F)' &=(1/F)\left(\frac{4{\ephi'}}{\ephi}-\frac{\erad'}{\erad}-\frac{8\ephi\ephi'+2\lambda^2\erad'{\erad''}}{4\ephi^2+\lambda^2(\erad')^2}+2\lambda\tan(\lambda\kang){\kang'}\right).
\end{align*}
Making use of the above expressions and the equations of motion \eqref{eomq1}, \eqref{eomq2}, and \eqref{eomp2},
it is now a straightforward computation to show that the equality \eqref{eq.gaugetransf} holds.

\section{Equations of motion and derivation of the solutions}\label{app.equations}

The equations of motion read as follows:
\begin{subequations}\label{eom}
\begin{align}
\label{eomq1}
  \dot{\erad}=&\{\erad,D[\shift]+H[\lapse]\}= \shift \erad'+\lapse\sqrt{{\erad}}\frac{\sin(2\lambda \kang)}{\lambda\sqrt{1\!+\!\lambda^2}}\left(1+\left(\frac{\lambda \erad'}{2\ephi}\right)^{\!2}\right),
\\\label{eomq2}
    \dot{\ephi}=&\{\ephi,D[\shift]+H[\lapse]\}= \left(\shift \ephi\right)' +2\lapse\sqrt{\erad}\krad\frac{\cos(2\lambda \kang)}{\sqrt{1\!+\!\lambda^2}}\left(1+\left(\frac{\lambda \erad'}{2\ephi}\right)^{\!2}\right)\nonumber\\
              &+\lapse\frac{\sin(2\lambda \kang)}{\lambda\sqrt{1\!+\!\lambda^2}}\left(\frac{\ephi}{2\sqrt{\erad}} +
                \frac{\lambda^2}{2}\bigg(\frac{\erad'}{2\ephi}\left(\sqrt{\erad}\right)'
              +\sqrt{\erad}\left(\frac{\erad'}{\ephi}\right)'\bigg)
                \right),
\\\label{eomp1}
  \dot{\krad}=&\{\krad,D[\shift]+H[\lapse]\}= \left(\shift \krad\right)' +\lapse''\frac{\sqrt{\erad}  \cos ^2(\lambda \kang)}{2 \sqrt{1\!+\!\lambda^2} \ephi}+\frac{\lapse\ephi}{4\sqrt{1+\lambda^2}\sqrt{\erad}}\left({\Lambda}-3\left(\frac{Q}{\erad}\right)^2\right)\nonumber \\
              &+ \frac{\lapse'\sqrt{\erad}}{2\sqrt{1\!+\!\lambda^2}\ephi^2} \Bigg(\lambda\sin (2 \lambda \kang)\left(\erad'\krad-2\ephi\kang'\right) +\cos ^2(\lambda \kang)\left(\frac{\ephi\erad'}{2 {\erad} }-\ephi'\right)\!\!\Bigg)\nonumber\\
              &+\frac{\lapse}{\sqrt{1\!+\!\lambda^2}} \Bigg(\frac{\ephi (\sin ^2(\lambda \kang)+\lambda^2)}{4 \lambda^2  \erad^{3/2}} +\frac{\cos^2(\lambda\kang)}{4\sqrt{\erad}\ephi}\left(\erad''-\frac{(\erad')^2}{4\erad}-\frac{\erad'\ephi'}{\ephi}\right)\nonumber\\
              &-\frac{\krad\sin(2\lambda\kang)}{2\lambda\sqrt{\erad}}\left(1+\left(\frac{\lambda\erad'}{2\ephi}\right)^2\right)-\left[\sin(2\lambda\kang)\frac{\lambda\sqrt{\erad}}{2\ephi^2}\diff\right]'\Bigg),
\\\label{eomp2}
    \dot{\kang}=&\{\kang,D[\shift]+H[\lapse]\}= \shift \kang' +\lapse' \frac{\sqrt{\erad}\erad'}{2\ephi^2}\frac{\cos^2(\lambda \kang)}{\sqrt{1\!+\!\lambda^2}} -\lapse \frac{\sin^2(\lambda \kang)+\lambda^2}{2\lambda^2\sqrt{{\erad}}\sqrt{1\!+\!\lambda^2}}\nonumber\\
    &+\lapse\frac{(\erad')^2}{8\sqrt{\erad}\ephi^2}\frac{\cos^2(\lambda \kang)}{\sqrt{1\!+\!\lambda^2}} -\lapse\frac{\sin(2\lambda \kang)}{\sqrt{1\!+\!\lambda^2}}\frac{\lambda \sqrt{\erad}\erad'}{2\ephi^3}\diff+\frac{\lapse\sqrt{\erad}}{2\sqrt{1+\lambda^2}}\left(\Lambda+\left(\frac{Q}{\erad}\right)^2\right),
\end{align}
\end{subequations}
in combination with the constraint equations $\diff=0$ and $\ham=0$,
with the definitions \eqref{D} and \eqref{H}.

\subsection{Static gauge}\label{app.static}

We partially
fix the gauge freedom by choosing $\dot{\erad}=0$ and $\sin(\lambda\kang)=0$.
Observe that this implies $\cos(2\lambda\kang)=1$.
Equation \eqref{eomq1} indicates we have two main cases
depending on whether or not $\erad'$ vanishes identically.
We start by assuming that $\erad'$
does not vanish identically, so that $\shift=0$ necessarily.
In addition, the vanishing of the diffeomorphism constraint $\diff=0$, cf. \eqref{D},
requires $\krad=0$. The remaining equations read
\begin{subequations}
\begin{align}
 0&= \dot{\ephi},\label{ephizero}\\
    0&=\dot{\krad}=  \frac{\lapse}{2\sqrt{1+\lambda^2}\sqrt{\erad} }  \left( \left(\frac{\erad'}{2\ephi}\right)'-\frac{(\erad')^2}{8\erad\ephi}
     +\frac{\ephi}{2}\left(\frac{1}{{\erad}}+{\Lambda}-3\left(\frac{Q}{{\erad}}\right)^2\right)
     \right)\label{kxdot}+\left(\frac{\lapse'\sqrt{\erad}}{2\sqrt{1+\lambda^2}\ephi}\right)',\\
    0&=\dot{\kang}= \frac{\lapse'\sqrt{\erad}\erad'}{2\sqrt{1+\lambda^2}\ephi^2} -\frac{\lapse}{2\sqrt{1+\lambda^2}\sqrt{\erad}}\left(1-\left(\frac{\erad'}{2\ephi}\right)^2\right)+\frac{\sqrt{\erad}\lapse}{2\sqrt{1+\lambda^2}}\left(\Lambda+\left(\frac{Q}{{\erad}}\right)^2\right) ,\label{kphidot0}\\\label{hamstatic}
  0&=\ham= \frac{1}{\sqrt{1+\lambda^2}}
     \Bigg[-\frac{{\ephi}}{2\sqrt{{\erad}}}+\frac{1}{2}\Bigg(\frac{\erad'}{2\ephi}\left(\sqrt{\erad}\right)'
              +\sqrt{\erad}\left(\frac{\erad'}{\ephi}\right)'\Bigg)
+ \frac{1}{2}\sqrt{E^x}E^\varphi
     \left(\Lambda+\left(\frac{Q}{E^x}\right)^2\right)\Bigg].
\end{align}
\end{subequations}
It is straightforward to solve the last equation for $\ephi$ to obtain %
\begin{align}
    \ephi=\varepsilon_1 \frac{\erad'}{2}\left(1-\frac{2M}{\sqrt{\erad}}+\frac{Q^2}{{\erad}}-\frac{\Lambda}{3}{\erad}\right)^{-1/2},
\end{align}
with $\varepsilon_1^2=1$ and $M\in\mathbb{R}$ being an integration constant. This expression
automatically satisfies \eqref{ephizero}. The range of ${\erad}$ will have to be
restricted so that the term between brackets is positive.
Now we can integrate \eqref{kphidot0} to obtain the lapse that, up to a trivial constant $c_1\neq0$,
is then given by
\begin{align}
    \lapse=c_1\left(1-\frac{2M}{\sqrt{\erad}}+\frac{Q^2}{{\erad}}-\frac{\Lambda}{3}{\erad}\right)^{1/2}.
\end{align}
One can check that
\eqref{kxdot} is now automatically satisfied, and thus all the equations.

It only remains to choose the (nonconstant) function ${\erad}(x)$ to completely
fix the gauge. The first, easiest, choice is to consider $\sqrt{\erad}(x)=x$.
Given the definition of $m$ in \eqref{eq.masspol} we thus have
\begin{align}\label{eq.m(x)}
     m(x)=M-\frac{Q^2}{2x}+\frac{\Lambda}{6}x^{3}.
 \end{align}
The domain of the solution in this case is restricted by $x>0$,
plus the range of possible values of $\erad$ found above,
  which, using \eqref{eq.m(x)}, can be conveniently
  expressed as
  \[
    1-\frac{2m(x)}{x}>0.
  \]
If we relabel $(t,x)$ as the pair of real functions $(\tilde{t}/c_1,r)$ on the manifold,
the metric
\eqref{eq.metric} reads
\begin{align}\label{eq.metapp1}
  {ds}^2 &=-\left(1-\frac{2m(r)}{{r}}\right){d\tilde{t}}^2 +
  \left(1-\frac{2\lbar m(r)}{{r}}\right)^{-1}\left(1-\frac{2m(r)}{{r}}\right)^{-1}{d{r}^2}
  +{r}^2{d\Omega}^2 .
\end{align}

In fact, a different form of the free function ${\erad}(x)$ can be chosen to remove the explicit pole in $q_{xx}$.
We do so by fixing
\begin{align}\label{eq.eradprime}
    ({\erad}')^2=4\erad\left(1-\frac{2\lbar m(\erad)}{\sqrt{\erad}}\right).
\end{align}
Renaming $\sqrt{\erad}(x)=:r(x)$, and relabeling $(t,x)$ as the pair of real functions $(\tt/c_1,z)$,  the metric in these coordinates reads
\begin{align}\label{eq.metapp2}
  {ds}^2 &=-\left(1-\frac{2\mofr}{{r(z)}}\right){d\tt}^2 +
  \left(1-\frac{2\mofr}{{r(z)}}\right)^{-1}{d{z}^2}
  +{r(z)}^2{d\Omega}^2,
\end{align}
which is only restricted by the values of ${z}$ that satisfy $2\mofr<r(z)$. Note that Eq. \eqref{eq.eradprime} can be expressed as
\begin{align}\label{app.rprime}
    \left(\frac{dr(z)}{dz}\right)^2 =1-\frac{2\lbar \mofr}{r(z)}.
\end{align}

\subsubsection{Near-horizon geometries}
\label{app.nhg}
We are now left with the case $\erad'=0$. We thus take $\sqrt{\erad}=a$ for some positive constant $a$. The first consequence is that, from \eqref{eq.masspol}, in this case we have
\[
  m=\frac{a}{2}.
\]
The diffeomorphism constraint equation $\diff=0$ is now automatically satisfied, while
the Hamiltonian constraint equation $\ham=0$ then provides the polynomial equation for $a$,
\begin{equation}
  a^4 \Lambda-a^2+Q^2=0.\label{eq:a}
\end{equation}
One can check that this relation also guarantees that \eqref{eomp2} is satisfied, so we are only left with Eqs. \eqref{eomq2} and \eqref{eomp1} for the functions
$\lapse$, $\shift$, $\ephi$, and $\krad$. The variable $\krad$
can be isolated from \eqref{eomq2}, and introduced in \eqref{eomp1},
which provides a partial differential equation (PDE) for the three functions $\lapse$, $\shift$, $\ephi$.
It is straightforward to check that such an equation ensures that the two-dimensional
Lorentzian metric
\[
  ds^2_2=-N^2(t,x)dt^2 +\left(1-\lbar\right)^{-1}\frac{(E^{\varphi}(t,x))^2}{a^2} \big(dx+N^x(t,x) dt\big)^2,
\]
constructed from the line element \eqref{eq.metric},
is of constant curvature. More precisely, its Gaussian curvature reads
\begin{equation}\label{eq:K_nh}
  \kappa=\frac{a^2-2Q^2}{a^4}(1-\lbar)=\left(\Lambda-\frac{Q^2}{a^4}\right)(1-\lbar),
\end{equation}
where we have used \eqref{eq:a} for the second equality.

To sum up, the solution provided by $\erad=a$ and $\sin(\lambda \kang)=0$
leads to the spacetime  $\manifold=\mathcal{M}^2\times S^2$
where $S^2$ is the sphere of radius $a$ and $(\mathcal{M}^2,\kappa)$ is a Lorentzian
space of constant (Gaussian) curvature $\kappa$, given by \eqref{eq:K_nh}.
These correspond to the so-called near-horizon geometries.

Any remaining choice of gauge for the set $\lapse(t,x)$, $\shift(t,x)$, and $\ephi(t,x)$ just provides a different chart of the near-horizon geometry.
Next, we make a choice to find some explicit
coordinate system. For simplicity, we take $\shift=0$, $\dot N=0$,
and $\ephi=a\sqrt{1-\lbar}$. The PDE mentioned above reduces to
\[
  N''+\kappa N=0.
\]
The general solution depends on the sign of $\kappa$, and it is given by
$N(x)=c_1\sin(\sqrt{\kappa}x+c_2)$ for $\kappa> 0$, 
$N(x)=c_1\sinh(\sqrt{-\kappa}x+c_2)$ for $\kappa< 0$, and $N(x)=c_1 z+c_2$ for $\kappa=0$ for
some constants $c_1$ and $c_2$.
Relabeling $(t,x)$
as $(T,z)$, and performing convenient shifts and
rescalings on $z$ and $T$ to absorb $c_1$ and $c_2$,
the metric \eqref{eq.metric} reads
\begin{align*}
  &ds^2=-\sin^2(\sqrt{\kappa} z) dT^2+dz^2+a^2d\Omega^2,\qquad \mbox{ for } \kappa> 0,\\
  &ds^2=-\sinh^2(\sqrt{-\kappa} z) dT^2+dz^2+a^2d\Omega^2,\qquad \mbox{ for } \kappa< 0,\\
  &ds^2=-z^{2} dT^2+dz^2+a^2d\Omega^2,\qquad \mbox{ for } \kappa= 0.
\end{align*}
The range of $z$ is the real line if $\kappa\leq 0$, and $z\in(0,\pi/\sqrt{\kappa})$
if $\kappa>0$.

Alternatively, let us find another chart that will be convenient
in order to show that the near-horizon geometries correspond
to one of the two limits in which the horizons degenerate.
The choice is $\ephi'=0$, and the lapse and shift then read
\begin{align*}
  \lapse &=\frac{1}{a}\sqrt{1+\lambda^2}\ephi\shift=\frac{a}{\ephi}\sqrt{1-\lbar},\\
  \shift &=\lapse^2,
\end{align*}
respectively.
The PDE reduces now to
\[
  \ddot\ephi\ephi+(\dot\ephi)^2=\left(1-\frac{2Q^2}{a^2}\right)(1-\lbar)^2,
\]
which can be solved to obtain
\[
  \ephi^2=\left(1-\frac{2Q^2}{a^2}\right)(1-\lbar)^2\left((t+c_1)^2+c_2\right).
\]
Relabeling now $(t,x)$ by $(Y-c_1,\zeta)$ and using \eqref{eq:a}
to write $1-2Q^2/a^2=2\Lambda a^2-1$,
the metric \eqref{eq.metric} takes the form
\begin{equation}\label{app.met.nhgY}
  ds^2=-\frac{1}{a^2}(1-2\Lambda a^2)(1-\lbar)(Y^2+c_2) d\zeta^2 +2 dYd\zeta +a^2d\Omega^2.
\end{equation}
It must be stressed that this is the same geometry irrespective of
the value of $c_2$,
and that different values of this constant simply provide
different patches of the near-horizon geometry (see, e.g., Ref. \cite{Bengtsson2022}).
In fact, convenient rescalings of $Y$ and $\zeta$ allow us to set
$c_2={-1,0,1}$.

\subsection{Homogeneous gauge}\label{app.homogeneous}

We partially fix the gauge by choosing $\erad'=0$ and $\ephi'=0$. The vanishing of $\diff$ implies $\kang'=0$. Then, $\ham'=0$ implies $\krad' \sin(2\lambda\kang)=0$.
If $\sin(2\lambda\kang)=0$, then either $\cos(\lambda\kang)=0$, so that $F=0$,
or  $\sin(\lambda\kang)=0$. The first case leads to a solution in phase space
with a degenerate metric. In the second
case $\sin(\lambda\kang)=0$ implies $\dot \erad=0$ because of \eqref{eomq1},
and therefore we fall into the near-horizon geometries analyzed above.
We  thus take $\sin(2\lambda\kang)\neq 0$ so that we are left with $\krad'=0$.
Now, the radial derivatives of \eqref{eomq1} and \eqref{eomq2} imply $\lapse'=0$
and $\shift''=0$, respectively.
The latter allows us to partially use the gauge freedom to set $\shift=0$.
From the geometrical perspective this is because $\shift=a(t)x+ b(t)$
ensures the existence of a function $Y$ such that $dx+\shift dt=\exp(-\int a(t)dt)dY$.\footnote{Although the outcome is the same, i.e., that one can set $\shift=0$,
  this corrects the argument in Sec. IV B of Ref. \cite{Alonso-Bardaji:2022ear}, where
  it was erroneously used that $\shift'=0$.}

The equations of motion then read
\begin{subequations}
\begin{align}
  \dot{\erad}&= \lapse\sqrt{{\erad}}\frac{\sin(2\lambda \kang)}{\lambda\sqrt{1\!+\!\lambda^2}},\label{eom1hom}
\\
    \dot{\ephi}&= \frac{\lapse}{\sqrt{1+\lambda^2}}\left(2\sqrt{\erad}\krad\cos(2\lambda \kang)
              +\frac{\ephi\sin(2\lambda \kang)}{2\lambda\sqrt{\erad}}\right) \label{eom2hom}
,
\\
  \dot{\krad}&=  \frac{\lapse}{2\sqrt{\erad}\sqrt{1+\lambda^2}}\Bigg\{\frac{\ephi}{2}\left({\Lambda}-3\left(\frac{Q}{\erad}\right)^2\right)+\frac{\ephi}{2\erad}\left(1+ \frac{\sin ^2(\lambda \kang)}{\lambda^2}\right) -\krad\frac{\sin(2\lambda\kang)}{\lambda}\Bigg\},\nonumber
\\
    \dot{\kang}&= \frac{\lapse}{2\sqrt{1+\lambda^2}}\Bigg\{\sqrt{\erad}\left(\Lambda+\left(\frac{Q}{\erad}\right)^2\right)  -\frac{1}{2\sqrt{\erad}}\left(1+ \frac{\sin ^2(\lambda \kang)}{\lambda^2}\right) \Bigg\},\label{eom4hom}\\
    0=\ham &=  \frac{1}{\sqrt{1+\lambda^2}}\Bigg\{ \frac{1}{2}\sqrt{E^x}E^\varphi\left(\Lambda+\left(\frac{Q}{E^x}\right)^2\right)-\frac{{\ephi}}{2\sqrt{{\erad}}}\left(1+\frac{\sin^2(\lambda \kang)}{\lambda^2}\right)  -\sqrt{{\erad}}{\krad}\frac{\sin(2\lambda \kang)}{\lambda}\Bigg\}.\nonumber
\end{align}
\end{subequations}
We isolate $\krad$ from the last equation, so that
\begin{align}
    \krad=\frac{\ephi \left(2 \lambda^2 \Lambda \erad^2-2 \erad \left(\lambda^2+\sin ^2(\lambda \kang)\right)+2 \lambda^2 Q^2\right)}{4 \lambda \erad^2 \sin (2\lambda \kang)},\label{hom:krad}
\end{align}
while we use \eqref{eom1hom} to obtain
\begin{equation*}
    \lapse=\frac{\lambda\sqrt{1+\lambda^2}\dot{\erad}}{\sqrt{\erad}\sin(2\lambda\kang)}.
\end{equation*}
Introducing $\lapse$ in \eqref{eom4hom} and integrating, we get
\begin{align}
    \frac{\sin(\lambda\kang)}{\lambda}=\varepsilon_2\sqrt{\frac{2m(\sqrt{\erad})}{\sqrt{\erad}}-1},\label{hom:kang}
\end{align}
with $\varepsilon_2^2=1$ and $m(\cdot)$ as given in \eqref{eq.m(x)},
so that
\begin{align}\label{m.erad}
    m(\sqrt{\erad})=M-\frac{Q^2}{2\sqrt{\erad}}+\frac{\Lambda}{6}(\sqrt{\erad})^3,
\end{align}
where $M\in\mathbb{R}$ an integration constant. Finally,
using \eqref{hom:krad} and \eqref{hom:kang} in \eqref{eom2hom}
and integrating, 
\begin{equation*}
    \ephi=c_2\sqrt{\erad}\sqrt{1-\frac{2\lbar m(\sqrt{\erad})}{\sqrt{\erad}}}\sqrt{\frac{2m(\sqrt{\erad})}{\sqrt{\erad}}-1},
\end{equation*}
for some integration constant $c_2$. At this point
it only remains to choose the function $\erad(t)$ to completely fix the gauge.

In analogy with the static case, we first consider $\sqrt{\erad}(t)=t$. If we relabel $(t,x)$ as the pair $(r,\tt)$, and absorbing the constant $c_2$ with a convenient rescaling, the metric reads
\begin{align}\label{eq.metapp3}
  {ds}^2 &=-
  \left(1-\frac{2\lbar m(\rr)}{\rr}\right)^{-1}\left(\frac{2m(\rr)}{\rr}-1\right)^{-1}{d\rr^2}+\left(\frac{2m(\rr)}{\rr}-1\right){d\tt}^2 
  +\rr^2{d\Omega}^2,
\end{align}
which is restricted by the values of $r$ that satisfy $2\lbar m(r)<r<2m(r)$. Contrary to what happens in the static domain,
the factor with $\lbar$ does restrict the range of $\rr$, so it is important to try to remove that pole in this case.

We do so by making yet another choice of gauge, this time
fixing $\erad(t)$ through its derivative by
\begin{align}\label{eq.eraddot}
    (\dot{\erad})^2=4\erad\left(1-\frac{2\lbar m(\erad)}{\sqrt{\erad}}\right).
\end{align}
Renaming $\sqrt{\erad}(t)=:r(t)$, and relabeling
$(t,x)$ as the pair of real functions $(z,\tt)$, the metric reads
\begin{align}\label{eq.metapp4}
{ds}^2 &=- 
  \left(\frac{2\mofr}{r(z)}-1\right)^{-1}{dz^2}+\left(\frac{2\mofr}{r(z)}-1\right){d\tt}^2
  +r(z)^2{d\Omega}^2.
\end{align}
The ranges of the coordinates, as determined by the existence
of the solution, are given by $\tt\in \mathbb{R}$,
while $z$ is restricted by the condition $r(z)<2\mofr$ 
plus the domain (or domains) of the existence of the solution
of \eqref{eq.eraddot}, i.e.,
\begin{align}
    \left(\frac{dr(z)}{dz}\right)^2 
    = 1-\frac{2\lbar \mofr}{r(z)},
\end{align}
which will correspond to ranges of $z$
for which $r(z)\geq2\lbar \mofr$.

\subsection{Horizon-crossing gauge}\label{app.covering}

We now start by partially fixing the gauge by
$\dot \erad=0$ and $\dot \ephi=0$.
From \eqref{eomq1} we have that, if $\erad'=0$, then
  $\lapse \erad \sin(2\lambda\kang)=0$. For a nondegenerate
  geometry we need that the product $\lapse\erad$ does not vanish everywhere,
  while, if $\sin(2\lambda\kang)=0$, we fall back to the near-horizon
  geometries found above. As a result, we assume in the remainder
  that $\erad{}'$ does not vanish identically.
From $\diff=0$ we thus find
\begin{align}\label{cross.kphi}
  \krad = %
  \frac{\ephi\kang'}{\erad'}.
\end{align}
Then, we can isolate the shift $\shift$
from \eqref{eomq1} and introduce it in \eqref{eomq2} to obtain
\begin{align*}
  \ephi\erad \sin(2\lambda\kang)
  \left(1+\left(\frac{\lambda \erad'}{2\ephi}\right)^2\right)
  \left(\lapse' \erad'\ephi+\lapse(\ephi'\erad'-\ephi\erad'')\right)=0.
\end{align*}
The case $\sin(\lambda\kang)=0$ was studied in the previous section \ref{app.static}.
In addition, if we consider $\cos(\lambda\kang)=0$, the vanishing of the Hamiltonian constraint \eqref{H} requires either a constant $\erad$, which we already discarded, or that $\ephi=0$, which makes the metric degenerate.
Therefore for the above equation to be satisfied, we are left with
the vanishing of the last factor, which integrates to
\begin{align*}
    \lapse=\frac{c_3}{2}\frac{\erad'}{\ephi},
\end{align*}
for some integration constant $c_3$.

On the other hand, introducing \eqref{cross.kphi} in \eqref{H},
the integration of $\ham=0$ yields
\begin{align*}
  \frac{\sin(\lambda\kang)}{\lambda}=\varepsilon_3
  \left(1+\left(\frac{\lambda \erad'}{2\ephi}\right)^2\right)^{-1/2}\sqrt{\left(\frac{\erad'}{2\ephi}\right)^2- \left(1-\frac{2m(\sqrt{\erad})}{\sqrt{\erad}}\right)},
\end{align*}
where $\varepsilon_3^2=1$, and we use again \eqref{m.erad}
for some integration constant $M$. Then, the shift, isolated from \eqref{eomq1}, reads
\begin{align*}
  \shift= %
  \varepsilon_4 c_3\frac{\sqrt{\erad}}{\ephi}
  \sqrt{1-\frac{2\lbar m(\sqrt{\erad})}{\sqrt{\erad}}}
  \sqrt{\left(\frac{\erad'}{2\ephi}\right)^2+\frac{2 m(\sqrt{\erad})}{\sqrt{\erad}}-1},
\end{align*}
where $\varepsilon_4$ is minus the sign of $\sin(2\lambda\kang)$. The remaining equations, \eqref{eomp1} and \eqref{eomp2}, are now satisfied.
If we rename the two free functions
$\sqrt{\erad}(x)=:r(x)$ and $\ephi(x)=:s(x)$ for compactness,
the metric \eqref{eq.metric} then reads
\begin{align}\label{r-s-metric}
    {ds}^2 =&
    -\left(1-\frac{2m(r(x))}{r(x)}\right) {dt}^2
              +2\left(1-\frac{2\lbar m(r(x))}{r(x)}\right)^{-1/2}\frac{s(x)}{r(x)}
              \sqrt{\left(\frac{r(x) r'(x)}{s(x)}\right)^2+\frac{2m(r(x))}{r(x)}-1}\,{dt}{dx}\nonumber\\
    &+ \left(1-\frac{2\lbar m(r(x))}{r(x)}\right)^{-1}\left(\frac{s(x)}{r(x)}\right)^{2} 
    {dx}^2 
   +r(x)^2 {d\Omega}^2 \:,
\end{align}
after setting $\varepsilon_4 c_3 = 1$ with no loss of generality
by a constant rescaling (and change of sign if needed) of $t$.

The fact that $s(x)$ is pure gauge becomes explicit now,
as it may be absorbed by a coordinate transformation $dy=s(x)dx$.
Several choices can be made at this point, and find for each choice
  the corresponding chart. Our choice here, as in Refs.~\cite{Alonso-Bardaji:2021yls,Alonso-Bardaji:2022ear} for the vacuum case,
is to fix $s$ by demanding that
$$s=\sqrt{r^2-2\lbar r m(r)}$$
in order to remove explicit
divergences in the metric element $q_{xx}$. Now, in order to completely fix the gauge,
we only need to choose the specific form of the function $r(x)$. As we still have possible divergences in the argument of the second square root of the
component $dtdx$ of the line element \eqref{r-s-metric}
(coming from the choice of $s(x)$), we set
\begin{align}\label{eq.rpapp}
    \big(r'(x)\big)^2= 1-\frac{2\lbar m(r(x))}{r(x)},
\end{align}
which implicitly defines $r(x)$.

After taking these two choices we now relabel $(t,x)$ as the
pair of real functions $(\tau,z)$, so that
the metric in these coordinates reads
\begin{align}\label{eq.metapp5}
        ds^2 = -\bigg(1-\frac{{2\mofr}}{r(z)}\bigg){d\tau}^2 
 +2\sqrt{\frac{2\mofr}{r(z)}}\,d\tau dz +{dz}^2 +r(z)^2d\Omega^2.
\end{align}
The domain of existence of the solution of
\eqref{eq.rpapp} (with $x$ replaced by $z$),
plus the requirement that
\[
  m(r(z))\geq 0,
  \]
provide the range of the coordinate $z$,
while $\tau$ covers the real line.

\subsection{Coordinate transformations}
\label{sec:coord_transf}

For completeness, we next provide
the coordinate transformations between the above charts on the intersection
of their corresponding domains.
Observe that, although the static and homogeneous regions do not overlap,
the horizon-crossing coordinates $(\tau,z)$ cover them partially [or completely, depending on the sign of $m(r)$].

In the region covered by the points where $m(r)\geq0$ and $2\mofr<r(z)$,
the change given by
    \begin{align}\label{tautot}
      d\tau&=d\tt+\bigg(1-\frac{{2\mofr}}{r(z)}\bigg)^{-1}
          \sqrt{\frac{2\mofr}{r(z)}}dz
    \end{align}
is a coordinate transformation from the region $2\mofr<r(z)$ of the
coordinates $(\tau,z)$, %
    where the metric reads \eqref{eq.metapp5},
    to the static region $m(r)\geq0$ of the coordinates $(T,z)$,
    where \eqref{eq.metapp2} holds. In addition, it also provides the transformation
from the homogeneous region  $r(z)<2\mofr$, to the whole domain of the coordinates
$(z,T)$, where the metric reads \eqref{eq.metapp4}.

  Further, the change $z\to r$, given by 
    \begin{align}\label{rtoz}
      dz^2 &= \left(1-\frac{{2\lbar m(\rr)}}{\rr}\right)^{-1} d\rr^2,
    \end{align}
transforms \eqref{eq.metapp2} to  \eqref{eq.metapp1} and \eqref{eq.metapp4} to \eqref{eq.metapp3} in the regions where $2\lbar m(r)<r$.

\section{Proof of the existence of critical values and horizons}\label{app.proofs}

In the following we refer to roots of $P$ as real roots of $P{(r,s)}$ in the first argument.
Prime denotes derivative with respect to the first argument. 
In Fig.~\ref{fig.proofs} all the possible forms of the polynomial $P$ are qualitatively represented for the different values of the parameters, which could be of help to follow the proof below about the positive roots of this polynomial.

\begin{lemma}\label{res:lemma1}
   Consider $P(r,\lproof)$ as defined in \eqref{def:P}, i.e.,
\begin{equation}
 P(r,\lproof)=\left\{\begin{array}{lr}
 
             \cfrac{\lproof\Lambda}{3}r^4-r^2+2\lproof M r-\lproof Q^2, & {\rm if}\,\,\, Q\neq 0,\\[10pt]
               \cfrac{\lproof\Lambda}{3}r^3-r+2\lproof M, & {\rm if}\,\,\, Q= 0,
              \end{array}\right.
\end{equation}
and fix $\lproof\in (0,1]$.
A value $r_0(M,Q,\Lambda,\lproof)>0$ such that $P(r_0,\lproof)=0$ and either
\begin{enumerate}
\item[a)] $P'(r_0,\lproof)<0$, or
\item[b)] $P'(r_0,\lproof)=0, P''(r_0,\lproof)<0$, or 
\item[c)] $P'(r_0,\lproof)=P''(r_0,\lproof)=0,P'''(r_0,\lproof)<0$, or
\item[d)] $P'(r_0,\lproof)=P''(r_0,\lproof)=0,P'''(r_0,\lproof)=0$, $P''''(r_0,\lproof)<0$
\end{enumerate}
exists only in the cases
\begin{enumerate}
\item[1.] $\Lambda>0$, $Q\neq 0$, $M>0$, with $8Q^2<9\lproof M^2$ and $\Lambda\in (\Lambda_-,\Lambda_+)\cap (0,\Lambda_+)$,  and (a) holds.
\item[1D.] $\Lambda>0$, $Q\neq 0$, $M>0$, with $8Q^2<9\lproof M^2<9Q^2$ and
  $\Lambda=\Lambda_-$,  and (b) holds.
\item[2.] $\Lambda>0$, $Q= 0$, and $0<\lproof^{3/2} 3\sqrt{\Lambda} M<1$, and (a) holds,
\item[3.] $\Lambda=0$ and $\sqrt{\lproof} M>|Q|$ and (a) holds,
\item[3D.] $\Lambda=0$ and $\sqrt{\lproof} M=|Q|>0$ and (b) holds,
\item[4.] $\Lambda<0$, $Q\neq 0$,  $\sqrt{\lproof} M>|Q|$ and $\Lambda\in(\Lambda_-,0)$ and (a) holds,
\item[4D.] $\Lambda<0$, $Q\neq 0$,  $\sqrt{\lproof} M>|Q|$ and
  $\Lambda=\Lambda_-$ and (b) holds,
\item[5.] $\Lambda<0$, $Q= 0$, and $M>0$ and (a) holds,
\end{enumerate}
where
\begin{equation}\label{eq:Lambdapm}
    \Lambda_\pm =\frac{1}{32\lproof^2Q^6}
    \left(8Q^4-\beta(\beta+4Q^2)\pm3\sqrt{\lproof M^2\beta^3}\right),
    \quad\mbox{ with }\quad \beta:=9\lproof M^2-8Q^2.
\end{equation}
\begin{remark}\label{remarkbeta+}
Once $\beta>0$, both $\Lambda_\pm$ are real and distinct. Moreover,
$\Lambda_+$ is always positive, while $\Lambda_-$ is negative,
vanishing, or positive when $\lproof M^2-Q^2$ is greater, equal, or
smaller than zero, respectively. When $\beta=0$ both $\Lambda_\pm$ coincide,
$\Lambda_+=\Lambda_-=(2sQ)^{-2}$. Therefore, the conditions in case
\emph{1} can be written as $\Lambda>0$, $Q\neq 0$, $M>0$, with $\{Q^2
\leq \lproof M^2$ and $\Lambda\in(0,\Lambda_+)\}$ or $\{8Q^2<9\lproof
M^2<9 Q^2$ and $\Lambda\in(\Lambda_-,\Lambda_+)\}$.
  \end{remark}
\begin{remark}\label{remark_rinf}
In cases \emph{1}, \emph{1D} and \emph{2},
the set of $r$ defined by $P(r,\lproof)\leq  0$ with infimum $r_0$
is given by the closed interval $[r_0,\rmax]$ for some finite value $\rmax>r_0$,
  for which $P'(\rmax,\lproof)>0$.
  In the rest of the cases
  the interval where $P(r,\lproof)\leq 0$ with infimum $r_0$
  is unbounded from above.

  Moreover, the limiting case for \emph{1}, \emph{1D}
    and \emph{2}, when $\Lambda>0$ and $M>0$, in which the largest root
    of $P(r,\lproof)$ is a double root (so that, say $r_0= \rmax$ in the limit),
   is given by $\Lambda=\Lambda_+$, and then
   $P'(r_0,\lproof)=0$ with  $P''(r_0,\lproof)>0$.
   Note that $\Lambda_+|_{Q\to 0 }=(9M^2 s^3)^{-1}$.
 \end{remark}
  \begin{remark}\label{remark_origin}
    When $Q\neq 0$, there exists $R>0$ such that $R\leq r_0$ and $P(r,s)\leq 0$ in
    $r\in[0,R]$.
    In the double root cases \emph{1D, 3D}, and \emph{4D},
    we have $R=r_0$, $P'(r_0,s)=0$ and $P''(r_0,s)<0$,
    while in the rest of the cases, $R<r_0$ and $P'(R,s)>0$.
  When $Q=0$, there are no positive roots smaller than $r_0$.
\end{remark}

\end{lemma}

\begin{proof}

We denote by $\Delta$ the discriminant of the fourth-order polynomial, given by
\begin{equation}\label{eq:Delta}
 \Delta=\frac{16}{27}\lproof^2\Lambda \left[-16 \lproof ^4 \Lambda ^2 Q^6-3 \lproof ^2 \Lambda  \left(27 \lproof ^2 M^4-36 \lproof  M^2 Q^2+8 Q^4\right)+9 (\lproof  M^2-Q^2)\right].
\end{equation}
Let us recall that
the discriminant is zero if and only if at least two roots are equal.
In such a case, there are at most two equal roots if and only if (see, e.g., Ref. \cite{Rees1922})
\begin{equation}\label{one_double}
  -1<-4\lproof^2\Lambda Q^2<3.
\end{equation}
If the discriminant is positive, there are either four roots or none.
If negative, there are only two roots.
Finally, if there are four distinct roots, then $\Delta>0$ necessarily.

\emph{Case 1}. Assume  $\Lambda>0$ and $Q\neq 0$. Since $P(0,\lproof)<0$,
and $P(r,s)$ is positive at $r\to\pm\infty$, then $P(r,\lproof)$ has two roots at least,
one positive $r_{1+}>0$, and one negative $r_{1-}<0$, such that \emph{(i)} if more roots
exist, then they are contained in the interval $(r_{1-},r_{1+})$,
and \emph{(ii)} $P'(r_{1-},\lproof)\leq 0$ and $P'(r_{1+},\lproof)\geq 0$.
Moreover, since $r_{1+}$ cannot be a local maximum,
it cannot satisfy some of the requirements of $r_0$.
As a result, for $r_0$ to exist in this case, we need at least a third
positive root.

If $\Delta<0$, there are no more roots,
and therefore, no such $r_0$ exists.
If $\Delta>0$
there are necessarily two more roots,
and since the product of the four roots
equals $-3Q^2/\Lambda$, then these two additional roots cannot vanish and
must have the same sign. Thus, the existence of $r_0$ requires that
these additional two roots are positive.
If $\Delta=0$, there are several possibilities with two or three roots.
However, since we need a third positive root,
and the product of the roots must be negative,
we are left, in principle, with only two possibilities so that $r_0$ can exist:
either there is a third simple root $r_{1s}<r_{1+}$ and $r_{1+}$ is a double root,
or there is a third double root $r_{1d}$, and necessarily $r_{1-}<r_{1d}<r_{1+}$.
In the first case we have $P(r,s)=\lproof\Lambda (r-r_{1+})^2(r-r_{1-})(r-r_{1s})/3$ and
therefore $P'(r_{1s},\lproof)>0$. Thence, $r_{1s}$ cannot satisfy the requirements
of $r_0$.
In the second case we have $P(r,s)=\lproof\Lambda (r-r_{1+})(r-r_{1-})(r-r_{1d})^2/3$, thus $P'(r_{1d},\lproof)=0$
and $P''(r_{1d},\lproof)<0$, and therefore $r_{1d}$ is a local maximum.

Let us further assume first that $M> 0$.
We prove next that, if there is a local maximum of $P(r,s)$, it must be for $r>0$,
which shows that $r_{1d}$ (if $\Delta=0$ and it exists)
and the two additional roots (if $\Delta>0$ and they exist)
must be positive.
Let $a$ satisfy $P'(a,\lproof)=0$, which is equivalent to
$2\lproof\Lambda a^3/3-a+\lproof M=0$. This implies
$a(3-2a^2\lproof\Lambda)>0$. If $a<0$,
then $3-2a^2\lproof\Lambda<0$, so
$P''(a,\lproof)=-2(1-2 a^2\lproof\Lambda)>0$.
Therefore $a> 0$ necessarily to have a local maximum there.
Observe that there is only one local maximum of $P(r,\lproof)$.
As a result, if $\Delta=0$ and $r_{1d}$ exists, it is positive,
and, if $\Delta>0$, three of the roots are positive.
In the former case, if $r_{1d}$ exists,
it satisfies the requirements [with \emph{(b)}], so that we can set
  $r_0=r_{1d}$, and the intervals in $r\geq 0$ where $P(r,\lproof)\leq 0$ are then given by
$[0,r_{1d}]$ and $[r_{1d},r_{1+}]$.
In the latter case, if we denote the three roots by $0<r_a<r_b<r_{1+}$,
we clearly have that $P'(r_b,\lproof)<0$, and thus $r_0=r_b$
(and only that) satisfies the requirements [with \emph{(a)}].
In this case the ranges for which $P(r,\lproof)\leq 0$ in $r\geq 0$ are
  given by the bounded intervals $[0,r_a]$ and $[r_b,r_{1+}]$.
The determination of the constraints $\Delta>0$,
and $\Delta=0$ plus the fact that the double root $r_{1d}$ is a third root
(and hence no more double roots exist),
in terms of $M$, $Q$, and $\Lambda$, is left to the end of the proof.

If $M=0$, the polynomial
is even, with a unique local maximum at $r=0$, where $P(r,\lproof)$ takes a negative value.
Therefore, only the root $r_{1+}$ is a positive root; thus, none satisfies the requirements of $r_0$.

Finally, the case $M<0$ can be dealt with by applying the same arguments above
under the change $r\to -r$ ---observe that $P(r,s)$ is invariant under
a change $(M\to -M, r\to -r)$--- and using that $r_{1-}$ cannot be a local maximum either.
As a result, one needs a third positive root, but the only possibilities
provide extra negative roots; therefore, no $r_0$ exists.

\emph{Case 2}. Assume  $\Lambda>0$ and $Q= 0$.
Now $P(r,s)$ is a third-order polynomial in $r$, it thus has at least one root
  $r_{2o}$, it satisfies $P(0,\lproof)=2\lproof M$, $P'(0,\lproof)=-1$,
has a local maximum at $a_-=-\sqrt{1/(\lproof\Lambda)}$,
and a local minimum at $a_+=\sqrt{1/(\lproof\Lambda)}$.
Because $P(0,\lproof)=2\lproof M$,
if $M\neq 0$ the root $r_{2o}$ has the sign of $-M$
and if there are additional roots (either one double or two distinct),
since the product of the three roots must equal $2\lproof M$,
they must have the opposite sign.

Assume first that $M>0$. Thus $r_{2o}<0$, and in order to have a positive root,
we need more roots. That condition is fulfilled when the local minimum
is attained at a nonpositive value of $P(r,s)$, that is, $P(a_+,\lproof)\leq 0$.
However, since $P''(r>0,\lproof)>0$, from the requirements of $r_0$ we must
request that
$P'(r_0,\lproof)<0$. This implies that we need two more distinct roots,
and therefore $P(a_+,\lproof)< 0$, which is equivalent to
$1/(3\sqrt{\lproof\Lambda})-\lproof M >0$.
The two roots must be positive, as mentioned, and, if we denote them
by $0<r_a<r_{2+}$, then
$r_a$ clearly satisfies $P'(r_a,\lproof)<0$ necessarily, and
thus $r_0=r_a$, and only that, satisfies the requirements of $r_0$.
Moreover, the only domain where $P(r,\lproof)\leq 0$
is the bounded interval $[r_0,r_{2+}]$, and, since $a_+<r_{2+}$, we have
$P'(r_{2+},\lproof)>0$. In this case, there are no more
positive roots smaller than $r_0$.

If $M=0$, the three roots of the polynomial are $r=0,\pm\sqrt{3/(\lproof\Lambda)}$, and it is straightforward to check that
$P'(r,s)$ is positive at the positive root.

Assume now that $M<0$. Then $r_{2o}>0$. Since it is the
only possible positive root, we necessarily have $P'(r_{2o},\lproof)>0$,
and thus no root fulfills the requirements of $r_0$.

\emph{Case 3}. If  $\Lambda=0$ and $Q\neq 0$,
$P(r,s)$ becomes a second order polynomial in $r$, so
the analysis is quite straightforward.
The necessary and sufficient condition for the existence of roots is
$\lproof M^2 -Q^2\geq 0$,  which needs $M\neq 0$,
and then all roots are positive if and only if $M>0$.
Assume $M>0$; then, if  $\lproof M^2 - Q^2 > 0$,
denoting by $0<r_{3a}<r_{3b}$ the two distinct roots,
$r_{3b}$ clearly satisfies the requirements of $r_0$
and we can set $r_0=r_{3b}$ [with \emph{(a)}].
In this case, the interval where $P(r,\lproof)\leq 0$ with infimum $r_0$
is unbounded from above.
The saturation of the inequality,
$\lproof M^2 -Q^2=0$, yields a double root $r_d=\lproof M$
with $P'(r_d,\lproof)=0$. Since $P''(r,s)<0$ everywhere,
$r_d$ meets the requirements with \emph{(b)}, and, since $P(0,\lproof)<0$,
the intervals where $P(r,s)\leq 0$ are given by $[0,r_0]$ and $[r_0,\infty)$ with $r_0=r_d$.
If $M<0$, no root satisfies the requirements.

In the case  $\Lambda=0$ and $Q= 0$  we only have one simple root $r_3=2\lproof M$
and $P'(r_3,s)=-1$,
and thus only if $M>0$ $r_3$ satisfies the requirements [with \emph{(a)}],
and we have $P(r,\lproof)<0$ only for $r>r_{3}$.

\emph{Case 4}. Take  $\Lambda<0$ and $Q\neq 0$.
We have that $P''(r,\lproof)<0$ for all $r$
and, given the asymptotic behavior, either there are no roots,
there is a double one, or there are two, and $P(r,\lproof)$ has one maximum in $r$.
Assume now $M>0$. Then the maximum is located at positive values of $r$
because $P'(0,\lproof)>0$.
Therefore, since $P(0,\lproof)<0$, if there is any root, it must be positive. 
Since  $P''(r,\lproof)<0$, in order to have a root $r_0$
fulfilling the requirements, we need either
$P'(r_0,\lproof)=0$ and to have a double root $r_{4d}$, or
$P'(r_0,\lproof)<0$
so that there are two distinct roots $0<r_a<r_b$.
In the double root case we need $\Delta=0$ (and that is sufficient),
while we have two simple roots since no four roots can exist
if and only if $\Delta<0$.
The double root $r_0=r_{4d}$ satisfies the requirements
[with \emph{(b)}], and we have $P(r,\lproof)<0$ for all $r>r_0=r_{4d}$.
In the $\Delta<0$ case, clearly, $r_0=r_b$, and only that
satisfies the requirements [with \emph{(a)}].
In both cases the interval in $r$ where $P(r,\lproof)<0$ with infimum $r_0$ is
unbounded from above.
Moreover, we have $P(r,\lproof)<0$ on $r\in(0,r_0)$ in the $\Delta=0$ case,
and on $r\in(0,r_a)$ in the $\Delta<0$ case.

Like before, the case $M<0$ can be treated by changing $r\to -r$ to the previous
analysis. But now the roots, if any, are negative,
and we thus have $P(r,\lproof)<0$ for all $r\geq 0$.
If $M=0$ then also $P(r,\lproof)<0$ for all $r\geq 0$.

\emph{Case 5}. Assume  $\Lambda<0$ and $Q= 0$. We have $P'(r,\lproof)<0$ for all $r$
and $P(0,\lproof)=2\lproof M$. Therefore there is a positive root
if and only if $M>0$, and that root is $r_0$, which is the only root
and it is simple,
thus satisfying the requirements [with \emph{(a)}].
The interval where $P(r,s)\leq 0$ is then $[r_0,\infty)$.

\begin{figure}
\centering
\begin{minipage}[b]{0.24\textwidth}
\centering
\includegraphics[width=\textwidth]{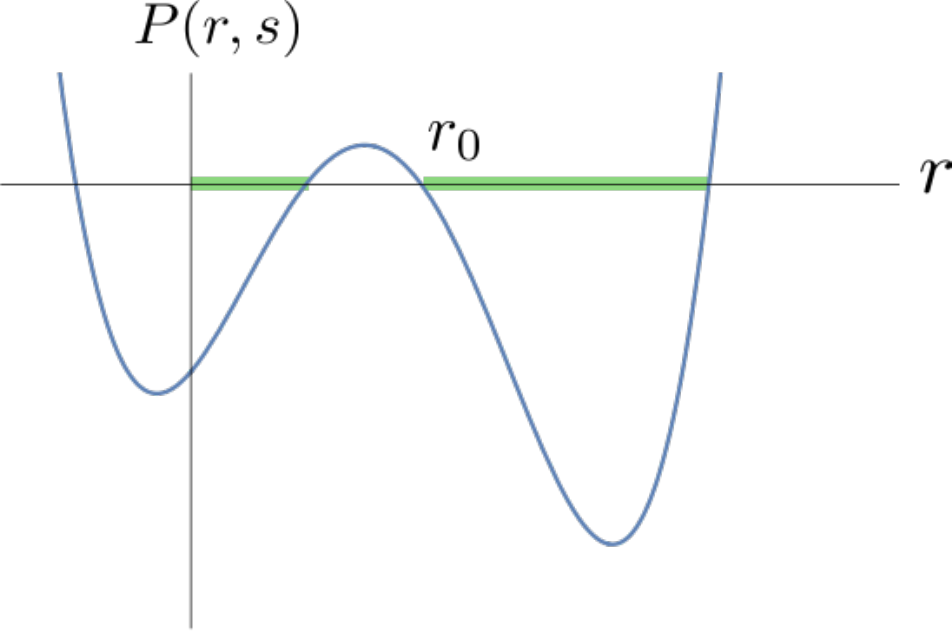}
\scriptsize{%
$\Lambda>0$, $Q\neq0$, $\Delta>0$\\and $M>0$}
\end{minipage}
\begin{minipage}[b]{0.24\textwidth}
\centering
\includegraphics[width=\textwidth]{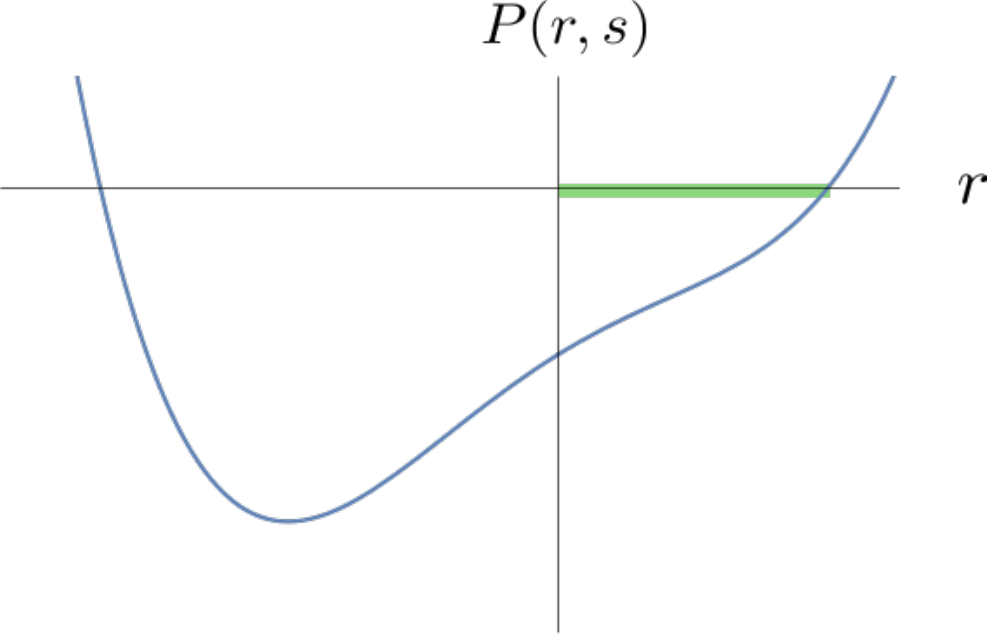}
\scriptsize{%
$\Lambda>0$, $Q\neq0$, $\Delta<0$\\and $M<0$}
\end{minipage}
\begin{minipage}[b]{0.24\textwidth}
\centering
\includegraphics[width=\textwidth]{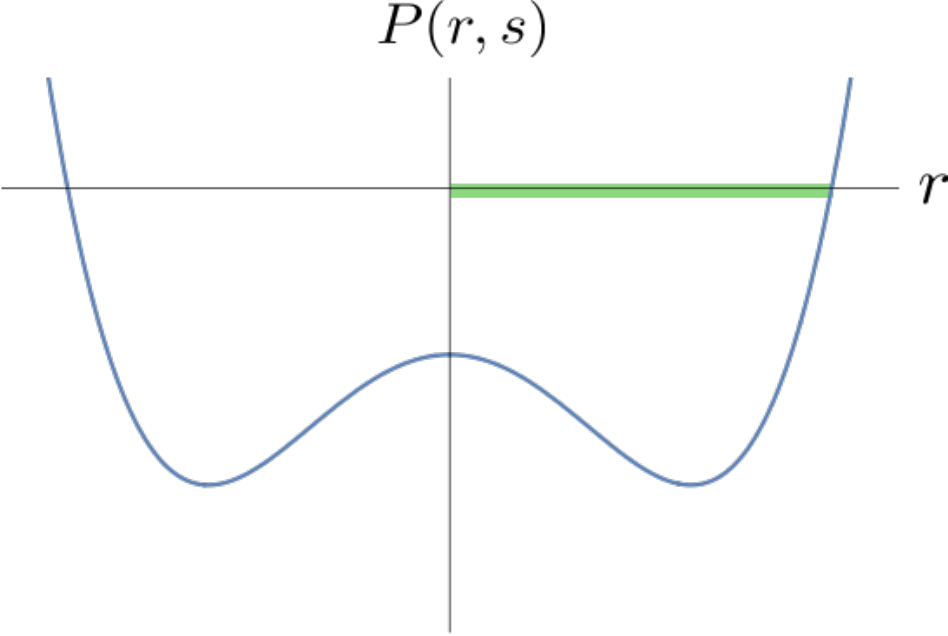}
\scriptsize{%
$\Lambda>0$, $Q\neq0$\\and $M=0$}
\end{minipage}\\[12pt]
\begin{minipage}[b]{0.24\textwidth}
\centering
\includegraphics[width=\textwidth]{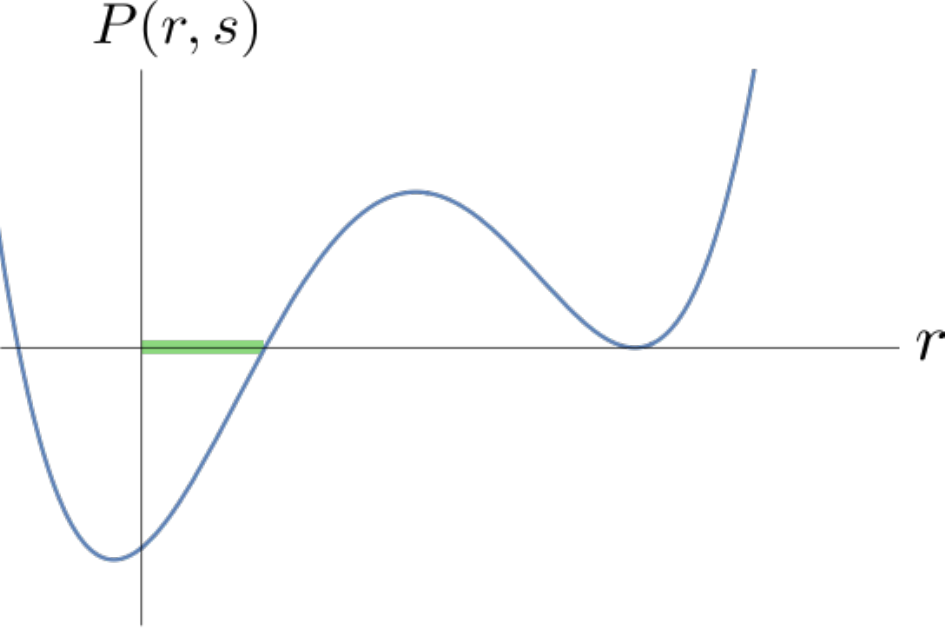}
\scriptsize{%
$\Lambda>0$, $Q\neq0$, $\Lambda=\Lambda_+$\\and $M>0$}
\end{minipage}
\begin{minipage}[b]{0.24\textwidth}
\centering
\includegraphics[width=\textwidth]{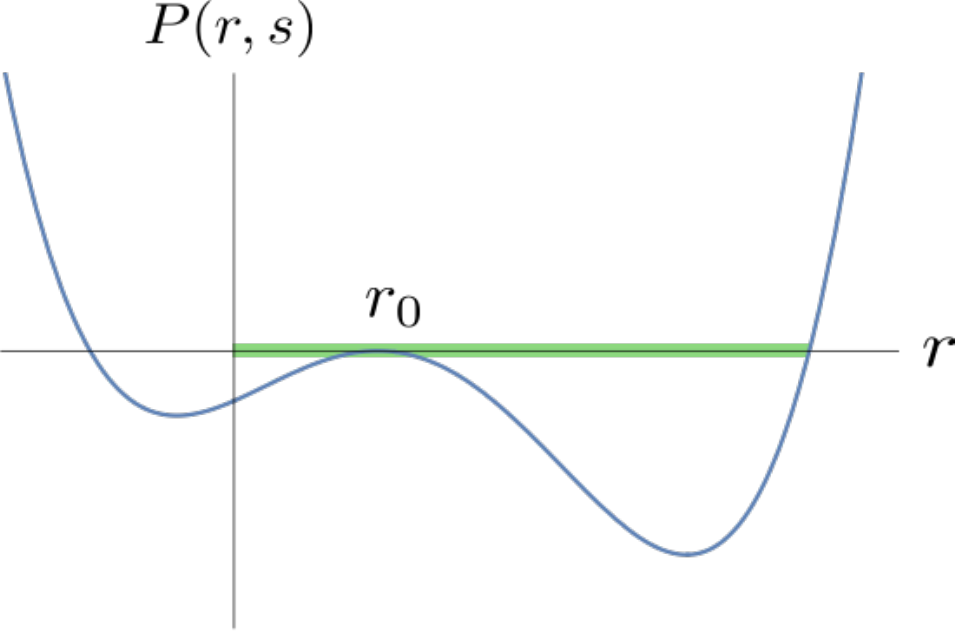}
\scriptsize{%
$\Lambda>0$, $Q\neq0$, $\Lambda=\Lambda_-$\\and $M>0$}
\end{minipage}
\begin{minipage}[b]{0.24\textwidth}
\centering
\includegraphics[width=\textwidth]{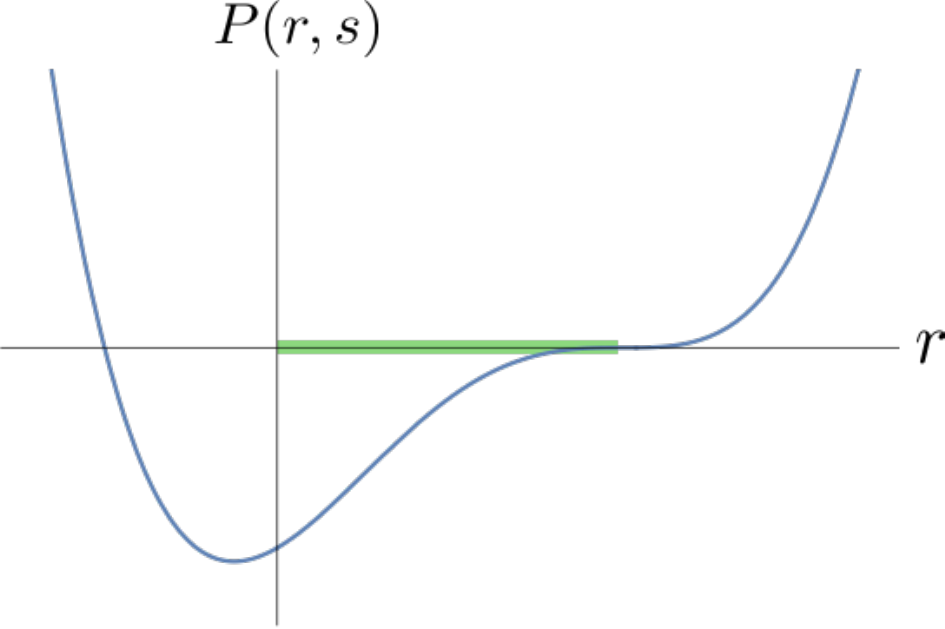}
\scriptsize{%
$\Lambda>0$, $Q\neq0$, $\Lambda=\Lambda_+=\Lambda_-$\\and $M>0$}
\end{minipage}\\[12pt]
\begin{minipage}[b]{0.24\textwidth}
\centering
\includegraphics[width=\textwidth]{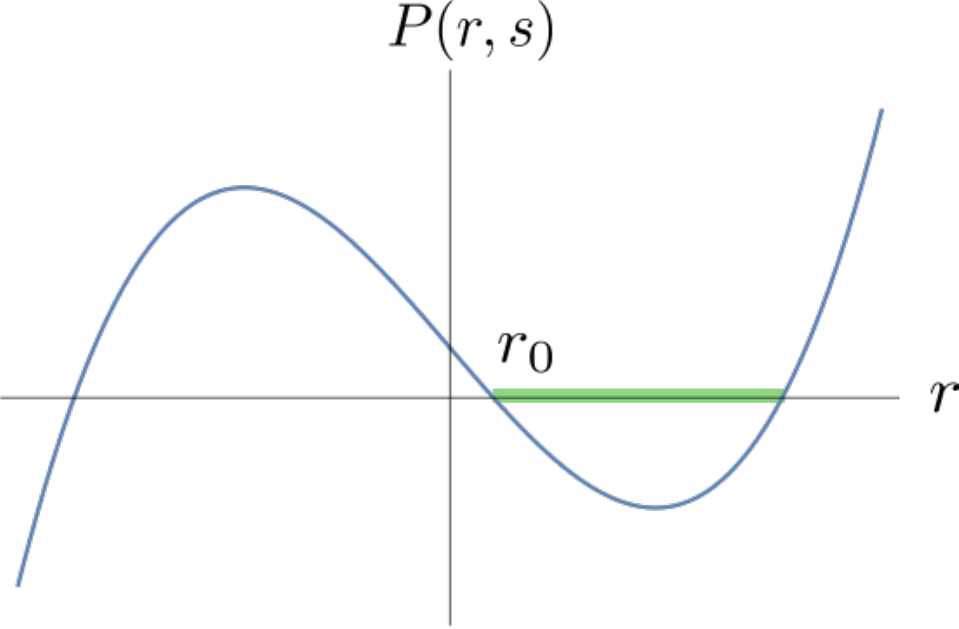}
\scriptsize{%
$\Lambda>0$, $Q=0$, $9s^{3}{\Lambda}M^2<1$\\and $M>0$}
\end{minipage}
\begin{minipage}[b]{0.24\textwidth}
\centering
\includegraphics[width=\textwidth]{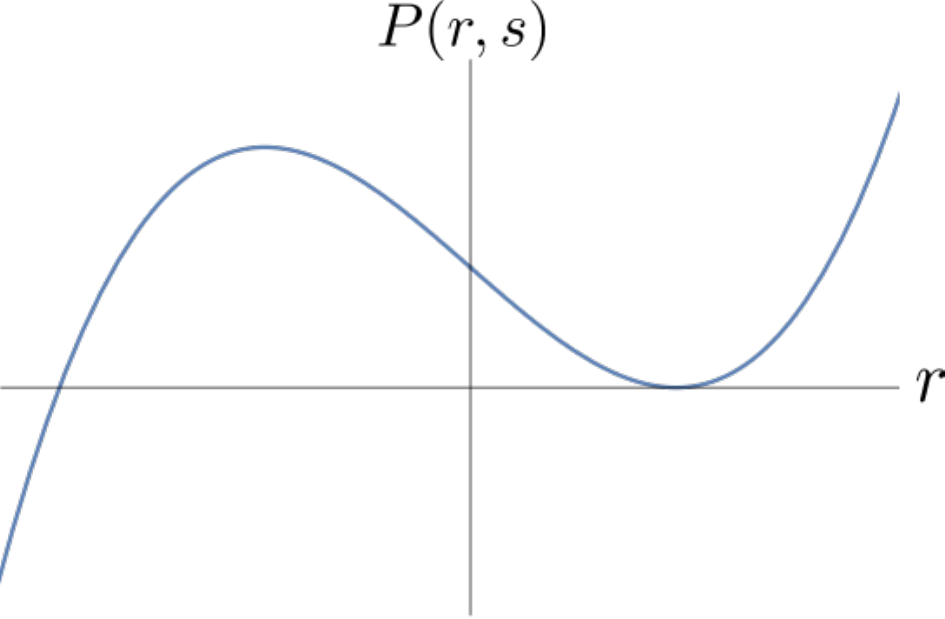}
\scriptsize{%
$\Lambda>0$, $Q=0$, $9s^{3}{\Lambda}M^2=1$\\and $M>0$}
\end{minipage}
\begin{minipage}[b]{0.24\textwidth}
\centering
\includegraphics[width=\textwidth]{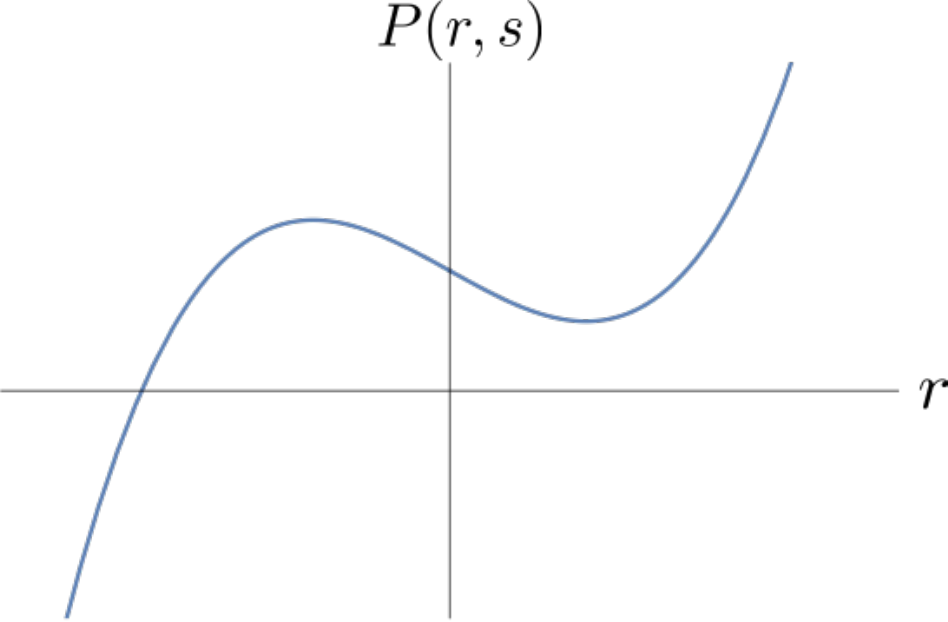}
\scriptsize{%
$\Lambda>0$, $Q=0$, $1<9s^{3}{\Lambda}M^2$\\and $M>0$}
\end{minipage}
\begin{minipage}[b]{0.24\textwidth}
\centering
\includegraphics[width=\textwidth]{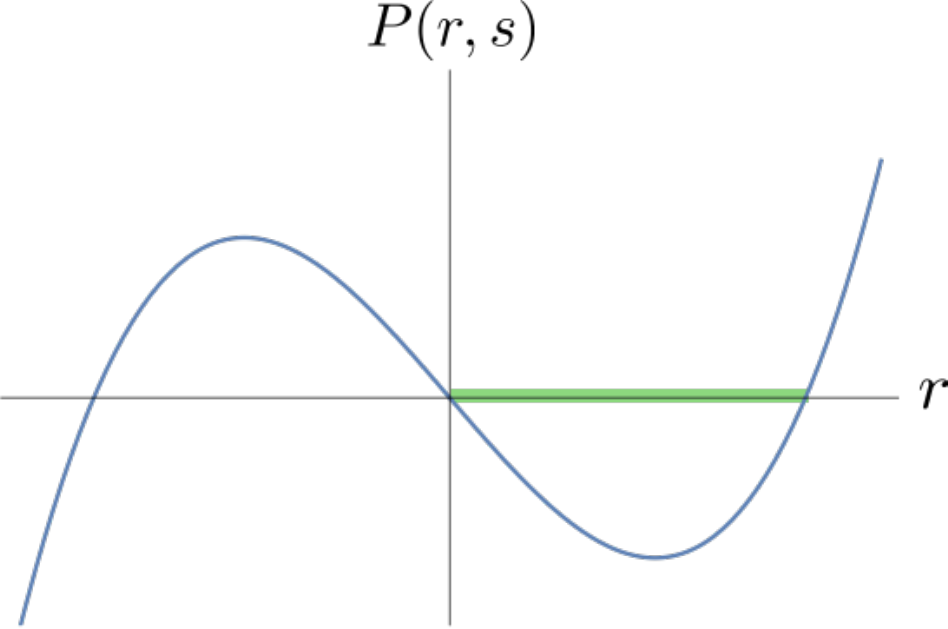}
\scriptsize{%
$\Lambda>0$, $Q=0$\\and $M=0$}
\end{minipage}\\[12pt]
\begin{minipage}[b]{0.24\textwidth}
\centering
\includegraphics[width=\textwidth]{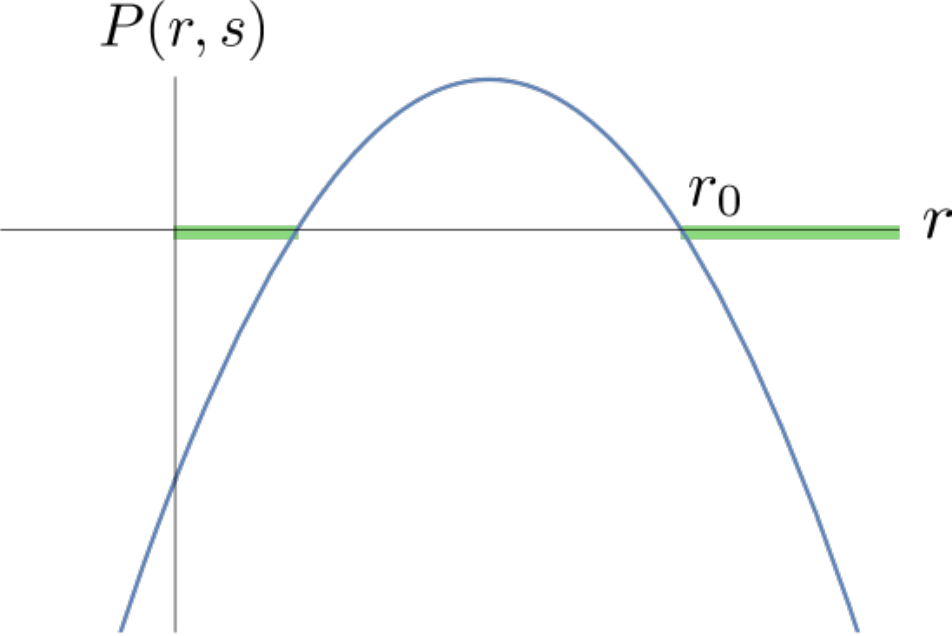}
\scriptsize{%
$\Lambda=0$, $0< Q^2<sM^2$\\and $M>0$}
\end{minipage}
\begin{minipage}[b]{0.24\textwidth}
\centering
\includegraphics[width=\textwidth]{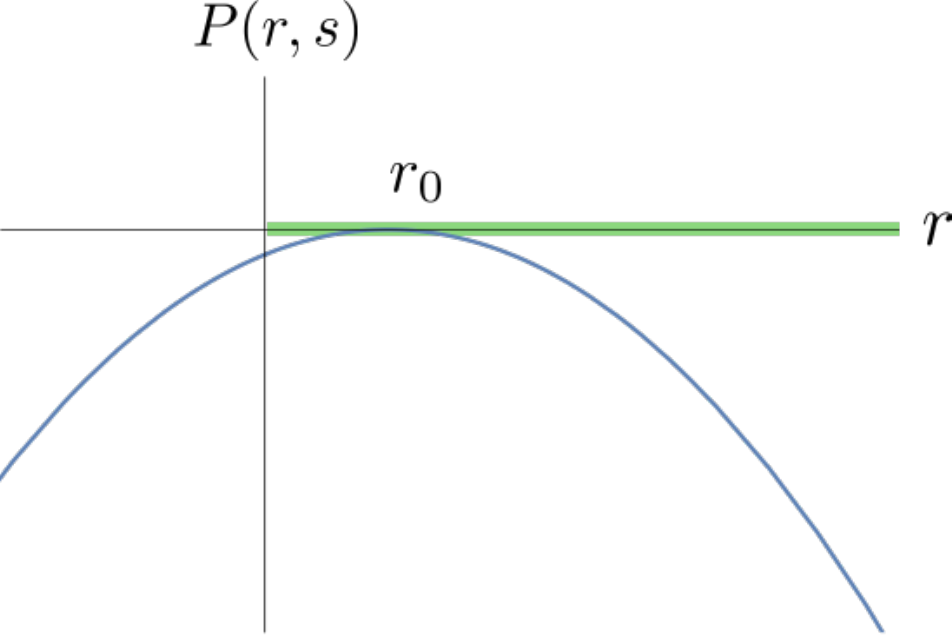}
\scriptsize{%
$\Lambda=0$, $Q^2=sM^2$\\and $M>0$}
\end{minipage}
\begin{minipage}[b]{0.24\textwidth}
\centering
\includegraphics[width=\textwidth]{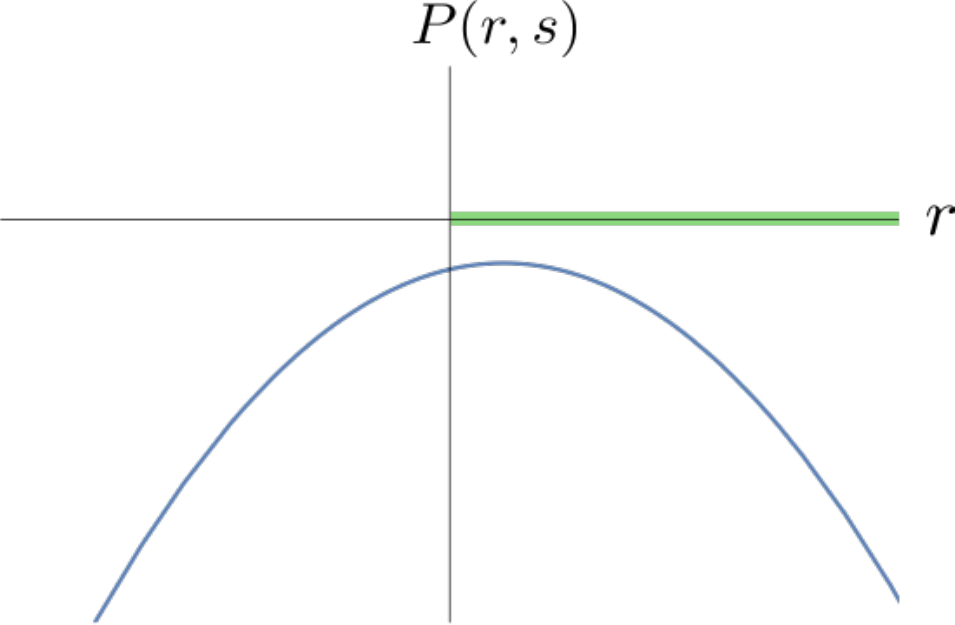}
\scriptsize{%
$\Lambda=0$, $Q^2>sM^2$\\and $M>0$}
\end{minipage}
\begin{minipage}[b]{0.24\textwidth}
\centering
\includegraphics[width=\textwidth]{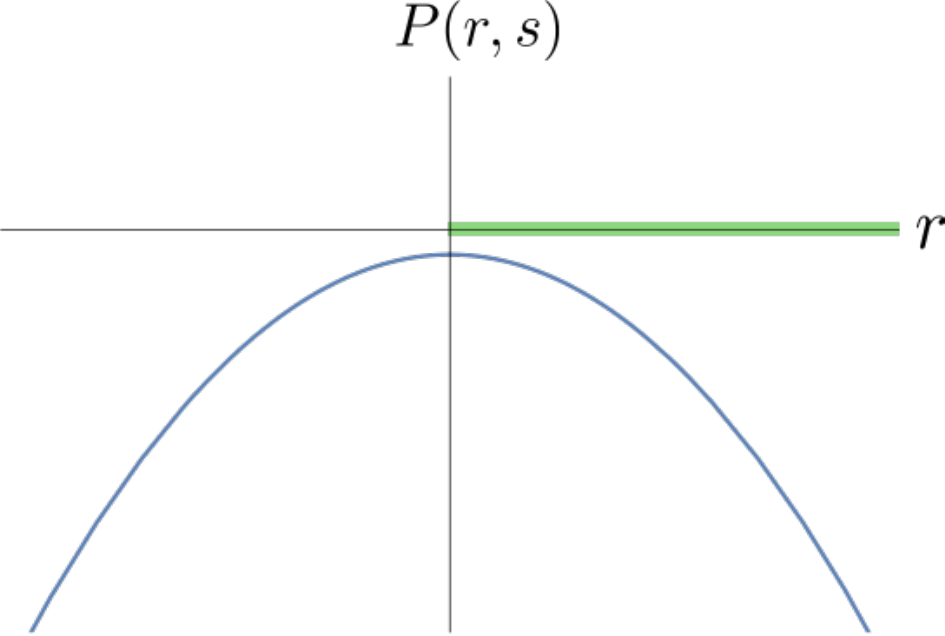}
\scriptsize{%
$\Lambda=0$, $Q\neq0$\\and $M=0$}
\end{minipage}\\[12pt]
\begin{minipage}[b]{0.24\textwidth}
\centering
\includegraphics[width=\textwidth]{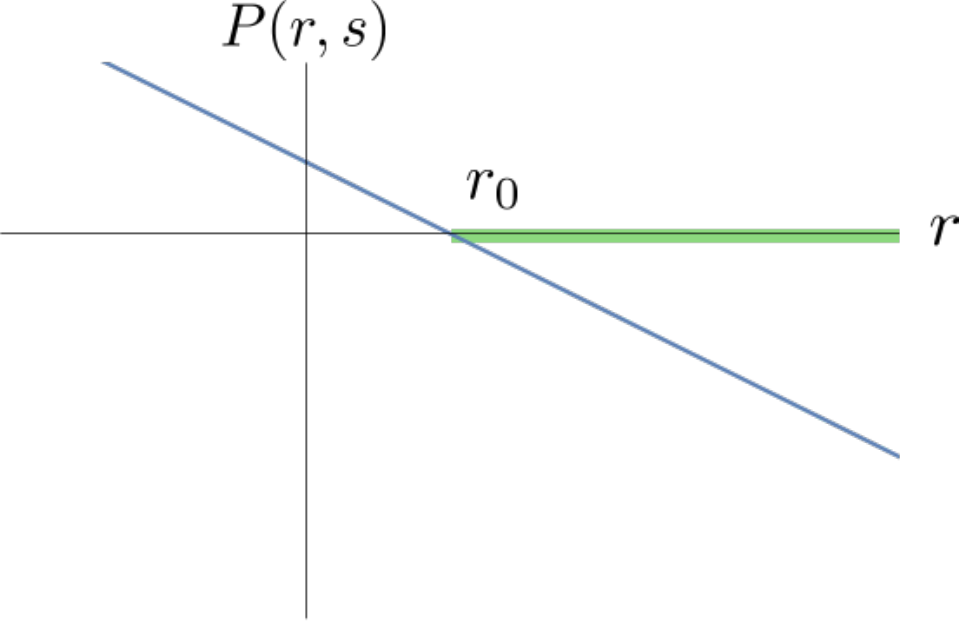}
\scriptsize{%
$\Lambda=0$, $Q=0$ and $M>0$}
\end{minipage}
\begin{minipage}[b]{0.24\textwidth}
\centering
\includegraphics[width=\textwidth]{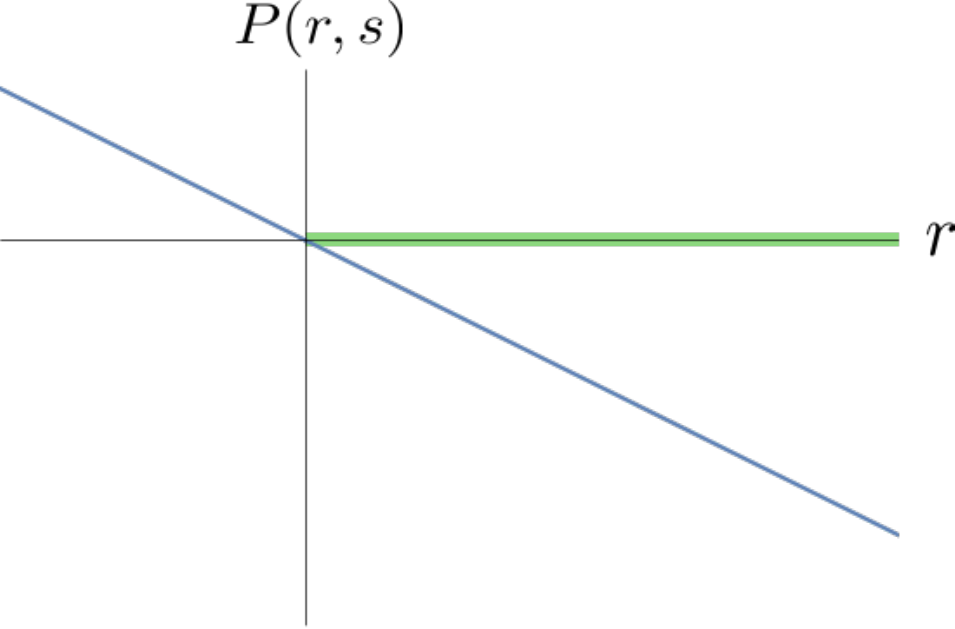}
\scriptsize{%
$\Lambda=0$, $Q=0$ and $M=0$}
\end{minipage}\begin{minipage}[b]{0.24\textwidth}
\centering
\includegraphics[width=\textwidth]{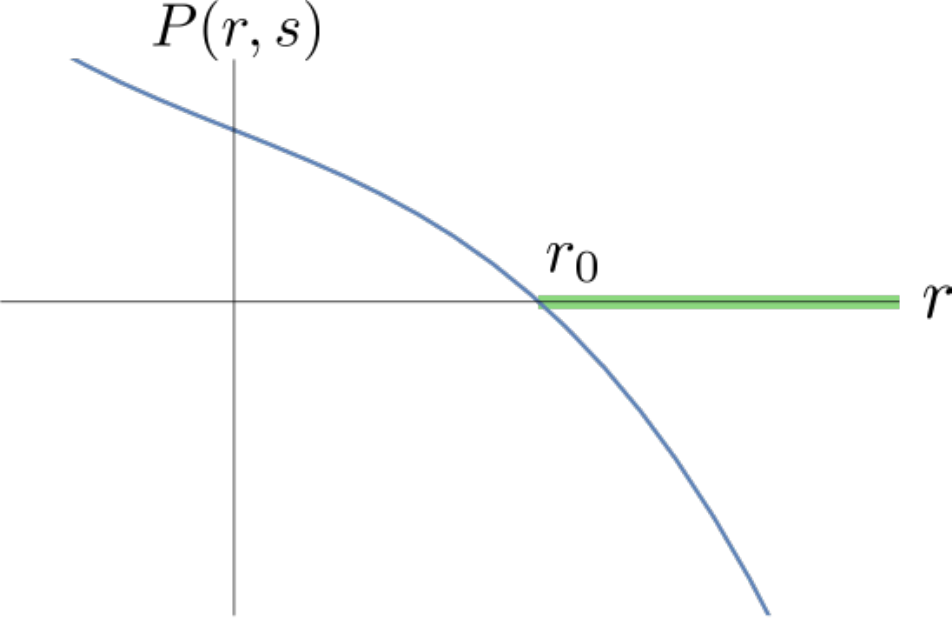}
\scriptsize{%
$\Lambda<0$, $Q=0$ and $M>0$}
\end{minipage}
\begin{minipage}[b]{0.24\textwidth}
\centering
\includegraphics[width=\textwidth]{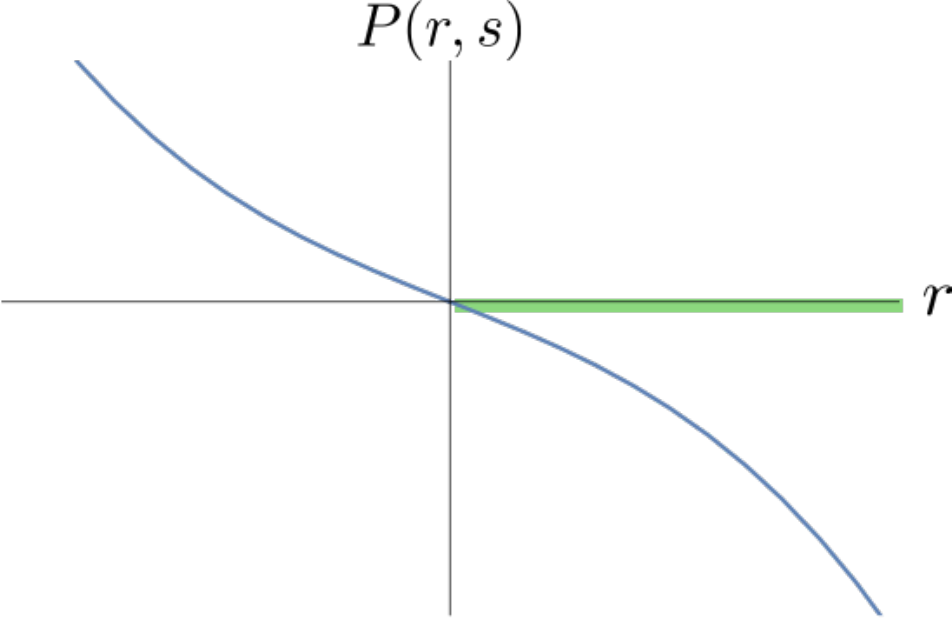}
\scriptsize{%
$\Lambda<0$, $Q=0$ and $M=0$}
\end{minipage}\\[12pt]
\vspace{12pt}
\begin{minipage}[b]{0.24\textwidth}
\centering
\includegraphics[width=\textwidth]{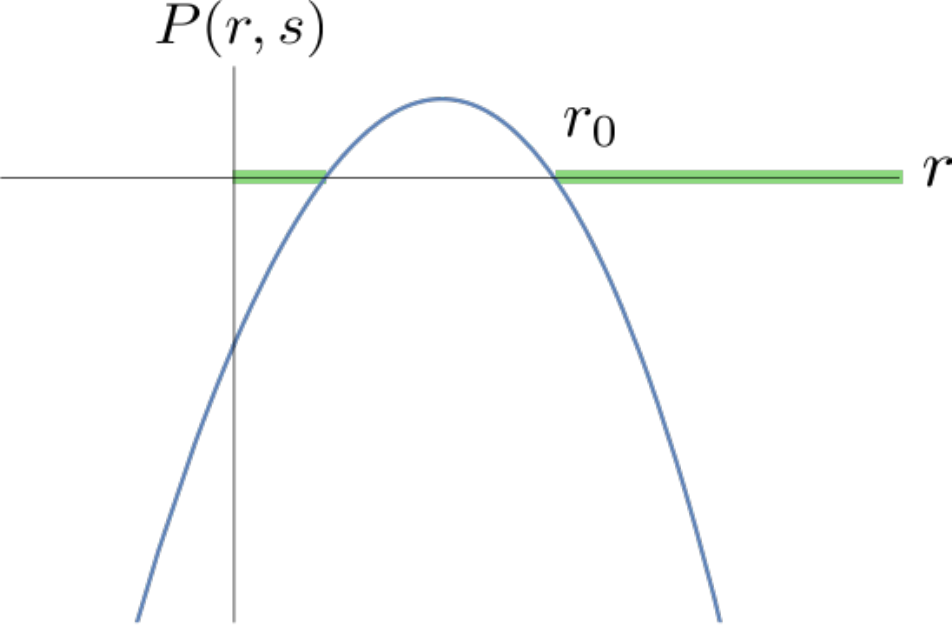}
\scriptsize{%
$\Lambda_-<\Lambda<0$, $Q\neq0$\\and $M>0$}
\end{minipage}
\begin{minipage}[b]{0.24\textwidth}
\centering
\includegraphics[width=\textwidth]{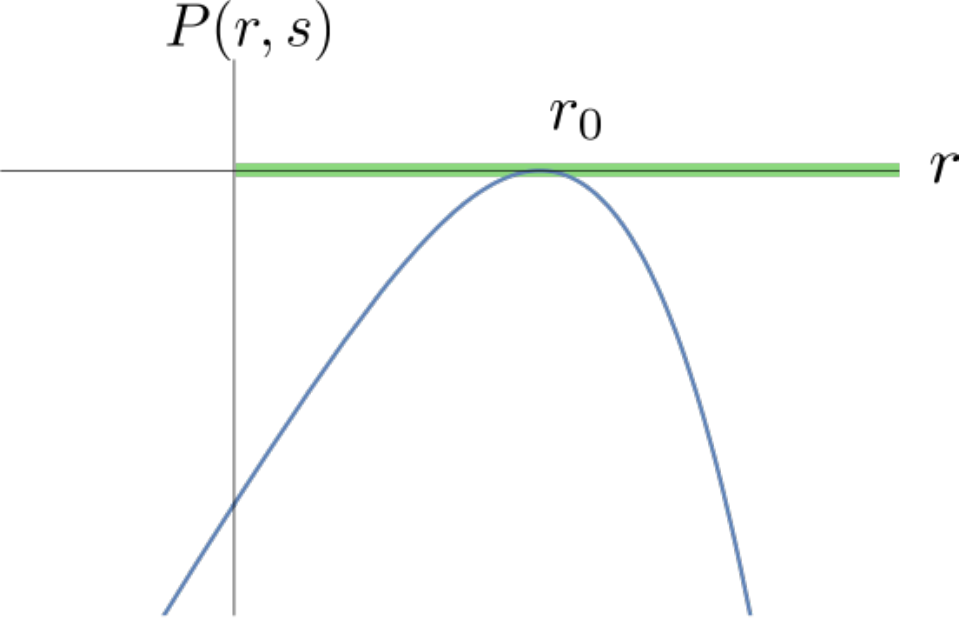}
\scriptsize{%
$\Lambda=\Lambda_-<0$, $Q\neq0$\\and $M>0$}
\end{minipage}
\begin{minipage}[b]{0.24\textwidth}
\centering
\includegraphics[width=\textwidth]{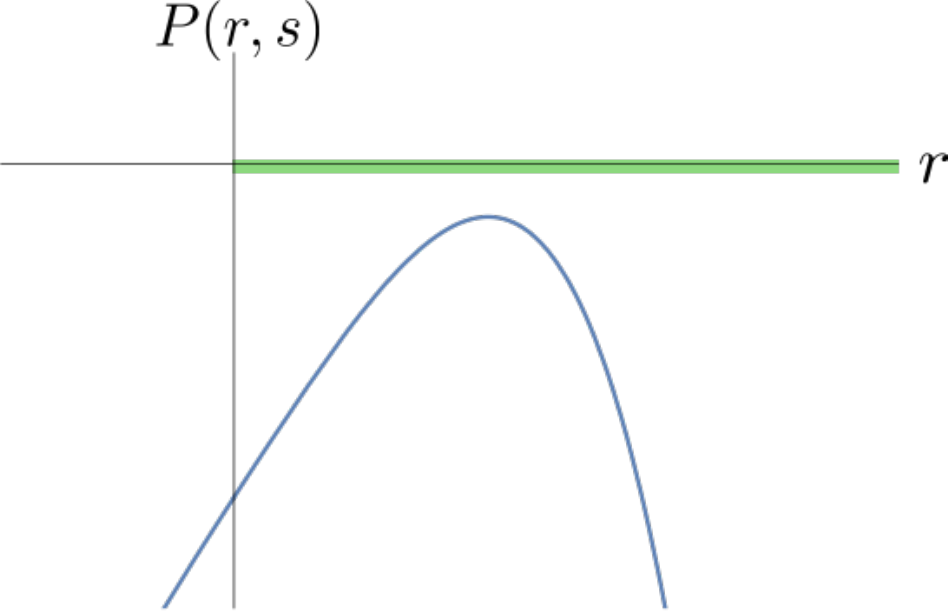}
\scriptsize{%
$\Lambda<\Lambda_-<0$, $Q\neq0$\\and $M>0$}
\end{minipage}
\begin{minipage}[b]{0.24\textwidth}
\centering
\includegraphics[width=\textwidth]{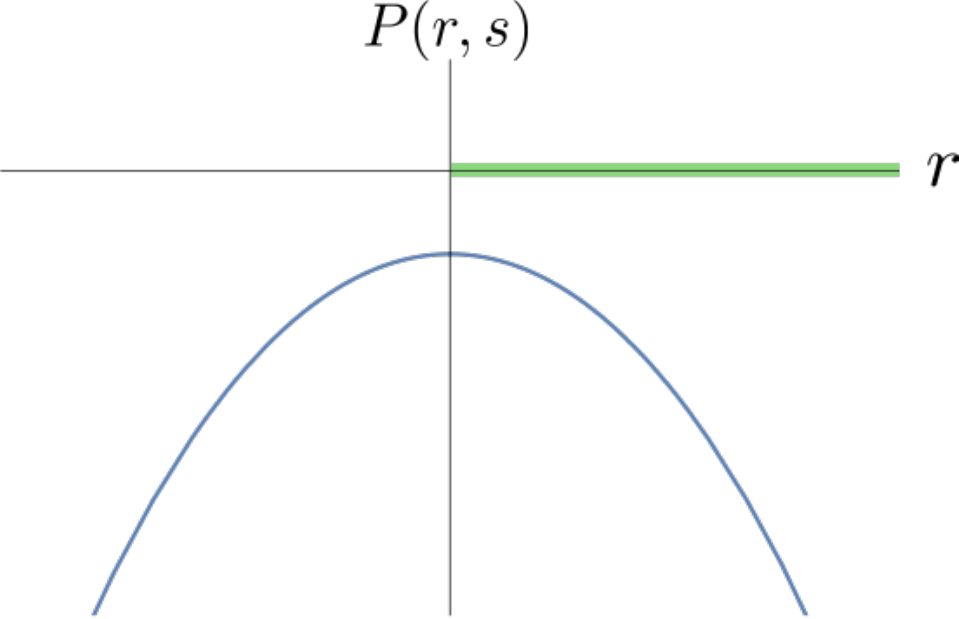}
\scriptsize{%
$\Lambda<0$, $Q\neq0$\\and $M=0$}
\end{minipage}
\caption{All possible cases for $P(r,s)$ with $M\geq0$. The allowed regions for $r$ where $P(r,s)<0$ are highlighted in green. We label the points $r_0$. Let us point out that the cases with $M<0$ are $P(r,s)|_{M=-A}=P(-r,s)|_{M=A}$ when $Q\neq0$ and $P(r,s)|_{M=-A}=-P(-r,s)|_{M=A}$ when $Q=0$, and thus no $r_0$ appears.}
\label{fig.proofs}
\end{figure}

It only remains to obtain the constraints in terms of $M, \Lambda$, and $Q$, imposed by
the conditions on $\Delta$ in cases 1 and 4,
and the existence of a third double root in case 1. Observe first that all
conditions on $\Delta$ are equivalent to the condition
$f(\Lambda):=\Delta/\Lambda\geq 0$.
We use next that the function $f(\Lambda)$
is a second-order polynomial in $\Lambda$
with $f''(\Lambda)<0$ for all $\Lambda$. Therefore, we have $f(\Lambda)\geq 0$ for
$\Lambda\in [\Lambda_-,\Lambda_+]$ with $f(\Lambda)=0$
  at the extremes $\Lambda= \Lambda_\pm$, if two distinct roots $\Lambda_\pm$ \eqref{eq:Lambdapm} exist.
  This happens when (and only when) $\beta:=9\lproof M^2-8Q^2>0$.
  In the degenerate case $\beta=0$, we have
$\Lambda_-=\Lambda_+$ and $f(\Lambda)=0$ only at $\Lambda=\Lambda_-=\Lambda_+$
and negative elsewhere.
On the other hand, since $f(0)=\lproof M^2-Q^2$,
then $f(0)\geq 0$ implies $\beta>0$.
Moreover, when $f(0)>0$, then  $\Lambda_-<0$ and $\Lambda_+>0$;
when $f(0)=0$, we have $\Lambda_-=0$ and $\Lambda_+>0$;
and when $f(0)<0$, the two roots $\Lambda_\pm$ (if they exist) are positive.

As a result, if $\Lambda>0$, we have $f(\Lambda)>0$ if and only if  $\beta>0$.
In such a case, if $f(0)\geq 0$ then $\Lambda\in (0,\Lambda_+)$;
and if $f(0)< 0$, then $\Lambda\in (\Lambda_-,\Lambda_+)$.
On the other hand, in the case $\Lambda>0$ with $f(\Lambda)=0$ and a third double root in $P(r,s)$, we need
that \eqref{one_double} holds, which reduces to just $1-4\lproof^2\Lambda_\pm Q^2>0$
because we are assuming $\Lambda>0$.
Since for $\beta\geq 0$
the following equalities hold
\[
  1-4\lproof^2\Lambda_\pm Q^2=\frac{\beta}{8 Q^4}
  \left(\beta+4Q^2\mp3\sqrt{\lproof M^2\beta}\right)
  =\frac{\beta}{8Q^4}
  \left(\beta+4Q^2\mp \sqrt{(\beta+4Q^2)^2-16 Q^{4}}\right),
\]
then we have at most one double root if and only if $\beta>0$
in both cases $\Lambda=\Lambda_-$
and $\Lambda=\Lambda_+$.
It only remains to check in which of the two cases the double root
is not $r_{1+}$. To do that, we use that a double root
must satisfy $P(a,\lproof)=0$ and $P'(a,\lproof)=0$. Using both equations
one obtains \footnote{First, we consider $4P(a,s)-aP'(a,s)=0$, solve for $a^2=3sMa-2sQ^2$, and substitute it (recursively) in $2P(a,s)-aP'(a,s)=0$. In that way, we obtain an equation linear in both $\Lambda$ and $a$. Solving for $1/\Lambda$, and substituting $Q^2/a$ by $3M/2-a/(2s)$ we obtain the desired relation.} the relation
  \[
    \frac{1}{\Lambda}-4\lproof^2 Q^2=\frac{2}{3}\lproof^2\beta\left(1-\frac{M\lproof}{a}\right)^{-1},
  \]
  which must hold for $\Lambda=\Lambda_+$ and $\Lambda=\Lambda_-$,
  providing $a_+$ and $a_-$ correspondingly.
  From this relation, and since $\Lambda_-<\Lambda_+$, we obtain
  $a_-<a_+$. Therefore, when $\Lambda=\Lambda_+$ the double root
  must correspond to $r_{1+}$ (this corresponds to
    the double root case in Remark \ref{remark_rinf}) %
  while $\Lambda=\Lambda_-$ [as long as $\Lambda_->0$,
  thus $f(0)<0$]
  produces the desired third (and double) root $r_{1d}<r_{1+}$.

On the other hand, if $\Lambda<0$ the only possibility is that $f(0)>0$, and
then $\Lambda\in(\Lambda_-,0)$ for $f(\Lambda)>0$ and $\Lambda=\Lambda_-$
in the extremal case $f(\Lambda)=0$.
\end{proof}

\subsection{Horizons}\label{app.proofs.hor}
We prove next the subclassification of the cases $C_1$, $D_1$, and $C_2$
(corresponding to cases $1$, $1D$, and $2$ above for $s=\lbar$)
into the BH, extremal and cosmos subcases,
depending on the existence and degeneracy of horizons $r_H$ and $r_C$.
The key to obtain the different cases is the fact that
$r_H(M,Q,\Lambda)=r_0(M,Q,\Lambda,1)$ and $r_C(M,Q,\Lambda)=\rmax(M,Q,\Lambda,1)$.
Let us define $\Lambda_\pm(s)$, as functions of $s$,
as given by \eqref{eq:Lambdapm}
for fixed values of $M$ and $\lmin$, with $\lmin:=8Q^2/(9M^2)$,
which is nonnegative by definition.
We will also need the following result.

\begin{lemma}\label{lemma:min_s_1}
  Consider $\lmin$ as defined above and assume $\lmin<1$. Restrict $s$ to $s>\lmin$.
  Then both $\Lambda_\pm(s)$ are monotonically decreasing functions.
\end{lemma}
\begin{proof}
  Let us use the sign $\sigma=\pm 1$
  to define
  $$
  f_\sigma(s):= \frac{9}{2}M^2 b^3 \Lambda_\sigma(s) +8
  $$
  on $s\in(\lmin,1]$.
  Clearly both $\Lambda_\sigma(s)$
  are  monotonically decreasing functions iff
  $f_\sigma(s)$ are. Using the substitution $8Q^2=9 M^2\lmin$
  on $\Lambda_\sigma(s)$, we explicitly obtain
  \[
    f_\sigma(s)=3\frac{\lmin}{s}\left(4-\frac{\lmin}{s}\right)
    +8 \sigma  \left(1-\frac{\lmin}{s}\right)^{3/2}.
\]
Then we have
\[
  f'_\sigma(s)=\frac{6\lmin}{s^3}
  \left(\lmin-2s+2\sigma(s-\lmin)\left(1-\frac{\lmin}{s}\right)^{-1/2}\right)=-\frac{6\lmin}{s^3}(\sqrt{s-\lmin}-\sigma\sqrt{s})^2,
\]
after using $s>\lmin$ in the second equality.
Therefore $f_\sigma'(s)<0$ in all the interval $(\lmin,1]$,
which proves the result.
  
\end{proof}

In the cases $C_1$ and $D_1$ the condition $8Q^2<9\lbar M^2$,
this is $\lmin<\lbar$, implies that
$8Q^2<9 M^2$, that is $\lmin<1$.  Therefore, the application of Lemma \ref{res:lemma1}
and  Remark \ref{remarkbeta+} for $s=1$
ensures that $\Lambda_+(1)$ exists and it is positive,
and that $r_H=r_0(M,Q,\Lambda,1)$ exists and it is a simple root if
(i) $\Lambda\in (\Lambda_-(1),\Lambda_+(1))\cup (0,\Lambda_+(1))$,
it is a double root if (ii) $\Lambda=\Lambda_-(1)$
or (iii) $\Lambda=\Lambda_+(1)$ (by Remark \ref{remark_rinf}),
and it does not exist otherwise.
Different cases will arise now depending on the existence
of intersections of the above values and ranges for $\Lambda$ and those
characterizing each case ($C_1$ and $D_1$).

In case $C_1$ we need to analyze when
the intersection of $(\Lambda_-(\lbar),\Lambda_+(\lbar))\cap (0,\Lambda_+(\lbar))$
with the above three, that is (i) $(\Lambda_-(1),\Lambda_+(1))\cap (0,\Lambda_+(1))
\cap  (\Lambda_-(\lbar),\Lambda_+(\lbar))\cap (0,\Lambda_+(\lbar))$,
(ii) $\{\Lambda_-(1)\}\cap  (\Lambda_-(\lbar),\Lambda_+(\lbar))\cap (0,\Lambda_+(\lbar))$,
and (iii) $\{\Lambda_+(1)\}\cap  (\Lambda_-(\lbar),\Lambda_+(\lbar))\cap (0,\Lambda_+(\lbar))$
are not empty. Since
$\Lambda_\pm(1)<\Lambda_\pm(\lbar)$, ensured by
Lemma \ref{lemma:min_s_1}, and $\Lambda_\pm(1)>0$, those intersections
equal (i) $(\Lambda_-(\lbar),\Lambda_+(1))\cap (0,\Lambda_+(1))$,
(ii) $\emptyset$, and (iii)  $\{\Lambda_+(1)\}\cap  (\Lambda_-(\lbar),\Lambda_+(\lbar))$,
respectively.
The emptiness of intersection (ii) accounts to the fact
that if $r_H$ degenerates, it cannot do so by merging with the lowest
positive root $r_I$, because, as shown in the text, $r_I$ is not in the
range of values of $r$ in the singularity-free cases.
In the cases (i) and (iii)  the intersection is not empty
iff $I(\lbar):=\Lambda_+(1)-\Lambda_-(\lbar)>0$.
By Lemma \ref{lemma:min_s_1}, $I(s)$ is monotonically increasing for $s>\lmin$,
and we have  $I(1)>0$ and $I(9b/8)>0$, and  it is straightforward to check
that $I(\lmin)<0$ for all $M>0$ and $\lmin\in(0,1)$.
As a result, there is a value $h(M,Q)$, and only one, such that $I(h)=0$,
and $r_H$ exists and it is a simple root of $P(r,s)$
iff $\lbar>h$ and $\Lambda\in(\Lambda_-(\lbar),\Lambda_+(1))\cap (0,\Lambda_+(1))$,
while $r_H$ is a double root (and it is the largest by Remark \ref{remark_rinf})
iff $\lbar>h$ and $\Lambda=\Lambda_+(1)$, and it does not exist otherwise.

In case $D_1$, for which $\Lambda=\Lambda_-(\lbar)>0$
[positivity of $\Lambda_-$
  needs $\lbar M^2<Q^2$ and thus $\lbar\in( \lmin,9\lmin/8)\cap (0,1)$],
the relevant intersections are thus
(i) $(\Lambda_-(1),\Lambda_+(1))\cap (0,\Lambda_+(1))\cap \{\Lambda_-(\lbar)\}$,
(ii) $\{\Lambda_-(1)\}\cap\{  \Lambda_-(\lbar)\}$,
and (iii) $\{\Lambda_+(1)\}\cap \{\Lambda_-(\lbar)\}$. Clearly (ii) is empty,
while the other two equal $\{\Lambda_-(\lbar)\}$ and are not empty (i) iff $I(\lbar)>0$
or (iii) iff $\Lambda_+(1)= \Lambda_-(\lbar)$, i.e.
$I(\lbar)=0$. In consequence,
$r_H$ exists in $D_1$, and it is a simple root iff $\lbar>h$,
it is a double root iff $\lbar=h$ (and it is the largest root),
and it does not exist otherwise. Observe that $I(9\lmin/8)>0$ ensures
that we can indeed have the three cases.

In case $C_2$ let us recall that the characterizing conditions are $\Lambda>0$,
$Q=0$ and $3\sqrt{\Lambda}M\in(0,\lbar^{-3/2})$.
Now, the application of Lemma \ref{res:lemma1} for $s=1$
ensures that $r_H=r_0(M,0,\Lambda,1)$ is a simple root of $P(r,1)$
iff $3\sqrt{\Lambda}M\in(0,1)$. Since
the possible range of $3\sqrt{\lambda}M$ can be divided as
$(0,\lbar^{-3/2})=(0,1)\cup [1,\lbar^{-3/2})$, we still have to control what happens when
$3\sqrt{\Lambda}M\in  [1,\lbar^{-3/2})$.
From the second part of Remark \ref{remark_rinf}, we obtain
that $r_H=r_0(M,0,\Lambda,1)$ is a double root when
$\Lambda=\Lambda_+|_{Q\to 0}=(9\lbar^3M^2)^{-1}$ and thus
$3\sqrt{\Lambda}M=1$, while no $r_H$ exists in the remaining interval $(1,\lbar^{-3/2})$.

This finishes the subclassification of the cases $C_1$,  $D_1$ and $C_2$
in terms of the horizons as given in the text.

\section{Near-horizon geometries as limits}\label{app.nhg_limits}
In this section we prove that, as in GR, the near-horizon geometries
correspond to one of the two possible limits
of the family of solutions with parameters $(M,Q,\Lambda,\lbar)$
to the values for which the horizon degenerates
(see Ref. \cite{Bengtsson2022} for a nice and pedagogical treatment).

We provide first the explicit system of equations that
characterize the extremal cases to check that those coincide,
as expected, with the equations that characterize the near-horizon geometries.
The system of equations needed for a double root of $G(r)$ in $r=r_d$ can be expressed as
$\{P(r_d,1)=0, P'(r_d,1)=0\}$.
It is direct to check that this system is equivalent to
\begin{equation}
  M=r_d\left(1-\frac{2}{3}\Lambda r_d^2\right),\quad \Lambda r_d^4-r_d^2+Q^2=0.
\end{equation}
These correspond, for $r_d=a$, to the two equations for
the constant $r=a$, together with $M$, $Q$ and $\Lambda$, that characterize
the near-horizon geometries,
\eqref{eq:a_main} and \eqref{eq.M.nhg}.

Next we show explicitly how the near-horizon geometries
of the extremal black holes
can be obtained by a limiting procedure from the family  $(M,Q,\Lambda,\lbar)$.

We start with the coordinates in which the metric reads
\eqref{metric:tau_x}, and perform the change $(\tau,z)\to (\zeta,Y)$,
given by
\[
  z=z_a+\varepsilon Y, \quad \tau=\frac{1}{\varepsilon}\zeta,
\]
for some parameter $\varepsilon$, and where $z_a$ satisfies $r(z_a)=a$,
with $a$ obeying the relation \eqref{eq:a_main}.
Expanding $r(z(Y))$ in $\varepsilon$, and using \eqref{eq.rp},
we obtain
\[
  r(z(Y))=r(z_a)+r'(z_a)(z-z_a)+O((z-z_a)^2)=a+\sigma\sqrt{1-\lbar} Y\varepsilon + O(\varepsilon^2),
\]
with $\sigma^2=1$. The procedure now entails
introducing the explicit form of $m$, as given by \eqref{eq.m(x)}
replacing $x$ by $r(z(Y))$, together with \eqref{eq.M.nhg}
and using \eqref{eq:a}, into \eqref{metric:tau_x}.
It is then straightforward to take the limit  $\varepsilon\to 0$
to obtain \eqref{app.met.nhgY} with $c_2=0$.

\section{Global properties around the horizons and
  critical values of $r$
}\label{app.tortoise}

With the aim at obtaining the building blocks to construct the Penrose diagrams,
in this Appendix we analyze the properties of the solution around
the horizons [zeros of $G(r)$] and the critical values of $r$ [zeros of $V(r)$].
To do that,
since the coordinates $(\tau,z)$ cover the horizons and the critical hypersurfaces
they may exist, 
we proceed to find suitable
null coordinates based on the radial null geodesic equations
in terms of $(\tau,z)$, that is, imposing $\gamma=0$ in \eqref{zdot_E}.
Observe that the analysis presented in this Appendix thus stands for all possible cases,
not only the singularity-free family of solutions.

For each value of $\signgeod$
we take one geodesic vector (up to orientation),
which we name
$k$ ($\signgeod=1$) and $l$ ($\signgeod=-1$), given explicitly by
\[
  l=\left(1+\sqrt{\frac{2m(r)}{r}}\right)^{-1}\partial_\tau-\partial_z,\qquad
  k=\left(1-\sqrt{\frac{2m(r)}{r}}\right)^{-1}\partial_\tau+\partial_z,
\]
at points outside the horizons.
The affine parametrization implies $d\mathbf{l}=0$ and $d\mathbf{k}=0$,
where we use the boldface to denote one-forms so that
$\mathbf{l}=l_\mu dx^\mu$ and $\mathbf{k}=k_\mu dx^\mu$.
We can thus define a pair of functions $(U,V)$ by
\begin{align}
  dU=-\mathbf{l}=d\tau+\left(1+\sqrt{\frac{2m(r(z))}{r(z)}}\right)^{-1}dz,\qquad
  dV=-\mathbf{k}=d\tau-\left(1-\sqrt{\frac{2m(r(z))}{r(z)}}\right)^{-1}dz,\label{eq.UV}
\end{align}
on any region in the domain of $(\tau,z)$ outside the horizons,
and perform, still, a further change $dU=A(u) du$ and $dV=B(v) dv$,
with arbitrary (smooth) and nowhere vanishing functions $A(u)$ and $B(v)$.
This second change usually amounts to obtaining the coordinates
$(u,v)$ of the Kruskal type, and with compact ranges if needed.
Note that the determinant of the Jacobian of the change
from $(\tau, z)$ to $(U,V)$
is $G(r)/2$. Therefore the change %
  to $(u,v)$ is well defined outside the horizons. Observe
    that the near-horizon geometries are necessarily left out from this treatment.
The metric in terms of the new null coordinates $(u,v)$ reads
\[
  ds^2=-%
  G(r) A(u) B(v) dudv+r^2d\Omega^2,
\]
where $r$ should be written in terms of the new variables $u$ and $v$.
The purpose now is to obtain that relation $r(u,v)$ near the horizons
and critical values of $r$,
that is, near the roots of $G(r)$ and $V(r)$.

To link the argument with the usual procedure
used in GR (see, e.g., Refs. \cite{Brill1994, notes_jose})
and also in \cite{Alonso-Bardaji:2022ear},
we start by defining the so-called tortoise function $r_*$
by $r_*:=\signe (r')(U-V)/2$, which, given \eqref{eq.UV}, is obtained from
\begin{equation}
  r_*(r)=\int \frac{1}{\sqrt{-2V(r)}}\frac{1}{\chif(r)}dr,
  \label{eq.r*}
\end{equation}
up to some additive constant.
The crucial point is that $r_*$ is a function of $r$ only,
and therefore a convenient choice of functions $A(u)$ and $B(v)$
can be used to obtain a relation of the form
\begin{equation}
  e^{2r_*(r(u,v))}=e^{\signe(r')(U(u)-V(v))}=e^{\signe(r')U(u)} e^{-\signe(r')V(v)}=(uv)^{2C},
  \label{eq:rstar_uv}
\end{equation}
where the first equality follows from the definition of $r_*$,
for any %
chosen constant $C$.
This relation
provides the form of the hypersurfaces of constant $r$ in terms of the pair $(u,v)$.
Since we are interested in regions around  the zeros of $G(r)=-r^{-\ell}P(r,1)$
and $2V(r)=r^{-\ell}P(r,\lbar)$ (recall, with $\ell=1,2$ for
$Q=0$ and $Q\neq 0$, respectively),
we write $D(r):=-\sqrt{-P(r,\lbar)}P(r,1)$, so that \eqref{eq.r*} takes the form
\begin{equation}
  r_*(r)=\int \frac{r^{3n/2}}{D(r)}dr.
  \label{eq.r*_H}
\end{equation}
Now one can look for the expression of this integral as an
expansion around the zeros of the denominator $D(r)$ for the different cases.
Let us take some positive constant $a>0$, such that $P(a,\lbar)\leq 0$.
Using the fact that the zeros of $P(r,1)$ and $P(r,\lbar)$ cannot coincide, cf. \ref{id.P},
the cases of study are then the following:
\begin{enumerate}
\item[a)] $r=a$ is a simple zero of $V(r)$,
which corresponds to a critical hypersurface $r=r_0$ or $r=\rmax$.
  Then we can use a Taylor expansion of $r^{3n/2}\sqrt{|r-a|}/D(r)$ around $r=a$
  (either at one side $r>a$, or the other $r<a$), so that
  $r^{3n/2}\sqrt{|r-a|}/D(r)=\sum_{i=0}^\infty{C_i (r-a)^i}$ with $C_0\neq 0$.
  In this way,
\[
  r_*=\int \frac{dr}{\sqrt{|r-a|}}\left(\sum_{i=0}^\infty{C_i (r-a)^i}\right)=\signe(r-a)\sqrt{|r-a|}h(r),
\]
where $h(r)$ is a function expandable around $r=a$ and $h(a)=2C_0$.
Then \eqref{eq:rstar_uv}, with the choice $C=1$, reads
\[
e^{\signe(r-a)\sqrt{|r-a|}h(r)}=uv.
\]
The hypersurface $r(u,v)=a$ in $(u,v)$ is thus given by the curve $uv=1$.
If we are expanding around $r=a$, with $r>a$, then that corresponds to $uv>1$,
while in the case $r<a$ the region is located where $uv<1$.

\item[b)] $r=a$ is either a double zero of $V(r)$ or a simple zero of $G(r)$,
which corresponds to a degenerate critical hypersurface $r=r_0=R$ or to a nondegenerate
horizon $r=r_H$ or $r=r_C$, respectively.
  Then around $r=a$ (at either side $r>a$ or $r<a$)
  we can write $r^{3n/2}(r-a)/D(r)=\sum_{i=0}^\infty{C_i (r-a)^i}$ with $C_0\neq 0$. Thus
\[
  r_*=\int \frac{dr}{r-a}\left(\sum_{i=0}^\infty{C_i (r-a)^i}\right)
  =C_0\log|r-a|+h(r),
\]
where $h(r)$ is a function expandable around $r=a$,
so that \eqref{eq:rstar_uv} takes the form
\[
  |r-a|e^{h(r)/C_0}=uv,
\]
around $r=a$, after choosing $C=C_0$.
Therefore, $r=a$ corresponds to  $uv=0$.

\item[c)] $r=a$ is a double zero of $G(r)$, which corresponds to
the degenerate horizon $r=r_H=r_C$. %
  Then around $r=a$
  we have $r^{3n/2}(r-a)^2/D(r)=\sum_{i=0}^\infty{C_i (r-a)^i}$ with $C_0\neq 0$, and therefore
\[
  r_*=\int \frac{dr}{(r-a)^2}\left(\sum_{i=0}^\infty{C_i (r-a)^i}\right)
  =-C_0\frac{1}{r-a}+C_1 \log|r-a|+h(r),
\]
where $h(r)$ is a function expandable around $r=a$.
Then, depending on the side we are interested in,
$r>a$ or $r<a$, we choose $C=\signe(r-a)C_0$ so that  \eqref{eq:rstar_uv} reads 
\[
 \exp{\left(\frac{-1}{|r-a|}\right)}|r-a|^{\signe(r-a)C_1/C_0}e^{\signe(r-a)h(r)/C_0}=uv.
\]
As a result $r=a$ corresponds in both cases to  $uv=0$.
\end{enumerate}

\bibliography{refs}

\providecommand{\noopsort}[1]{}\providecommand{\singleletter}[1]{#1}%
\begin{thebibliography}{10}

\bibitem{Ashtekar:2011ni}
A.~Ashtekar and P.~Singh, ``{Loop Quantum Cosmology: A Status Report},'' {\em
  Class. Quant. Grav.}, vol.~28, p.~213001, 2011.

\bibitem{Agullo:2016tjh}
I.~Agullo and P.~Singh, {\em {Loop Quantum Cosmology}}, pp.~183--240.
\newblock WSP, 2017.

\bibitem{Rovelli:2013zaa}
C.~Rovelli and E.~Wilson-Ewing, ``{Why are the effective equations of loop
  quantum cosmology so accurate?},'' {\em Phys. Rev. D}, vol.~90, no.~2,
  p.~023538, 2014.

\bibitem{BenAchour:2018khr}
J.~Ben~Achour, F.~Lamy, H.~Liu, and K.~Noui, ``{Polymer Schwarzschild black
  hole: An effective metric},'' {\em EPL}, vol.~123, no.~2, p.~20006, 2018.

\bibitem{Gambini:2020nsf}
R.~Gambini, J.~Olmedo, and J.~Pullin, ``{Spherically symmetric loop quantum
  gravity: analysis of improved dynamics},'' {\em Class. Quant. Grav.},
  vol.~37, no.~20, p.~205012, 2020.

\bibitem{Kelly:2020uwj}
J.~G. Kelly, R.~Santacruz, and E.~Wilson-Ewing, ``{Effective loop quantum
  gravity framework for vacuum spherically symmetric spacetimes},'' {\em Phys.
  Rev. D}, vol.~102, no.~10, p.~106024, 2020.

\bibitem{Ashtekar:2018lag}
A.~Ashtekar, J.~Olmedo, and P.~Singh, ``{Quantum Transfiguration of Kruskal
  Black Holes},'' {\em Phys. Rev. Lett.}, vol.~121, no.~24, p.~241301, 2018.

\bibitem{Ashtekar:2018cay}
A.~Ashtekar, J.~Olmedo, and P.~Singh, ``{Quantum extension of the Kruskal
  spacetime},'' {\em Phys. Rev. D}, vol.~98, no.~12, p.~126003, 2018.

\bibitem{Ashtekar:2020ckv}
A.~Ashtekar and J.~Olmedo, ``{Properties of a recent quantum extension of the
  Kruskal geometry},'' {\em Int. J. Mod. Phys. D}, vol.~29, no.~10, p.~2050076,
  2020.

\bibitem{Bodendorfer:2019cyv}
N.~Bodendorfer, F.~M. Mele, and J.~M\"unch, ``{Effective Quantum Extended
  Spacetime of Polymer Schwarzschild Black Hole},'' {\em Class. Quant. Grav.},
  vol.~36, no.~19, p.~195015, 2019.

\bibitem{Bodendorfer:2019nvy}
N.~Bodendorfer, F.~M. Mele, and J.~M\"unch, ``{(b,v)-type variables for black
  to white hole transitions in effective loop quantum gravity},'' {\em Phys.
  Lett. B}, vol.~819, p.~136390, 2021.

\bibitem{Modesto:2008im}
L.~Modesto, ``{Semiclassical loop quantum black hole},'' {\em Int. J. Theor.
  Phys.}, vol.~49, pp.~1649--1683, 2010.

\bibitem{Hossenfelder:2009fc}
S.~Hossenfelder, L.~Modesto, and I.~Premont-Schwarz, ``{A Model for
  non-singular black hole collapse and evaporation},'' {\em Phys. Rev. D},
  vol.~81, p.~044036, 2010.

\bibitem{Kelly:2020lec}
J.~G. Kelly, R.~Santacruz, and E.~Wilson-Ewing, ``{Black hole collapse and
  bounce in effective loop quantum gravity},'' {\em Class. Quant. Grav.},
  vol.~38, no.~4, p.~04LT01, 2021.

\bibitem{Gambini:2021uzf}
R.~Gambini, F.~Ben\'\i{}tez, and J.~Pullin, ``{A Covariant Polymerized Scalar
  Field in Semi-Classical Loop Quantum Gravity},'' {\em Universe}, vol.~8,
  no.~10, p.~526, 2022.

\bibitem{Husain:2021ojz}
V.~Husain, J.~G. Kelly, R.~Santacruz, and E.~Wilson-Ewing, ``{Quantum Gravity
  of Dust Collapse: Shock Waves from Black Holes},'' {\em Phys. Rev. Lett.},
  vol.~128, no.~12, p.~121301, 2022.

\bibitem{Husain:2022gwp}
V.~Husain, J.~G. Kelly, R.~Santacruz, and E.~Wilson-Ewing, ``{Fate of quantum
  black holes},'' {\em Phys. Rev. D}, vol.~106, no.~2, p.~024014, 2022.

\bibitem{Tibrewala:2012xb}
R.~Tibrewala, ``{Spherically symmetric Einstein-Maxwell theory and loop quantum
  gravity corrections},'' {\em Class. Quant. Grav.}, vol.~29, p.~235012, 2012.

\bibitem{Gambini:2014qta}
R.~Gambini, E.~M. Capurro, and J.~Pullin, ``{Quantum spacetime of a charged
  black hole},'' {\em Phys. Rev. D}, vol.~91, no.~8, p.~084006, 2015.

\bibitem{Alonso-Bardaji:2021yls}
A.~Alonso-Bardaji, D.~Brizuela, and R.~Vera, ``{An effective model for the
  quantum Schwarzschild black hole},'' {\em Phys. Lett. B}, vol.~829,
  p.~137075, 2022.

\bibitem{Alonso-Bardaji:2022ear}
A.~Alonso-Bardaji, D.~Brizuela, and R.~Vera, ``{Nonsingular spherically
  symmetric black-hole model with holonomy corrections},'' {\em Phys. Rev. D},
  vol.~106, no.~2, p.~024035, 2022.

\bibitem{Alonso-Bardaji:2020rxb}
A.~Alonso-Bardaji{} and D.~Brizuela, ``{Holonomy and inverse-triad corrections
  in spherical models coupled to matter},'' {\em Eur. Phys. J. C}, vol.~81,
  no.~4, p.~283, 2021.

\bibitem{Alonso-Bardaji:2021tvy}
A.~Alonso-Bardaji and D.~Brizuela, ``{Anomaly-free deformations of spherical
  general relativity coupled to matter},'' {\em Phys. Rev. D}, vol.~104, no.~8,
  p.~084064, 2021.

\bibitem{Bojowald:2015zha}
M.~Bojowald, S.~Brahma, and J.~D. Reyes, ``{Covariance in models of loop
  quantum gravity: Spherical symmetry},'' {\em Phys. Rev. D}, vol.~92, no.~4,
  p.~045043, 2015.

\bibitem{Bojowald:2018xxu}
M.~Bojowald, S.~Brahma, and D.-h. Yeom, ``{Effective line elements and
  black-hole models in canonical loop quantum gravity},'' {\em Phys. Rev. D},
  vol.~98, no.~4, p.~046015, 2018.

\bibitem{Teitelboim73}
C.~Teitelboim, ``How commutators of constraints reflect the spacetime
  structure,'' {\em Annals of Physics}, vol.~79, pp.~542--557, 1973.

\bibitem{Pons:1996av}
J.~M. Pons, D.~C. Salisbury, and L.~C. Shepley, ``{Gauge transformations in the
  Lagrangian and Hamiltonian formalisms of generally covariant theories},''
  {\em Phys. Rev. D}, vol.~55, pp.~658--668, 1997.

\bibitem{Brill1994}
D.~R. Brill and S.~A. Hayward, ``Global structure of a black hole cosmos and
  its extremes,'' {\em Classical and Quantum Gravity}, vol.~11, p.~359, feb
  1994.

\bibitem{Bengtsson2022}
I.~Bengtsson, S.~Holst, and E.~Jakobsson, ``Classics illustrated: limits of
  spacetimes,'' {\em Class. Quant. Grav.} vol.~31, p.~205008, 2014.

\bibitem{ONeill1983}
B.~O'Neill, {\em Semi-Riemannian Geometry}.
\newblock Academic Press, 1983.

\bibitem{notes_jose}
J.~M.~M. Senovilla, {\em Lecture notes on conformal diagrams, UPV/EHU}, 2010.

\bibitem{Rees1922}
E.~L. Rees, ``Graphical discussion of the roots of a quartic equation,'' {\em
  The American Mathematical Monthly}, vol.~29, pp.~51--55, feb 1922.

\end{thebibliography}

\end{document}